\documentclass[reprint,
nofootinbib,
superscriptaddress,
 amsmath,amssymb,amsfonts,
 aps,
pra,
longbibliography,
]{revtex4-2}

\usepackage{xurl}
\usepackage{url}
\usepackage[normalem]{ulem}
\usepackage{svg}
\usepackage{xr}

\usepackage[english]{babel}

\makeatletter
\renewcommand\onecolumngrid{%
\do@columngrid{one}{\@ne}%
\def\set@footnotewidth{\onecolumngrid}%
\def\footnoterule{\kern-6pt\hrule width 1.5in\kern6pt}%
}

\renewcommand\twocolumngrid{%
        \def\footnoterule{%
        \dimen@\skip\footins\divide\dimen@\thr@@
        \kern-\dimen@\hrule width.5in\kern\dimen@}
        \do@columngrid{mlt}{\tw@}
}%
\makeatother   

\makeatletter
\newcommand{\Footnotemark}{\footnotemark%
 \expandafter\global\expandafter\let\csname saved@Href@A\endcsname%
    \Hy@footnote@currentHref}
\newcommand{\Footnotetext}[1]{%
   \expandafter\let\expandafter\Hy@footnote@currentHref\csname saved@Href@A\endcsname%
   \footnotetext{#1}}
\makeatother

\usepackage{graphicx}
\usepackage{dcolumn}
\usepackage{tabularx}
\usepackage{xcolor}
\usepackage{bm}

\usepackage{physics}
\usepackage{amsmath,amsfonts,amssymb}
\usepackage{amsthm}
\usepackage{mathtools}
\usepackage{graphicx}
\usepackage[margin=0.85in]{geometry}
\usepackage{tablefootnote}
\usepackage{subcaption}
\usepackage[linesnumbered,ruled,vlined]{algorithm2e}
\SetKwInput{KwInput}{Input}                
\SetKwInput{KwOutput}{Output}
\SetKwFor{RepTimes}{for}{do}{end}

\usepackage{hyperref}
\hypersetup{
    colorlinks=true, %
    linkcolor=blue, %
    citecolor=blue, %
    urlcolor=blue %
}
\usepackage{microtype}
\usepackage{thm-restate}
\usepackage[capitalize,nameinlink]{cleveref} 
\usepackage{tikz}
\usepackage{pgfplots}
\pgfplotsset{compat=1.13}
\usepackage{wrapfig}
\usepackage{enumerate}
\usepackage{braket}
\usepackage{qcircuit}
\usepackage{circledsteps} 
\usepackage{dsfont}

\usepackage{mleftright}

\usepackage{ulem}
\usepackage{xcolor}
\usepackage{titlesec}

\usepackage{footnote}
\usepackage{graphicx}

\usepackage[utf8]{inputenc}
\usepackage{verbatim}
\usepackage{soul}
\usepackage{booktabs}
\usepackage{xr}

\newcommand{\mparen}[1]{\mleft(#1\mright)}
\newcommand{\mbracket}[1]{\mleft[#1\mright]}
\newcommand{\mbrace}[1]{\mleft\{#1\mright\}}
\renewcommand{\abs}[1]{\mleft|#1\mright|}
\renewcommand{\norm}[1]{\mleft\|#1\mright\|}

\newtheorem{theorem}{Theorem}

\newtheorem{corollary}{Corollary}

\newtheorem{definition}{Definition}

\newtheorem{lemma}{Lemma}

\newtheorem{proposition}{Proposition}

\newcommand{\eps}{\varepsilon}
\newcommand{\E}{\mathop{\mathbb{E}}}

\renewcommand{\erf}{\mathrm{erf}}
\newcommand{\Naturals}{\mathbb{N}}
\newcommand{\Reals}{\mathbb{R}}

\newcommand{\Unitaries}{U}
\newcommand{\SUnitaries}{SU}
\newcommand{\StateDist}{\mathcal{S}}
\newcommand{\UnitaryDist}{\mathcal{U}}
\newcommand{\poly}{\mathrm{poly}}

\newcommand{\polylog}{\mathrm{polylog}}

\newcommand{\PivotSet}{\mathcal{T}}
\newcommand{\F}{\mathbb{F}}
\DeclareMathOperator*{\argmax}{arg\,max}

\def\vev#1{\langle{#1}\rangle}

\def\and{\quad {\rm and} \quad}

\def\C{\mathbb{C}}
\def\CS{{\cal S}}
\def\ketbra#1{|{#1}\rangle\!\langle{#1}|}
\newcommand{\FXEB}{\mathcal{F}_{\mathsf{XEB}}}
\newcommand{\hFXEB}{\widehat{\mathcal{F}}_{\mathsf{XEB}}}



\begin{document}

\title{Demonstrating an unconditional separation between quantum and classical information resources}

\author{William Kretschmer}
\thanks{Corresponding author. Email: \url{kretsch@cs.utexas.edu}}
\affiliation{Department of Computer Science, The University of Texas at Austin, Austin, TX 78712, USA}
\affiliation{Simons Institute for the Theory of Computing, University of California, Berkeley, Berkeley, CA 94720, USA}
\author{Sabee Grewal}
\affiliation{Department of Computer Science, The University of Texas at Austin, Austin, TX 78712, USA}

\author{Matthew DeCross}
\affiliation{Quantinuum, Broomfield, CO 80021, USA}
\author{Justin A. Gerber}
\affiliation{Quantinuum, Broomfield, CO 80021, USA}
\author{Kevin Gilmore}
\affiliation{Quantinuum, Broomfield, CO 80021, USA}
\author{Dan Gresh}
\affiliation{Quantinuum, Broomfield, CO 80021, USA}
\author{Nicholas Hunter-Jones}
\affiliation{Department of Physics, The University of Texas at Austin, Austin, TX 78712, USA}
\affiliation{Department of Computer Science, The University of Texas at Austin, Austin, TX 78712, USA}
\author{Karl Mayer}
\affiliation{Quantinuum, Broomfield, CO 80021, USA}
\author{Brian Neyenhuis}
\affiliation{Quantinuum, Broomfield, CO 80021, USA}

\author{David Hayes}
\affiliation{Quantinuum, Broomfield, CO 80021, USA}

\author{Scott Aaronson}
\affiliation{Department of Computer Science, The University of Texas at Austin, Austin, TX 78712, USA}

\date{\today}

\begin{abstract} 
A longstanding goal in quantum information science is to demonstrate quantum computations that cannot be feasibly reproduced on a classical computer. 
Such demonstrations mark major milestones: they showcase fine control over quantum systems and are prerequisites for useful quantum computation. 
To date, quantum advantage has been demonstrated, for example, through violations of Bell inequalities and sampling-based quantum supremacy experiments.
However, both forms of advantage come with important caveats: Bell tests are not computationally difficult tasks, and the classical hardness of sampling experiments relies on unproven complexity-theoretic assumptions.
Here we demonstrate an unconditional quantum advantage in information resources required for a computational task, realized on Quantinuum's H1-1 trapped-ion quantum computer operating at a median two-qubit partial-entangler fidelity of $99.941(7)\%$. 
We construct a task for which the most space-efficient classical algorithm provably requires between 62 and 382 bits of memory, and solve it using only 12 qubits.
Our result provides the most direct evidence yet that currently existing quantum processors can generate and manipulate entangled states of sufficient complexity to access the exponentiality of Hilbert space. This form of quantum advantage---which we call \emph{quantum information supremacy}---represents a new benchmark in quantum computing, one that does not rely on unproven conjectures. 
\end{abstract}

\maketitle

Quantum theory postulates that a system of $n$ particles is represented by a vector in Hilbert space with exponentially many continuous parameters. 
The goal of quantum computing is to determine when and how this exponentiality can be harnessed to perform computations that are infeasible for classical computers.
Numerous theoretical results identify quantum advantages, most famously Shor’s algorithm for factoring integers~\cite{Sho99-factoring}, which achieves an exponential speedup over the best known classical algorithms.

Despite the progress of the past 30 years, 
it has not yet been demonstrated to skeptics' satisfaction that quantum computational advantage can be achieved in reality.
Skeptics have argued that while quantum algorithms like Shor’s exhibit exponential speedup \emph{in theory}, this speedup may never be physically accessible~\cite{levin2003tale,dyakonov2019,kalai2020argument,Hooft_1999, Wolfram2002}. At the heart of many of these arguments lies the belief that physically realizable quantum systems cannot be engineered 
to exploit the exponential dimensionality of Hilbert space. Reasons for this skepticism range from practical concerns about decoherence and error accumulation to deeper claims about the completeness of quantum mechanics itself.
Indeed, the question of whether Hilbert space is physical---or merely a mathematical convenience---has been debated since the early days of quantum mechanics~\cite{Fuchs_2011}.
Experiments that probe the exponentiality of Hilbert space are thus important for both the foundations of physics and for establishing that quantum algorithms can access more powerful information resources than classical ones, independent of assumptions about noise levels or advances in classical algorithms.

To date, the most widely recognized experimental demonstrations of quantum advantage fall into two broad categories. The first, violations of Bell inequalities, exhibit the nonlocal nature of quantum correlations, which cannot be replicated by any ``locally real'' classical theory \cite{Clauser-BellInequalityViolation,Aspect-BellInequalityViolation,Zeilinger-BellInequalityViolation}. These experiments, recognized by the 2022 Nobel Prize in Physics, do not directly probe questions of computational power or the exponentiality of quantum mechanics. 
The second category consists of sampling-based quantum supremacy experiments~\cite{AAB+19-google-supremacy, WBC+21-ustc-superconducting, ZCZ21, Sycamore70, ZCZ24, PhysRevX.15.021052, madsen2022quantum}, where quantum devices sample from probability distributions that are widely believed to be classically intractable to sample from faithfully.
However, the exponential classical hardness of these sampling problems is \emph{conditional}, i.e., it relies on unproven complexity-theoretic assumptions~\cite{AA13-boson-sampling,AG20-xeb}.

Fundamental questions therefore remain: Can we experimentally demonstrate an \textit{unconditional} quantum advantage for a well-defined computational task?
And can a quantum device perform a computation that establishes the exponential dimensionality of Hilbert space as a physically accessible resource?

We answer these questions in the affirmative. Using the tools of \textit{one-way communication complexity}, we formulate a computational task for which any classical algorithm provably requires at least 62 bits of memory, and we solve this task using only 12 qubits on a trapped-ion quantum computer. This yields an unconditional separation between quantum and classical information resources, realized in experiment.

Following the terminology of \cite{ABK23-relation}, we refer to this form of advantage as \emph{quantum information supremacy}: an experimental demonstration in which a quantum device solves a task using significantly fewer qubits than the number of bits required by any classical algorithm, according to a provable lower bound.

Our result provides direct evidence that today’s quantum hardware can prepare quantum states of sufficient complexity to access the exponentiality of Hilbert space. These states cannot be simulated, compressed, or described by any small number of classical bits. Thus, our experiment directly challenges the view that physically realizable quantum systems are always ``secretly classical” or reducible to low-dimensional descriptions.

\section*{From Communication Complexity to Quantum Information Supremacy} 
The design of our experiment centers around a quantum advantage for a one-way communication task.
In a typical such task, two parties---Alice and Bob---receive input strings $x$ and $y$, respectively, and wish to jointly compute a function on their inputs. However, their interaction is subject to a communication bottleneck: Alice may only send a single message $m_x$ to Bob, after which Bob alone must perform a computation that outputs the result $z$ (\Cref{fig:communication-complexity}).
Their goal is to minimize the length of a message $m_x$ that enables Bob to compute the correct output.
A trivial protocol for any such task is for Alice to send her entire input to Bob, i.e., $m_x = x$.
The challenge, then, is to identify tasks for which the communication required is significantly smaller than the input size.
The one-way communication model is also a setting for quantum advantage: for certain tasks~\cite{BJK08-comm,GKK+09-comm,Mon19-transmit}, the use of a quantum message $\ket{\psi_x}$ can provably reduce the amount of communication required by an exponential factor compared to any protocol that uses classical communication alone.

The basic premise of quantum information supremacy~\cite{ABK23-relation} is to recast such a quantum advantage in communication as an advantage in \textit{storage}.
Rather than considering two parties, Alice and Bob, separated \emph{in space}, we envision a single device separated between two points \emph{in time}, $t_0$ and $t_1$.
In this analogy, the communication bottleneck between Alice and Bob becomes a limitation on the amount of information that the device can transmit across the temporal boundary (\Cref{fig:circuit-template}).
This reinterpretation is desirable because it enables experimental realization on existing hardware without requiring distributed quantum systems.
A demonstration of quantum advantage in this setting---where much more than $n$ bits of classical memory are required to perform comparably to a device using $n$ qubits---exhibits a measurable gap in information capacity between classical and quantum resources.

\begin{figure}
  \begin{subfigure}[b]{0.45\textwidth}
    \hspace*{0.275cm}\includegraphics[width=\textwidth]{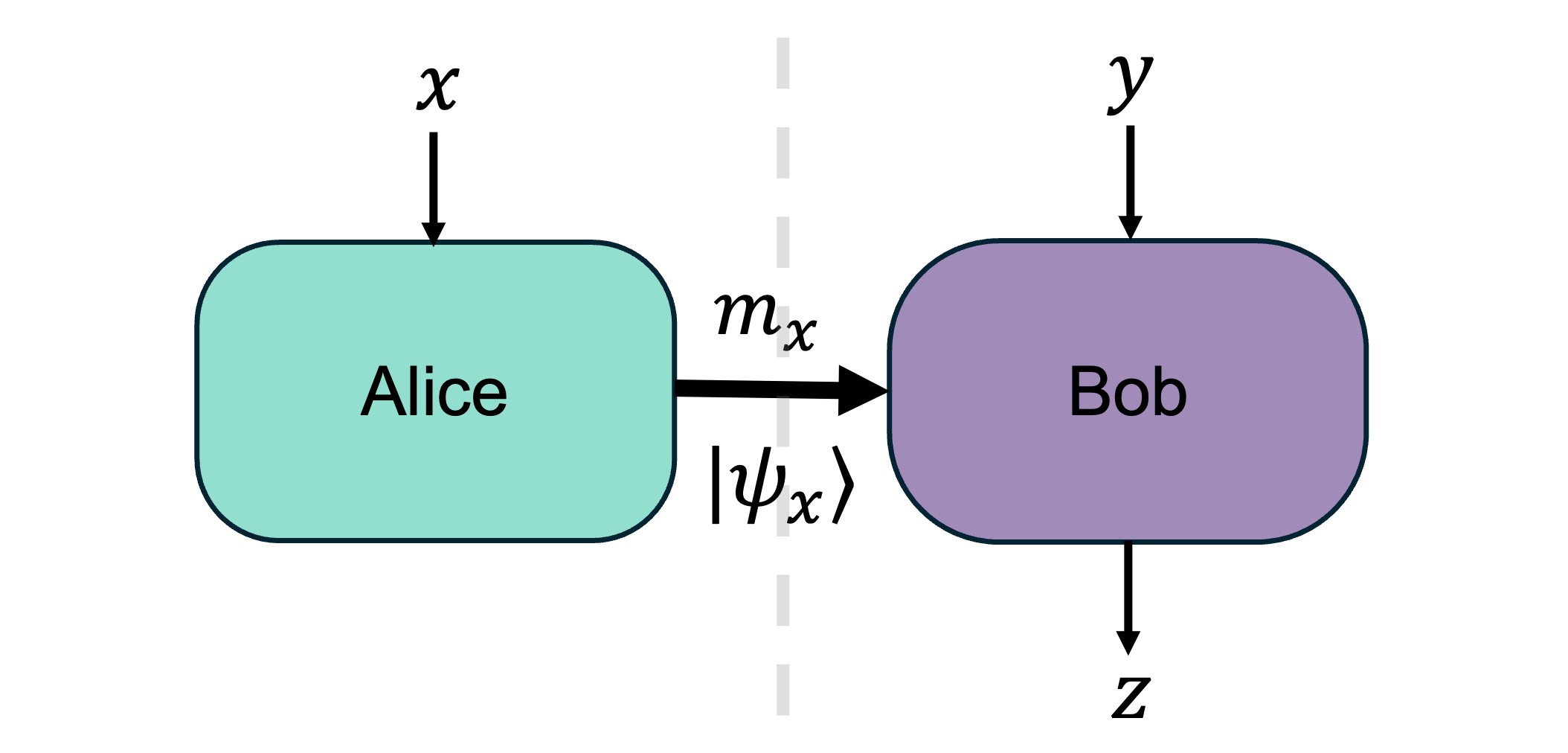}
    \caption{One-way communication model.
    Alice receives input $x$, Bob receives input $y$, and Alice sends a classical message $m_x$ or a quantum message $\ket{\psi_x}$ to Bob. Bob then outputs a value $z$ based on the message and his input.}
    \label{fig:communication-complexity}
  \end{subfigure}
  \hfill
  \begin{subfigure}[b]{0.45\textwidth}
    \vspace*{0.25cm}
    \includegraphics[width=\textwidth]{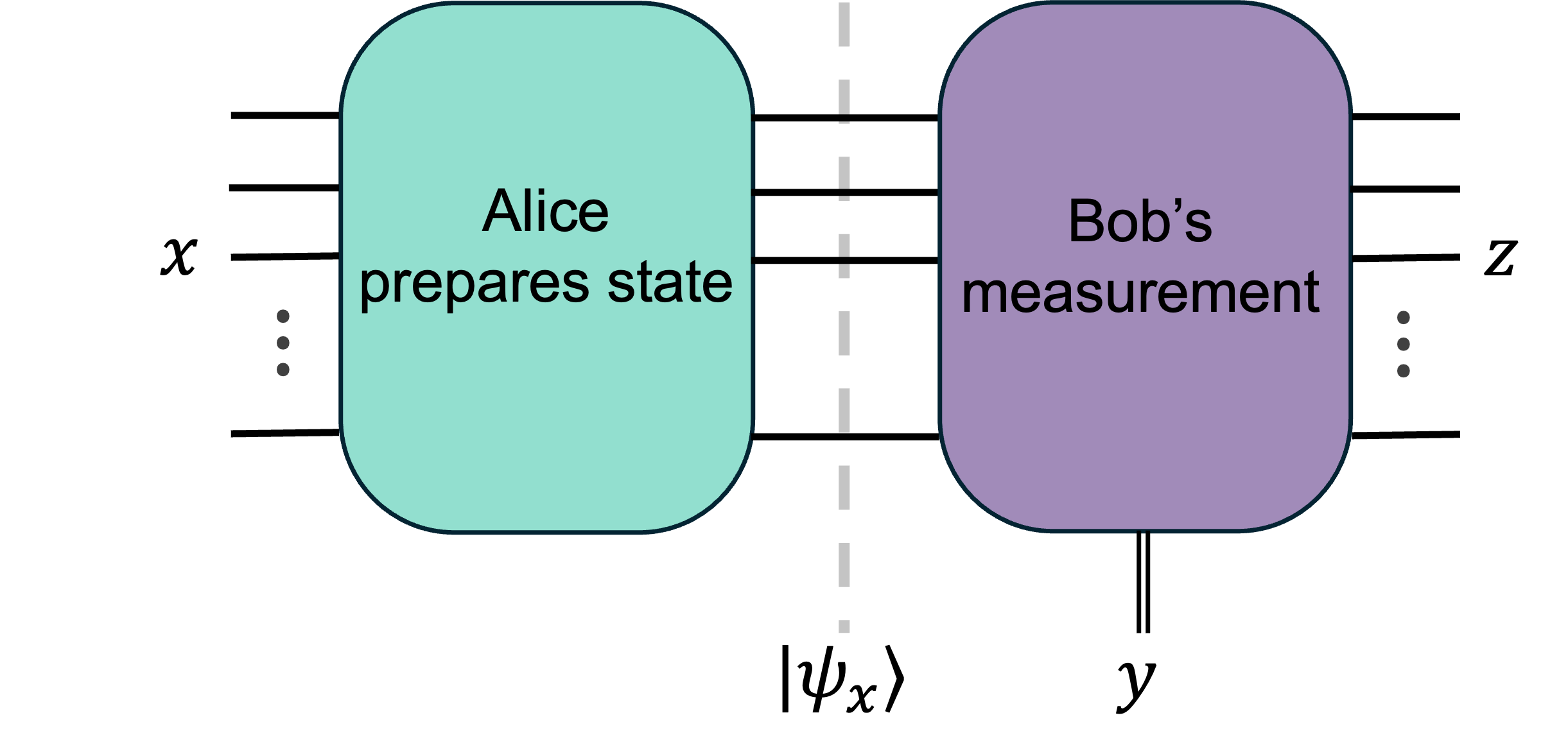}
    \caption{Realization of the communication task as a quantum circuit. The communication model is reinterpreted as a time-separated process on a single device. Given input $x$, Alice prepares a quantum state $\ket{\psi_x}$, which is passed forward in time to Bob. Bob performs a measurement determined by his input $y$, producing an output $z$.}
    \label{fig:circuit-template}
  \end{subfigure}
  \caption{One-way communication between two parties (a) is reinterpreted as a time-separated process on a single quantum device (b).}
  \label{fig:template-for-qis}
\end{figure}

\section*{Improved Quantum-Classical Separations} 
Although asymptotically exponential quantum advantages in one-way communication complexity have long been known~\cite{BJK08-comm}, previously established separations are insufficiently strong to achieve nontrivial quantum information supremacy on existing hardware (see Supplemental Material).
To remedy this, our chief theoretical contribution is a new one-way communication problem that can be solved using $n$ qubits of quantum communication, but which provably requires on the order of $2^n$ bits of communication to solve, up to lower-order terms in the exponent.

The problem we introduce is essentially a distributed version of linear cross-entropy benchmarking (XEB), which has seen use in other quantum supremacy demonstrations~\cite{AAB+19-google-supremacy, WBC+21-ustc-superconducting, ZCZ21, Sycamore70, ZCZ24, PhysRevX.15.021052}.
Put simply, the task asks Alice, who holds a quantum state, and Bob, who holds a quantum measurement, to work together to produce bit strings that are high-probability outputs of Bob's measurement on Alice's state. %
Formally, in a single instance of the problem, Alice's input $x$ is a classical description of some $n$-qubit quantum state $\ket{\psi_x}$, while Bob's input is a description of an $n$-qubit measurement in the form of a circuit $U_y$.
The task is to produce length-$n$ bit strings $z$ whose distribution, over many repetitions of the problem, achieves high linear cross-entropy benchmarking fidelity
\[
\FXEB \coloneqq \E\mbracket{2^n \abs{\braket{z|U_y|\psi_x}}^2 - 1}
\]
when $U_y$ is applied to $\ket{\psi_x}$, with the expectation taken over the input distribution $(x,y)$ and any additional randomness used by Alice and Bob in the protocol.
The quantum protocol for this task is simple: Alice's message to Bob is $\ket{\psi_x}$, and Bob processes the state by applying $U_y$ and measuring in the computational basis.
Assuming that the state $\ket{\psi_x}$ and circuit $U_y$ exhibit sufficiently strong anticoncentration, the $\FXEB$ achieved by implementation on a quantum device is well-approximated by the fidelity of the device~(see \cite{AAB+19-google-supremacy} and SM~\Cref{sec:variational}). 
Thus, most of the work lies in proving that any classical protocol for our distributed XEB task requires an exponential number of bits to replicate the behavior of the quantum protocol, even in the presence of non-negligible device noise. %
Our main result (proved in the Supplemental Material) establishes this:

\begin{theorem}\label{thm:classical-lb}
    Let Alice's state be sampled from the Haar measure over $n$-qubit states, and let Bob's measurement be sampled independently from the uniform distribution over $n$-qubit Clifford measurements.
    Any classical protocol that achieves average $\FXEB \geq \eps$ with respect to this input distribution must use at least
    \[
    \min \mbrace{\Omega\mparen{\eps^2 2^n}, \eps 2^{n - O(\sqrt{n})}}
    \]
    bits of communication. 
\end{theorem}
The Supplemental Material (\Cref{sec:norms}) also derives lower bounds for other measurement ensembles beyond random Cliffords.
We emphasize that the correlation of $\FXEB$ with device fidelity is not a necessary assumption for quantum information supremacy: the classical lower bound holds against any device demonstrating large $\FXEB$.

Asymptotically, \Cref{thm:classical-lb} represents a nearly quadratic improvement over the state of the art separations in one-way communication~\cite{BJK08-comm,BRSW12-bell}.
Moreover, the factors hidden in the big-$\Omega$ and big-$O$ are fully computable and reasonable in practice (see SM~\Cref{sec:bound_summary}): \Cref{thm:classical-lb} witnesses a quantum advantage with as few as $n = 7$ qubits, compared to $n = 9$ for the best previously-known separation (\Cref{fig:noiseless_bound}).

\begin{figure}
    \centering
    \includegraphics[width=\linewidth]{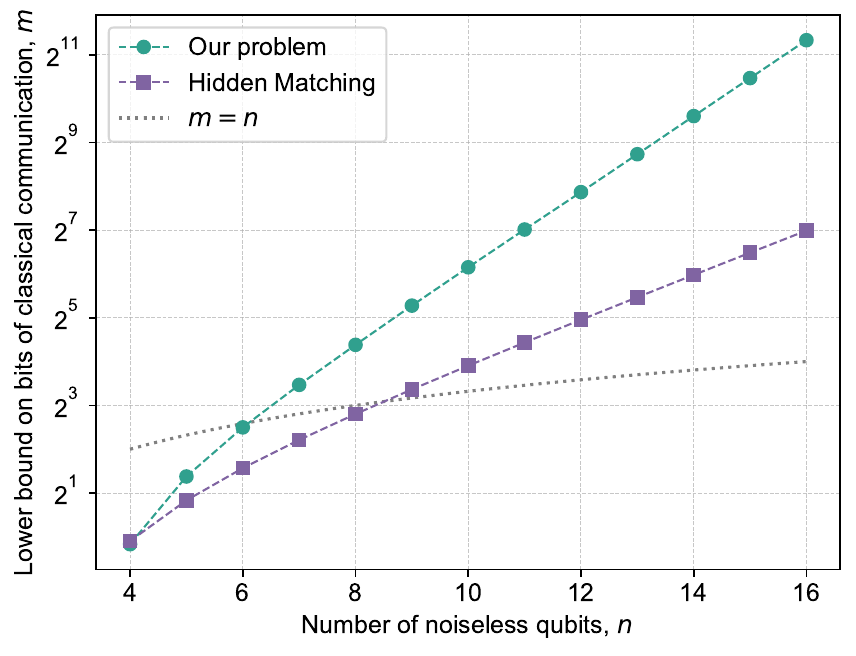}
    \caption{Comparison of classical one-way communication lower bounds for our problem (green, from \Cref{thm:classical-lb}) and Hidden Matching (purple, from~\cite{BRSW12-bell}), assuming implementation on a noiseless $n$-qubit quantum device (i.e., $\eps = 1$).
    Points above the gray curve ($m = n$) indicate quantum advantage.}
    \label{fig:noiseless_bound}
\end{figure}

Complementing \Cref{thm:classical-lb}, we additionally derive an \textit{upper} bound on the number of bits that suffice to achieve $\FXEB \ge \eps$ with classical communication.
Specifically, we exhibit a classical protocol that achieves $\FXEB \ge \eps$ using roughly $\frac{\eps 2^n}{\ln(2)^2 n}$ bits of communication for $\eps \le 1$ (see SM~\Cref{sec:classical_ub} for details).

\section*{Implementation and Experimental Results}  
We carry out an experimental realization of the protocol underlying \Cref{thm:classical-lb} on the Quantinuum H1-1 quantum computer \cite{expt_data_full}, which provides $20$ fully-connected trapped-ion qubits. The quantum charge-coupled device (QCCD) architecture~\cite{Pino2020} implemented by H1-1 supports arbitrary qubit connectivity and high-fidelity gate operations, making it well-suited for achieving large $\FXEB$. Performance metrics for H1-1 have been reported in prior work~\cite{Pino2020,RyanAnderson2021,RyanAnderson2022,github_spec}, and we report updated benchmarking specifications relevant to the experiment %
used in this study.

The dominant sources of error on H1-1 are two-qubit gates and memory errors (errors incurred during, e.g., qubit idling, transport, and cooling). 
Using a variation of randomized benchmarking~\cite{Magesan2011, Proctor2019}, we measure a two-qubit gate error that increases linearly with the gate angle, with an error of $5.9(7)\times 10^{-4}$ for the partially entangling gate at the typical (median) gate angle of $\theta \approx 0.213\pi$ used in the experiment. 

Several optimizations were employed to improve the quality of Alice's state preparation and Bob's measurement.
In the ideal (noiseless) version of the task, Alice would receive a complete description of a Haar-random state and prepare it exactly.
However, generating exact Haar-random states is infeasible on near-term quantum hardware.
Instead, we variationally train a parameterized quantum circuit to approximate Alice's Haar-random message $\ket{\psi_x}$ with high fidelity, using a loss function that accounts for the estimated noise from entangling gates and memory error.
Specifically, the circuit is a brickwork ansatz with periodic boundary conditions consisting of alternating layers of arbitrary $SU(2)$ one-qubit gates and two-qubit $ZZ(\theta) = \exp(-i(\theta / 2) Z \otimes Z)$ gates, with $86$ two-qubit layers in total.
An example of this circuit structure is shown in \Cref{subfiga:param-circuit}.
The expected error incurred by this method is significantly smaller than the expected error from implementing Haar-random states directly.
In particular, the average predicted fidelity with $\ket{\psi_x}$ for the trained circuits was $0.464$, whereas an implementation of state-of-the-art arbitrary state preparation~\cite{ICKHC16-isometries} would achieve fidelity below $0.004$ under our noise model.

Bob’s Clifford measurement admits a practically efficient implementation, which we derive in the Supplemental Material. As shown in \Cref{subfigb:clifford-meas}, any Clifford measurement can be expressed using a fixed circuit template, with variation introduced solely through classical control bits. This structure is well-suited to current hardware, which supports mid-circuit classically controlled operations but requires the locations of conditional gates to be fixed in advance to compile to low-level instructions.

In our implementation, Bob’s input string $y$ determines the measurement basis by specifying the values of these classical control bits. These controls are preloaded onto FPGA-based control hardware located several meters from the quantum processor. During execution, the FPGA reads the control bits after state preparation is complete and triggers gate operations accordingly via RF-modulated laser pulses. These pulses are routed through optical fibers and across free space to the trapped-ion QPU. 

\begin{figure}
\centering
  \begin{subfigure}[b]{0.45\textwidth}
    \centering
    \includegraphics[scale=0.52]{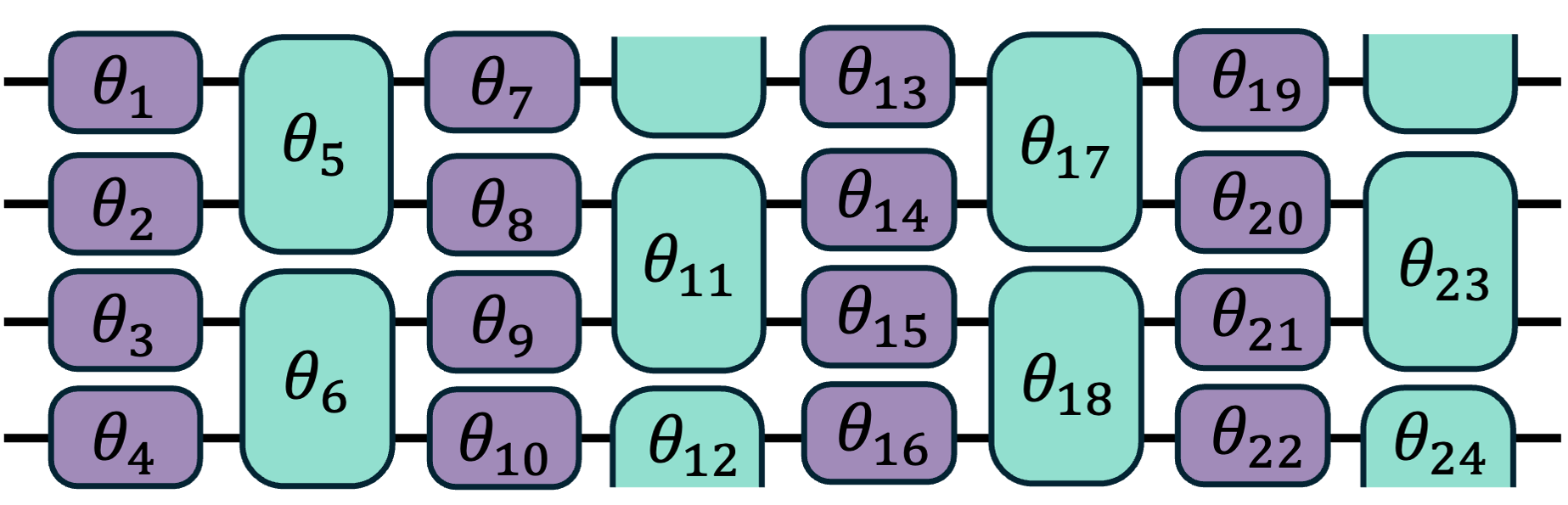}
    \caption{$4$-qubit example parameterized circuit used for quantum state preparation with $4$ $ZZ(\theta)$ layers. The ansatz consists of alternating layers of parameterized arbitrary one-qubit gates and brickwork two-qubit $ZZ(\theta)$ gates with periodic boundary conditions.
    In our implementation, this structure is scaled to $12$ qubits with $86$ $ZZ(\theta)$ layers.} 
    \label{subfiga:param-circuit}
  \end{subfigure}
  \hfill
  \begin{subfigure}[b]{0.45\textwidth}
    \centering
    \includegraphics[scale=0.52]{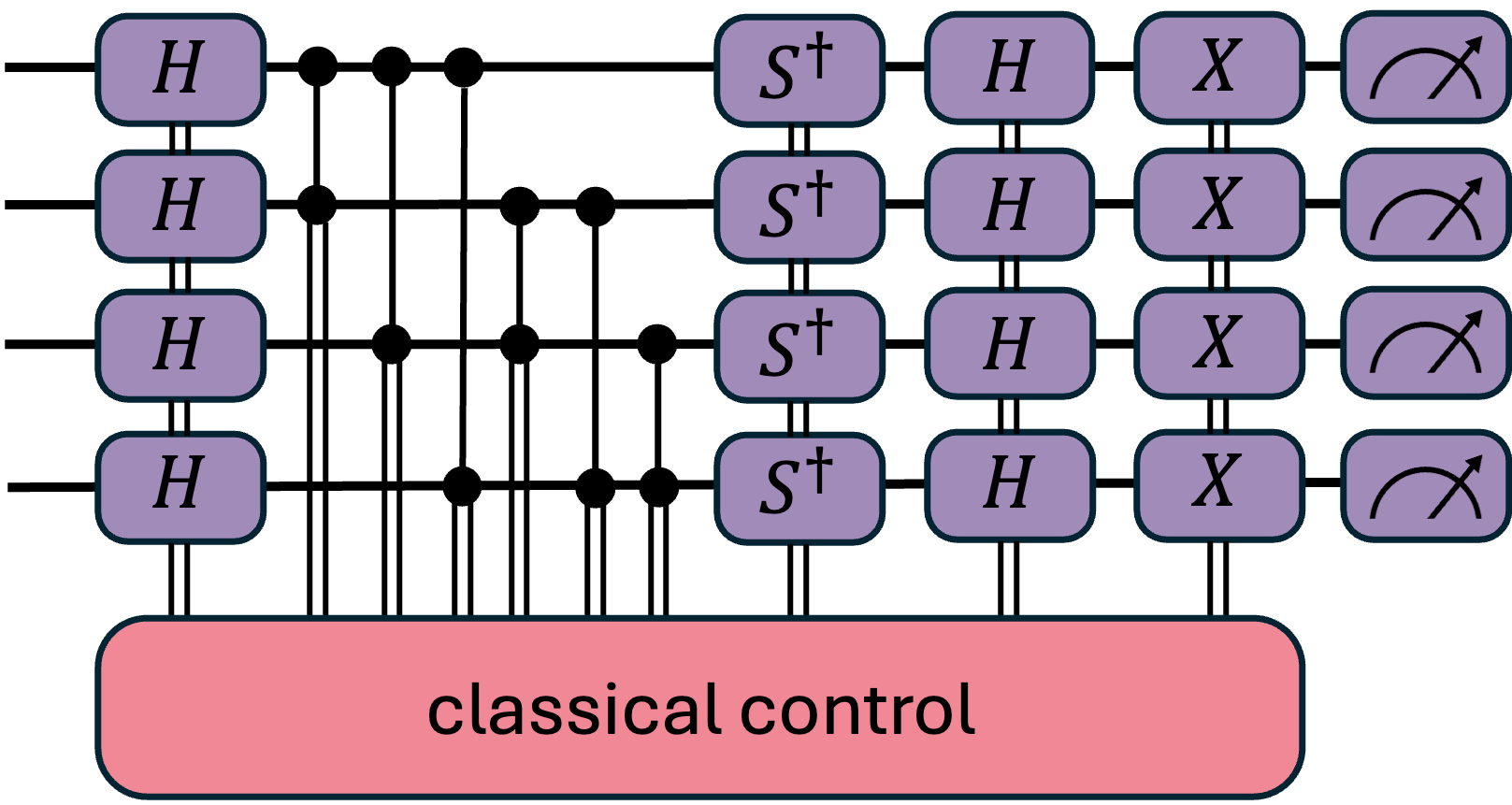}
    \caption{$4$-qubit example circuit template for implementing Clifford measurements.
    The structure is fixed, comprising layers of Hadamard, controlled-$Z$, inverse Phase ($S^\dagger$), Hadamard, and Pauli-$X$ gates.
    Classical controls are the only source of variation; by toggling these controls, the same circuit template can implement any Clifford measurement (see Supplemental Material).} 
    \label{subfigb:clifford-meas}
  \end{subfigure}
  \caption{$4$-qubit example circuits for implementing our protocol.
  The full implementation consists of concatenating the parameterized ansatz circuit (a) with the measurement circuit (b). 
  } 
  \label{fig:main-circuit}
\end{figure}

The experiment requires a source of true randomness to generate the inputs $x$ and $y$, because the classical lower bound (\Cref{thm:classical-lb}) only holds in the case where $x$ and $y$ are uniformly random. 
Thus, pseudorandom bits do not suffice to establish a \emph{provable} quantum advantage.
Every trial in our experiment used a fresh source of true randomness, which was obtained via a hardware-based randomness source as described in~\cite{Haw-true-randomness}.

Our experiment achieved a sample average XEB fidelity of $\hFXEB = 0.427(13)$, computed from 10,000 independent trials.
Applying \Cref{thm:classical-lb} with $n=12$ and $\eps = 0.427$ implies that any classical protocol achieving the same average score would require at least 78 bits of memory. 
To conservatively account for statistical uncertainty, we also consider $\eps = 0.362$, five standard errors below the mean, which yields a lower bound of 62 bits---well above the 12 qubits used in the quantum implementation.
Conversely, our classical upper bound (SM~\Cref{sec:classical_ub}) shows that a classical protocol using 330 bits of memory can achieve $\FXEB = 0.427$, while one using 382 bits can reach $\FXEB = 0.492$, five standard errors above the mean.
We deduce with high confidence that at least $62$ bits are necessary to replicate our experimental outcomes and at most $382$ bits suffice.
Our results are displayed in \cref{fig:achieved_xeb}, and a detailed methodology for obtaining these bounds appears in the Supplemental Material.

\begin{figure}
    \centering
    \includegraphics[width=\linewidth]{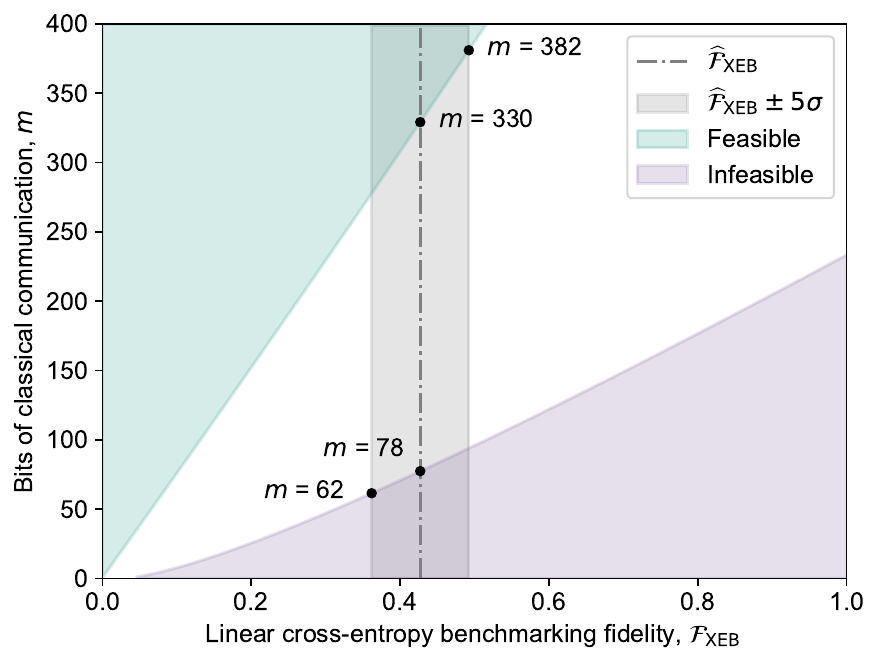}
    \caption{Classical communication bounds as a function of $\FXEB$ for the $n=12$ task.
    The lower shaded region (purple) indicates infeasibility, meaning that no $m$-bit classical protocol achieves the given $\FXEB$.
    Conversely, the upper shaded region (green) indicates feasibility, meaning that there exists an $m$-bit classical protocol to achieve the given $\FXEB$.
    The vertical line marks the sample mean $\hFXEB = 0.427$ achieved by our experiment, and the vertical shaded region (gray) marks $\hFXEB \pm 5\sigma$.
    }
    \label{fig:achieved_xeb}
\end{figure}

\section*{Discussion and Future Work}

We performed a task using 12 qubits on Quantinuum's H1-1 trapped-ion quantum computer, for which any classical protocol achieving comparable performance provably requires at least 62 bits of memory. This establishes an unconditional experimental separation between quantum and classical information resources. 
Unlike prior demonstrations of quantum advantage that rely on unproven complexity assumptions, our result is provable and permanent: no future development in classical algorithms can close this gap.%

Several recent advances made this experiment possible---indeed, this experiment could not have been performed even a few years ago.
First, the high-fidelity two-qubit gates available on the H1-1 trapped-ion quantum processor were crucial for reliably implementing our protocol; we estimate that with all else held equal, a maximum average two-qubit gate infidelity of $2.5\times 10^{-3}$ would be required to show a fairly modest separation of $2n$ classical bits vs.\ $n$ qubits.
Second, new theoretical developments (specifically, \Cref{thm:classical-lb}) identified a communication task with a sharp enough quantum–classical separation to enable an experimental test.
Third, the use of optimized parameterized circuits was essential for preparing Alice's message state with sufficiently high fidelity on near-term hardware.

Our experiment is not the first attempt to demonstrate quantum advantage via one-way communication complexity.
In particular, Kumar, Kerenidis, and Diamanti~\cite{KKD19-comm} reported an experiment that plausibly achieved a quantum information advantage for the Hidden Matching problem~\cite{BJK08-comm} using a linear optical setup.
However, their implementation used an exponentially large number of optical modes and no entanglement, so it did not probe whether $n$ physical qubits yield an exponential state space.
Additionally, while their performance exceeded \textit{known} classical protocols, it did not surpass the \textit{proven} classical lower bound, leaving quantum advantage unconfirmed (see SM~\Cref{sec:related-work} for further discussion).

Just as the first Bell inequality violations contained loopholes that were only closed in later experiments, our experiment is subject to similar caveats that future work may address or eliminate.
For example, a skeptic could object that our experiment did not maintain a true temporal separation between Alice and Bob's inputs: even though Bob's random Clifford measurement was generated from an independent hardware randomness source and stored separately from the QPU, the randomness generation occurred \textit{before} preparation of Alice's state on the QPU, and thus the state of the QPU could in principle have foreknowledge of the basis in which it would be measured.
This objection is analogous to the ``setting-independence'' loophole (sometimes called the ``measurement-independence'' or ``freedom-of-choice'' loophole) in Bell experiments \cite{scheidl-freedom-of-choice-2010,bell-nonlocality-brunner-2014}.
A future experiment could instead generate Bob’s measurement choice only \textit{after} Alice's state is fully prepared, which would require more sophisticated hardware capable of querying an online source of randomness mid-circuit.

Another source of skepticism lies in whether the QPU truly realizes a 12-qubit Hilbert space, or whether it is better modeled as evolving within a larger Hilbert space due to environmental interactions, leakage to non-computational spin states, or other uncontrolled degrees of freedom. 
This opens the door to alternative interpretations.
For instance, a skeptic might believe that what appears as a highly-entangled 12-qubit state is actually a low-entanglement state in a larger space, where additional degrees of freedom conspire to mimic the expected quantum behavior. 
A future experiment demonstrating a much larger separation could quell these doubts.
For example, at $n=26$, at least $1.1$~million classical bits would be necessary to achieve $\FXEB \approx 1$.
Such a demonstration would presumably require hardware with much higher two-qubit gate fidelity, which is currently the main bottleneck.
However, we also believe that there remains room to tighten the classical lower bound or improve the fidelity attained by variational state preparation, either of which could lead to larger separations on existing hardware.

\section*{Acknowledgments}
We acknowledge the entire Quantinuum team for their many contributions toward the successful operation of the H1-1 quantum computer. 
We also acknowledge Honeywell for fabricating the trap used in this experiment and Aqacia for providing access to their QRNG dataset. 
We thank H.~Buhrman, F.~Curchod, and S.~Dasu for useful comments on the manuscript.

W.K. is supported by the U.S. Department of Energy, Office of Science, National Quantum Information Science Research Centers, Quantum Systems Accelerator, and by NSF Grant CCF-231173.
S.G. is supported in part by an IBM Ph.D.\ Fellowship.
N.H-J. and S.A. are supported in part by DOE Grant DE-SC0025615.
This work was done in part while S.G., N.H-J., and S.A. were visiting the Simons Institute for the Theory of Computing, supported by NSF Grant QLCI-2016245.

\section*{Author Contributions}

W.K. and S.A. conceived the project.
W.K., S.G., and N.H.-J. contributed to the complexity-theoretic analysis underlying \Cref{thm:classical-lb}.
W.K., S.G., M.D., D.H., and S.A. designed the protocol for experimental implementation of the state preparation and random Clifford measurement.
W.K., S.G., and M.D. developed the classical and quantum programs required to implement the protocol and collected and analyzed the experimental data. 
M.D. and K.M. performed device benchmarking and collected and analyzed the resulting data. 
J.A.G., K.G., D.G., and B.N. maintained and operated the trapped-ion quantum hardware used in the experiment.
All authors contributed to the preparation of the manuscript and the Supplemental Material.

\section*{Data and Code Availability}

The full data presented in this work is available at \href{https://github.com/sabeegrewal/quantum-info-supremacy-data}{\texttt{quantum-info-supremacy-data}}.
The code required to verify and reproduce the results presented in this work is available at \href{https://github.com/sabeegrewal/quantum-info-supremacy}{\texttt{quantum-info-supremacy}}.

\section*{Competing Interests}

The authors declare no competing interests.

\clearpage
\onecolumngrid
\appendix

\setcounter{equation}{0}
\setcounter{figure}{0}
\renewcommand{\theequation}{S\arabic{equation}}
\renewcommand{\thefigure}{S\arabic{figure}}
\renewcommand{\appendixname}{}

\begin{center}
{\large \bf Supplemental Material for:\\ ``Demonstrating an unconditional separation between quantum and classical information resources''}
\end{center}

\tableofcontents

\clearpage
\section{Background}

Our goal in this work is to demonstrate quantum information supremacy. That is, we wish to conduct an experimental demonstration in which a quantum device solves a task using significantly fewer qubits than the number of bits required by any classical algorithm, according to a provable lower bound.
Our approach follows the high-level roadmap sketched by Aaronson, Buhrman, and Kretschmer \cite{ABK23-relation}, where they suggested leveraging unconditional quantum-classical separations in one-way communication complexity to perform such a demonstration.
In this section, we review how quantum-classical separations in communication complexity can be adapted into experimental protocols for demonstrating quantum information supremacy, and we explain why previously established separations are insufficient for realization on current quantum hardware.

\subsection{From Communication Complexity to Quantum Information Supremacy}

In a one-way communication problem, two parties, Alice and Bob, receive inputs $x$ and $y$, respectively.
Their goal is to perform some computational task $T$, which might be to evaluate a function $f(x, y)$, to output a string $z$ such that $((x,y),z)$ belongs to a specified relation, to sample from some distribution depending on $x$ and $y$, or otherwise. 
To do so, Alice sends a message $m_x$ to Bob that may depend on $x$, and then Bob computes an output that may depend on $m_x$ and $y$.
The cost of the protocol is measured by the length of Alice's message $m_x$ that is communicated between the parties; all other computation is irrelevant to the complexity measure.
The sole difference between the quantum and classical one-way communication models is whether or not Alice's message $m_x$ is comprised of quantum or classical bits.

Exponential quantum advantages in one-way communication complexity have long been known.
A simple example is the Hidden Matching (HM) problem due to Bar-Yossef, Jayram, and Kerenidis \cite{BJK08-comm}.
In HM, Alice receives a string $x \in \{0,1\}^N$, and Bob receives a perfect matching $y \subset [N] \times [N]$ on the set $[N]$.
Their task is to output a tuple $(i, j, z)$ such that $(i, j) \in y$ and $x_i \oplus x_j = z$.
When $N = 2^n$, \cite{BJK08-comm} showed that HM is solvable with unit probability by a quantum communication protocol in which Alice sends Bob only $n$ qubits of information.
On the other hand, they proved that any classical bounded-error randomized protocol for the same task requires $\Omega\mparen{2^{n/2}}$ bits.

Exponential separations between quantum and classical one-way communication can be leveraged to experimentally test the prediction that $n$ qubits require $\exp(n)$ bits to describe classically.
In particular, one could model Alice and Bob with two spatially-separated devices that have a communication bottleneck between them, and conduct the protocol by shuttling Alice's message from one device to the other.
If the protocol consistently uses $n$ qubits to solve a communication task that requires $\exp(n)$ classical bits, one concludes that the states sent by the protocol cannot be simulated or encoded using only a polynomial number of classical bits.
In particular, the states transmitted in this protocol must be physically realizing the exponential dimensionality of Hilbert space. 

Unfortunately, such a demonstration would prove quite challenging for existing quantum hardware, because it requires \textit{two} high-fidelity quantum devices \textit{and} the ability to communicate quantum information between them.
This is why Aaronson, Buhrman, and Kretschmer's proposal for quantum information supremacy~\cite{ABK23-relation} abstracts the protocol into one that can be performed on a single device.
Instead of two spatially-separated devices, one considers a device separated between two points \textit{in time}.
In other words, the device receives ``Alice's input'' $x$ at some initial time $t_0$, and ``Bob's input'' at a later time $t_1$.
Under this analogy, the message from Alice to Bob becomes the \textit{storage} of the device between $t_0$ and $t_1$.
And like above, a successful demonstration of the protocol shows that the state stored between $t_0$ and $t_1$ did not have a short classical description.

In complexity-theoretic language, this transformation from a communication protocol into a single-device protocol corresponds to mapping a communication complexity separation into a separation of complexity classes with quantum and classical \textit{advice}.
Informally, one can view advice as an arbitrary pre-computed auxiliary input meant to assist with a computation.
\cite{ABK23-relation} showed that the Hidden Matching lower bound of~\cite{BJK08-comm} lifts to a separation of the complexity classes $\mathsf{FBQP/qpoly}$ and $\mathsf{FBQP/poly}$, which are the classes of relational problems that can be solved efficiently by quantum algorithms with quantum or classical advice, respectively.
One can then interpret the preparation of Alice's message at time $t_0$ as preparing the quantum advice state used in the separation of $\mathsf{FBQP/qpoly}$ and $\mathsf{FBQP/poly}$.
Thus, this form of quantum information supremacy is also a physical realization of the separation $\mathsf{FBQP/qpoly} \neq \mathsf{FBQP/poly}$.

\subsection{Desiderata for Quantum Information Supremacy}\label{subsec:desiderata}

An exponential separation between quantum and classical one-way communication complexities does not, by itself, suffice to realize quantum information supremacy.
Any quantum-classical communication separation must also satisfy five additional requirements to be experimentally feasible on current hardware:

\begin{enumerate}
    \item \textbf{Tolerance for noise.} Quantum computers today are noisy.
    If we run a deep computation on an existing quantum device in an attempt to prepare a general state $\ket{\psi}$, the output state might only have some small fidelity $\eps$ relative to $\ket{\psi}$.
    Thus, the protocol intended to demonstrate quantum information supremacy might succeed with probability much less than $1$.
    Therefore, the lower bound must rule out the possibility that the protocol becomes efficiently classically simulable in the presence of noise.
    In other words, it is essential that the exponential classical communication lower bound also holds against protocols that succeed with the same probability as the noisy quantum device.

    \item \textbf{Average-case hardness.} To enable a meaningful experimental demonstration, the classical hardness must hold \textit{on average} over the input distributions sampled (not just for worst-case inputs).
    This is necessary for the feasibility of obtaining high confidence of a quantum advantage from finitely many samples.
    If there exist efficient classical protocols that succeed on all but a negligible $\delta$ fraction of inputs, then ruling them out would require at least $\Omega(1/\delta)$ trials.

    \item \textbf{Verifiability.} There must exist a benchmark for verifying that the quantum protocol succeeded in performing the task. 
    While this benchmark need not be computationally efficient, it should allow success to be certified from a reasonable number of trials.
    To give an example of a problem that lacks effective verifiability, consider the task of sampling from a distribution to small error in total variation distance.
    That is, the inputs $x$ and $y$ correspond to some distribution $\mathcal{D}_{x,y}$ such that the goal is for Bob to output a sample from a distribution that is close to $\mathcal{D}_{x,y}$.
    Then verifying statistical closeness to $\mathcal{D}_{x,y}$ requires a number of samples that scales with the support size of the distribution, which could be exponentially large. 

    \item \textbf{Simplicity of implementation.} 
    The noise rate of a circuit implemented on a quantum device without error correction scales with the complexity of its implementation.
    Thus, the quantum circuits used in the protocol should be as simple as possible.
    Unfortunately, for an exponential separation to hold, the preparation of Alice’s quantum message must inherently involve exponential complexity, for if it didn't, the description of the circuit that prepares Alice's message would yield an efficient classical protocol.
    However, the gate complexity of Bob's measurement---which might be as high as $\Omega(4^n)$ for a general $n$-qubit POVM---can and should be minimized.
    Ideally, the cost of Bob's measurement is small compared to the cost of state preparation.

    \item \textbf{Quantitative tightness.} To demonstrate a separation that remains observable in the presence of noise and with a limited number of qubits, the quantitative gap between quantum and classical communication should be as large as possible, ideally $n$ qubits vs.\ $\Omega(2^n)$ classical bits.
    Moreover, the constants in the asymptotic lower bound must be favorable enough to show a provable quantum advantage at realistic instance sizes.
    A stronger separation ensures that a meaningful advantage can be observed, even after accounting for experimental imperfections.  
\end{enumerate}

A natural first attempt at demonstrating quantum information supremacy is to leverage existing quantum-classical separations in communication complexity.
Leading candidates include the $n$ vs.\ $\Omega(2^{n/2})$ separations of (Boolean) Hidden Matching~\cite{BJK08-comm,GKK+09-comm} and the $n$ vs.\ $\Omega(2^n)$ separation of Distributed Quantum Sampling~\cite{Mon19-transmit}, which are the strongest-known quantum-classical separations in communication complexity.

Hidden Matching satisfies most of the requirements: Buhrman, Regev, Scarpa, and de Wolf established its tolerance for noise~\cite{BRSW12-bell}, the hardness is average case~\cite{BJK08-comm,BRSW12-bell}, membership in the relation is verifiable, and Aaronson, Buhrman, and Kretschmer~\cite{ABK23-relation} noted that the measurement is implementable with just $O(n)$ gates.
However, the classical lower bound falls short of quantitative tightness by a quadratic factor, and the advantage at practical instance sizes is too small.
Even in the noiseless setting, the best classical lower bound on communication complexity only exceeds $n$ bits for $n \ge 9$~\cite{BRSW12-bell}.

Montanaro's Distributed Quantum Sampling obtains a tighter separation~\cite{Mon19-transmit}.
But, as the name suggests, the task involves sampling from a distribution over $n$-bit strings to small error in total variation distance, necessitating an exponential number of trials to verify the separation. 
Additionally, Bob's measurement requires $\Omega(2^n/n)$ gates on average. 

Considering that the existing separations do not meet our requirements, our primary theoretical contribution is to introduce a new one-way communication problem that satisfies all five desiderata for quantum information supremacy.

\section{The Problem We Solve}

Our communication problem is inspired by the linear cross-entropy benchmark that has seen use in verifying quantum supremacy experiments based on random circuit sampling, such as those performed at Google \cite{AAB+19-google-supremacy, Sycamore70},  USTC \cite{WBC+21-ustc-superconducting, ZCZ21, ZCZ24}, and Quantinuum \cite{PhysRevX.15.021052}. The communication task is essentially a distributed version of the linear cross-entropy heavy output generation (XHOG) problem, which has been studied in other theoretical contexts~\cite{AG20-xeb,Kre21-tsirelson,AH23-random}.
In this distributed version of the problem, Alice receives a quantum state and Bob receives a measurement to perform on that state. Their shared goal is to achieve a high linear cross-entropy benchmark (XEB) with respect to the output distribution defined by applying Bob’s measurement to Alice’s state.
The problem is additionally average-case: it is parameterized by distributions over Alice's state and Bob's measurement.
We define it formally as follows.

\begin{definition}[$\eps$-DXHOG, or Distributed Linear Cross-Entropy Heavy Output Generation]
    Fix a distribution $\StateDist$ over (a classical description of) an $n$-qubit state $\ket{\psi}$ given to Alice.
    Fix an independent\footnote{The task remains well-defined even if the distributions $\StateDist$ and $\UnitaryDist$ are not independent, but they always will be independent for the distributions considered in this paper.} distribution $\UnitaryDist$ over (a classical description of) an $n$-qubit unitary $U$ given to Bob.
    Their task is to output a sample $z \in \{0,1\}^n$ such that
    \[
    \E_{\ket{\psi} \sim \StateDist, U \sim \UnitaryDist}\mbracket{ \abs{\braket{z|U|\psi}}^2} \ge \frac{1 + \eps}{2^n},
    \]
    where the expectation also averages over any source of randomness used by Alice and Bob in the protocol.
\end{definition}

Like in the ordinary XHOG task, the intuition is that achieving a large score for $\eps$ demonstrates that Bob's output $z$ is correlated with the larger-probability basis states of $\ket{\psi}$ in the $U$-basis.
An equivalent and sometimes more convenient way to define DXHOG is with respect to the average \textit{linear cross-entropy benchmarking fidelity}~\cite{AAB+19-google-supremacy}:
\[
\FXEB \coloneqq \E_{\ket{\psi} \sim \StateDist, U \sim \UnitaryDist}\mbracket{ 2^n \abs{\braket{z|U|\psi}}^2 - 1} = 2^n  \E_{\ket{\psi} \sim \StateDist, U \sim \UnitaryDist}\mbracket{ \abs{\braket{z|U|\psi}}^2} - 1,
\]
so that $\eps$-DXHOG means achieving $\FXEB \ge \eps$.
We emphasize that $\FXEB$ denotes the \emph{population} mean of the underlying random process that samples $z$.
In practice, to verify that quantum hardware samples from a distribution satisfying $\FXEB \ge \eps$, one computes the \emph{sample} mean $\hFXEB$ of a finite number of trials $\{\ket{\psi_i}, U_i, z_i\}_{i = 1,\ldots,k}$:
\[
\hFXEB \coloneqq \frac{1}{k} \sum_{i=1}^k 2^n\abs{\braket{z_i|U_i|\psi_i}}^2 - 1.
\]
Then, one can use a statistical hypothesis test based on $\hFXEB$ to deduce whether $\FXEB \ge \eps$ with high confidence.

A natural quantum protocol to consider for DXHOG is the following: Alice first sends the state $\ket{\psi}$ to Bob using $n$ qubits of communication. Then, Bob applies $U$ to $\ket{\psi}$ and measures in the computational basis to get $z$, which we denote by $z \sim U\ket{\psi}$.
Assuming that the measurement distribution of $U\ket{\psi}$ has Porter-Thomas statistics (or alternatively, either $\StateDist$ or $\UnitaryDist$ is a $2$-design), this protocol when implemented noiselessly achieves
\begin{align*}
    \FXEB &=
    2^n\E_{\ket{\psi} \sim \StateDist, U \sim \UnitaryDist} \mbracket{\E_{z \sim U\ket{\psi}} \mbracket{\abs{\braket{z|U|\psi}}^2}} - 1\\
    &= 2^n\E_{\ket{\psi} \sim \StateDist, U \sim \UnitaryDist} \mbracket{\sum_{z \in \{0,1\}^n} \abs{\braket{z|U|\psi}}^4} - 1\\
    &= 2^n \frac{2}{2^n + 1} - 1\\
    &\approx 1.
\end{align*}
There is also a trivial protocol that uses no communication whatsoever and achieves $\FXEB = 0$ for any $\StateDist$ and $\UnitaryDist$, by outputting $z$ uniformly at random:

\begin{align*}
    \FXEB &=
    2^n\E_{\ket{\psi} \sim \StateDist, U \sim \UnitaryDist} \mbracket{\E_{z \sim \{0,1\}^n} \mbracket{\abs{\braket{z|U|\psi}}^2}} - 1\\
    &= 2^n\E_{\ket{\psi} \sim \StateDist, U \sim \UnitaryDist} \mbracket{\sum_{z \in \{0,1\}^n} \frac{1}{2^n}\abs{\braket{z|U|\psi}}^2} - 1\\
    &= \frac{2^n}{2^n} - 1\\
    &= 0.
\end{align*}

Observe that if the quantum protocol is corrupted by some global depolarizing noise with rate $1 - \eps$, then the noisy protocol is equivalent to a probabilistic mixture of the quantum protocol with probability $\eps$ and the trivial protocol with probability $1 - \eps$.
Thus, by linearity of expectation, the noisy protocol achieves $\FXEB \approx \eps$.
So, our goal is to show that any classical protocol to achieve the same $\FXEB$ requires $\gg n$ bits of communication for realistic choices of $n$ and $\eps$. Our main theoretical result establishes this:

\begin{theorem}
\label{thm:main_informal}
    There exists a distribution $\StateDist$ over $n$-qubit states and a distribution $\UnitaryDist$ over $n$-qubit unitaries such that any randomized classical one-way communication protocol for $\eps$-DXHOG must use at least
    \[
    \Omega\mparen{\min \mbrace{\eps^2 2^n, \frac{\eps 2^n}{n}}}
    \]
    bits of communication from Alice to Bob. On the other hand, there is a noiseless $n$-qubit quantum protocol for $1$-DXHOG.
\end{theorem}

We claim that \Cref{thm:main_informal} satisfies all of our requirements:

\begin{enumerate}
    \item The tolerance for noise is captured in the dependence of the lower bound on the $\FXEB$, $\eps$.
    This assumes that $\FXEB$ is a good approximation to the fidelity, which is expected for sufficiently random states and unitaries~\cite{AAB+19-google-supremacy,GKC+21-xeb,DHJB24-noise}.
    We also numerically verified the correlation between $\FXEB$ and fidelity for the circuits used in experiment; see \Cref{fig:svsims} in \Cref{sec:variational}.
    \item The problem and lower bound are inherently average case, because they depend on distributions $\StateDist$ and $\UnitaryDist$ over the inputs.
    \item To verify correctness from $k$ samples $\{\ket{\psi_i}, U_i, z_i\}_{i = 1,\ldots,k}$ obtained from experiment, we compute the sample mean $\hFXEB$ and standard error $\sigma$ of the empirical scores $2^n \abs{\braket{z_i|U_i|\psi_i}}^2 - 1$.
    To deduce with high confidence (say, $5\sigma \approx 99.99994\%$) that the protocol is solving $\eps$-DXHOG, we simply check whether $\hFXEB - 5\sigma \ge \eps$.
    \item The state distribution $\StateDist$ will always be Haar-random, so the complexity of implementation depends only on the choice $\UnitaryDist$ of Bob's measurement distribution.
    We actually prove lower bounds for several different choices of $\UnitaryDist$, including random Cliffords, approximate $t$-designs, and the Haar measure.
    For the experimental demonstration, we pick $\UnitaryDist$ to be random Clifford, which admits implementation with $O(n^2 / \log n)$ gates~\cite{AG04-stabilizer}---much less than the cost of state preparation.
    \item Our separation between quantum and classical communication complexities can be as large as $n$ vs.\ $\Omega\mparen{2^n / n}$, at least if we consider noiseless quantum protocols ($\eps = 1$).
    This represents a near quadratic improvement over Hidden Matching~\cite{BJK08-comm}.
    Technically, our lower bound for Clifford $\UnitaryDist$ is slightly weaker: the $\frac{\eps 2^n}{n}$ term gets replaced by $\eps 2^{n - O(\sqrt{n})}$.
    However, the difference between the bounds is small at practical instance sizes; see \Cref{fig:ensemble_comparison} in \Cref{sec:bound_summary}.
    Additionally, a version of the bound with explicit constants is easily computable, though a bit cumbersome to write down (hence the big-$\Omega$ in \Cref{thm:main_informal}).
    The exact formula and further discussion appear in \Cref{sec:bound_summary}.
\end{enumerate}

Our main experimental result is a realization of quantum information supremacy on Quantinuum’s H1-1 trapped-ion quantum computer, based on \Cref{thm:main_informal}. We implemented the quantum protocol for DXHOG on 12 qubits and executed 10,000 independent instances, achieving a sample mean of $\hFXEB = 0.427$. By \Cref{thm:main_informal}, no classical protocol can match this score on average using fewer than $78$ bits of memory. Furthermore, even to achieve a population mean $\FXEB$ within $5\sigma$ of our sample mean $\hFXEB$ would require at least $62$ bits of classical memory. These results are unconditional, holding for any classical protocol, regardless of computation time. 
We formally prove \Cref{thm:main_informal} in \Cref{sec:classical_lb} and describe the details of the experiment in \Cref{sec:experiment}.

\section{Related Work}\label{sec:related-work}

On the theoretical side, we note again that prior work obtained an $n$ vs. $\Omega(2^{n/2})$ quantum-classical one-way communication separation for Hidden Matching ~\cite{BJK08-comm,GKK+09-comm,BRSW12-bell} and an $n$ vs. $\Omega(2^n)$ quantum-classical two-way communication separation for a sampling problem~\cite{Mon19-transmit}.
We discussed in \Cref{subsec:desiderata} why these separations do not suffice to experimentally realize quantum information supremacy on existing hardware, even though these were the strongest communication complexity separations known prior to this work.
We note that many variants and generalizations of Hidden Matching have seen study in communication complexity~\cite{GKK+09-comm,Mon11-comm,VY11-stream,SWY12-hyper,DM20-bhm}, but none obtained quantitatively stronger separations between quantum and classical one-way communication complexities than $n$ vs. $\Omega(2^{n/2})$.

Our work also bears a conceptual resemblance to earlier quantum communication protocols such as superdense coding~\cite{Bennett-superdense} and quantum random access codes (QRACs)~\cite{ANTV99-qrac,Nay99-qrac,ANTV02-qrac-combined}, which allow two communicating parties to encode classical information into quantum information. 
In superdense coding, two classical bits are communicated using only one qubit, provided that sender and receiver share entanglement in advance. 
In QRACs, one party encodes multiple classical bits into a small number of qubits, and the other party can probabilistically recover any one of the original bits.
Both protocols can be used to show an advantage of quantum over classical communication, but not an \textit{exponential} advantage: the factor of two in superdense coding is optimal~\cite{BSST02-capacity}, and quantum random access codes outperform classical encodings by at most a logarithmic factor~\cite{ANTV99-qrac,Nay99-qrac,ANTV02-qrac-combined,LdW25-qrac}.

The most related prior experiment is due to Kumar, Kerenidis, and Diamanti~\cite{KKD19-comm} who conducted a demonstration of the Bar-Yossef-Jayram-Kerenidis Hidden Matching protocol~\cite{BJK08-comm}. 
However, their demonstration simulated $n$ qubits using a number of optical modes that scaled exponentially in $n$.
Therefore, their experiment did not actually test the proposition that $n$ entangled qubits already give rise to states that require $\exp(n) \gg n$ classical bits to specify.
Moreover, while the experiment outperformed the best-known classical protocol, it did not exceed the performance of the proven classical lower bound for Hidden Matching, and thus their claim of quantum advantage remains plausible but unproven.

Another recent experiment~\cite{CKDK21-np-ver} claimed a demonstration of quantum advantage for an $\mathsf{NP}$ verification problem.
The demonstration was based on a result of Aaronson, Beigi, Drucker, Fefferman, and Shor~\cite{ABDFS09-unentanglement} that with a short quantum witness, a verifier can become efficiently convinced of the satisfiability of a size-$m$ 3SAT formula, assuming that the witness consists of roughly $\sqrt{m}$ unentangled quantum states of size $O(\log m)$ each.
By contrast, under a plausible computational assumption, any classical witness needs roughly $m$ bits to do the same.
This experiment is qualitatively incomparable to our demonstration of quantum information supremacy because it crucially relies on the promise of unentanglement between the states, and its advantage is based on computational efficiency, so it does not directly probe the states' information content.
Moreover, the experiment's claim to advantage depends on an unproven computational assumption, and the separation is only a quadratic (rather than exponential) quantum advantage, unlike ours.

A recent line of work~\cite{HKP21-info-bounds,ACQ22-measure, CCHL22-memory,HBC+22-experiments} has investigated the power of accessing multiple copies of an unknown quantum state in a machine learning context.
These results show that certain estimation tasks become exponentially easier when one can jointly manipulate several copies at once.
In particular, there are tasks that require $\exp(n)$ copies when only one copy is accessible at a time, but can be solved with 
$\poly(n)$ copies when collective access is allowed. This was recently demonstrated in proof-of-principle experiments~\cite{HBC+22-experiments}.
Compared to our work, these experiments demonstrate a different form of quantum advantage in storage.
They illustrate how expanding quantum memory---e.g., moving from access to a single copy of an unknown state to two copies---can reduce the number of copies required for specific estimation tasks.
However, such experiments do not establish a separation between quantum and classical information resources, nor do they test the exponentiality of Hilbert space.
In contrast, our experiment directly demonstrates a quantum-classical separation: we solve a computational task using significantly fewer qubits than the number of classical bits provably required, thereby establishing quantum information as a physically accessible resource that cannot be succinctly represented classically.

\section{Preliminaries}

The natural logarithm of $x$ is denoted $\ln(x)$, and $e \approx 2.718$ always represents the base of the natural logarithm.
$[n]$ denotes the set of integers $\{1, 2, \ldots, n\}$.
If $\PivotSet \subseteq [n]$, $\overline{\PivotSet}$ denotes its complement in $[n]$.
For a vector $v$, $\norm{v}$ denotes its Euclidean ($\ell_2$) norm.
The Frobenius norm of a matrix $M$ is $\norm{M}_F \coloneqq \sqrt{\Tr(M^\dagger M)}$, and its operator norm is $\norm{M}_{op} \coloneqq \sup_v \frac{\norm{Mv}}{\norm{v}}$.
Define $\erf:\C \to \C$ as $\erf(z) \coloneqq \frac{2}{\sqrt{\pi}}\int_0^z e^{-t^2}\mathrm{d}t$.

Below are definitions of standard quantum gates that appear frequently in this work:
\[
X \coloneqq \begin{bmatrix}
    0 & 1\\
    1 & 0
\end{bmatrix}
\qquad
Z \coloneqq \begin{bmatrix}
    1 & 0\\
    0 & -1
\end{bmatrix}
\qquad
H \coloneqq \frac{1}{\sqrt{2}}\begin{bmatrix}
    1 & 1\\
    1 & -1
\end{bmatrix}
\qquad
S \coloneqq \begin{bmatrix}
    1 & 0\\
    0 & i
\end{bmatrix}
\]
\[
CNOT \coloneqq \begin{bmatrix}
    1 & 0 & 0 & 0\\
    0 & 1 & 0 & 0\\
    0 & 0 & 0 & 1\\
    0 & 0 & 1 & 0\\
\end{bmatrix}
\qquad\qquad
CZ \coloneqq \begin{bmatrix}
    1 & 0 & 0 & 0\\
    0 & 1 & 0 & 0\\
    0 & 0 & 1 & 0\\
    0 & 0 & 0 & -1\\
\end{bmatrix}
\]

Another name for $H$ is the Hadamard gate, and another name for $S$ is the Phase gate.
The \textit{$n$-qubit Clifford group} is the group of $n$-qubit unitary transformations generated by $H$, $S$, and $CNOT$ gates. The \textit{$n$-qubit stabilizer states} are the set of $n$-qubit states reachable from $\ket{0^n}$ by applying a unitary transformation from the $n$-qubit Clifford group.

The group of $N$-dimensional unitary operators is denoted $\Unitaries(N)$, and its subgroup of special unitary matrices (i.e., unitaries that have determinant $1$) is denoted $\SUnitaries(N)$.
The \textit{Haar measure over $\Unitaries(N)$} is the unique probability distribution over $\Unitaries(N)$ that is invariant under left- or right-multiplication by any $U \in \Unitaries(N)$.
We sometimes informally call this distribution simply ``the Haar measure'' or ``a Haar-random unitary''.
The \textit{Haar measure over $\SUnitaries(N)$} is defined analogously.

An (approximate) unitary $t$-design is, informally speaking, a distribution over unitary transformations that mimics the first $t$ moments of the Haar measure.
Several definitions of approximate designs appear in the literature, but in this work we always use the following multiplicative/relative error definition introduced by Brand\~ao, Harrow, and Horodecki~\cite{BHH16-designs}:

\begin{definition}[{\cite{BHH16-designs}}]
\label{def:unitary_t_design}
A distribution $\UnitaryDist$ over $\Unitaries(N)$ is a \emph{$\delta$-approximate unitary $t$-design} if:
\[
(1 - \delta)\E_{U \sim \Unitaries(N)} \left[(U \cdot U^\dagger)^{\otimes t}\right] \preceq \E_{U \sim \UnitaryDist} \left[(U \cdot U^\dagger)^{\otimes t}\right] \preceq (1 + \delta)\E_{U \sim \Unitaries(N)} \left[(U \cdot U^\dagger)^{\otimes t}\right],
\]
where $U \sim \Unitaries(N)$ indicates that $U$ is sampled from the Haar measure over $\Unitaries(N)$, $U \cdot U^\dagger$ denotes the superoperator that maps a density matrix $\rho$ to $U\rho U^\dagger$, and $A \preceq B$ means that $B - A$ is a completely positive map.
\end{definition}

The unitary Haar measure also yields an induced distribution over quantum states: if $\UnitaryDist$ is the Haar measure over $\Unitaries(2^n)$ and $U \sim \UnitaryDist$, then for every $n$-qubit state $\ket{\psi}$, $U\ket{\psi}$ has the same distribution.
In a common abuse of terminology, we call this distribution the \textit{Haar measure over $n$-qubit states} (or in informal language, ``a Haar-random state'').

An equivalent and at times more convenient formulation of the Haar measure over $n$-qubit states is in terms of normalizing a random complex Gaussian vector.
We define the \textit{Gaussian $n$-qubit state distribution} to sample a $2^n$-dimensional vector $\ket{\psi}$ of the form 
\begin{equation}
\label{eq:def_gaussian_state}
\ket{\psi} \coloneqq \sum_{x \in \{0,1\}^n} (\alpha_x + i\beta_x )\ket{x}
\end{equation}
in which $\alpha_x$ and $\beta_x$ are each independent mean-zero normal random variables with variance $\frac{1}{2^{n+1}}$.
Then $\frac{\ket{\psi}}{\norm{\ket{\psi}}}$ is distributed according to the Haar measure over $n$-qubit states.
Note that a random Gaussian $n$-qubit state will have norm exponentially close to $1$ with overwhelming probability because of concentration of measure.
Thus, Gaussian states are a practically useful approximation to Haar-random states, thanks to the independence between coordinates.
Looking ahead, we will often use the two distributions as drop-in replacements for one another, though always with formal proof that such substitution is mathematically valid (see \Cref{prop:haar_gaussian_equivalence} in particular).

\section{Classical Lower Bound for DXHOG}
\label{sec:classical_lb}
\subsection{Lower Bound Proof Ideas}
\label{sec:proof_ideas}

Consider a randomized classical communication protocol for $\eps$-DXHOG with respect to a distribution $\StateDist$ over Alice's state and an independent distribution $\UnitaryDist$ over Bob's unitary.
We may assume without loss of generality that any optimal protocol for $\eps$-DXHOG is fully deterministic, making no use of randomness, because a randomized protocol is merely a convex combination (probabilistic mixture) of deterministic protocols.
So, we assume that the communication protocol takes the following form: given input $\ket{\psi}$, Alice sends Bob a string $x = f(\ket{\psi}) \in \{0,1\}^m$ for some deterministic function $f$.
Then, given $x$ and the unitary $U$, Bob outputs a string $z = g_x(U) \in \{0,1\}^n$ for some deterministic function $g_x$.

For convenience, given a distribution $\UnitaryDist$ over the $n$-qubit unitary group $\Unitaries(2^n)$ and a function $g: \Unitaries(2^n) \to \{0,1\}^n$, we define the following matrix that appears repeatedly in calculations:
\begin{equation}
\label{eq:def_MgU}
M(g; \UnitaryDist) \coloneqq \E_{U \sim \UnitaryDist}\mbracket{U^\dagger \ket{g(U)} \bra{g(U)} U}.
\end{equation}
Here, even though $U^\dagger$ acts on the left and $U$ on the right, it will result in the correct ordering of the unitaries when used in computations in \Cref{sec:norms}.
Because this matrix is an average of rank-$1$ projectors, it is always trace-$1$ and positive semidefinite.
Observe that the $\FXEB$ achieved by such a protocol can be upper bounded in terms of these matrices:
\begin{align}
\E_{\ket{\psi} \sim \StateDist, U \sim \UnitaryDist}\mbracket{\abs{\braket{g_{f(\ket{\psi})}(U)|U|\psi}}^2} &=
\E_{\ket{\psi} \sim \StateDist}\mbracket{ \braket{\psi|M(g_{f(\ket{\psi})}; \UnitaryDist)|\psi}}\nonumber\\
&\le \E_{\ket{\psi} \sim \StateDist}\mbracket{ \max_{x \in \{0,1\}^m}\braket{\psi|M(g_x; \UnitaryDist)|\psi}} \label{eq:avg_over_max}
\end{align}
This can be understood as showing that, in any optimal protocol, Alice's function $f$ is fully determined by Bob's collection of functions $\{g_x : x \in \{0,1\}^m\}$.
In particular, given $\ket{\psi}$, we can assume without loss of generality that Alice sends $f(\ket{\psi}) = x$ for whichever $x$ maximizes $\braket{\psi|M(g_{x}; \UnitaryDist)|\psi}$.
In this case, the bound in \Cref{eq:avg_over_max} is actually an equality.

To prove that any classical protocol for DXHOG requires exponential communication, we will bound \Cref{eq:avg_over_max} by showing that for any choice of $\{g_x\}$ and for any fixed $x \in \{0,1\}^m$, $\braket{\psi|M(g_{x}; \UnitaryDist)|\psi}$ is unlikely (over the distribution of $\ket{\psi}$) to deviate significantly from its mean.
Then, we will apply a union bound to upper bound the expectation of the maximization over all $x \in \{0,1\}^m$.

A key step in establishing the first bound is noticing that when $\StateDist$ is Haar-random, the distribution of $\braket{\psi|M(g_{x}; \UnitaryDist)|\psi}$ is (approximately\footnote{In fact, the approximation is exact if $\StateDist$ is the the Gaussian $n$-qubit state distribution instead of Haar-random.
We argue that it suffices to consider the Gaussian case by observing that $\FXEB$ lower bounds for DXHOG under these two distributions are equivalent (\Cref{prop:haar_gaussian_equivalence}).
}) a sum of independent exponential random variables whose means are proportional to the eigenvalues of $M(g_{x}; \UnitaryDist)$.
Using Bernstein-type concentration inequalities for sums of (sub)-exponential random variables~\cite{Ver18-hdp,Pin22-subexp}, we obtain a tail bound on $\braket{\psi|M(g_{x}; \UnitaryDist)|\psi}$ in terms of the Frobenius norm $\norm{M(g;\UnitaryDist)}_F$ and operator norm $\norm{M(g;\UnitaryDist)}_{op}$ of the matrix.
As a consequence, we show that the maximum $\FXEB$ achievable is bounded in terms of these norms:

\begin{restatable}{theorem}{xebthm}
\label{thm:XEB_bound_from_norms}
    Let $\StateDist$ be the Haar measure over $n$-qubit states. 
    Let $\UnitaryDist$ be a distribution over $\Unitaries(2^n)$ with the property that for every $g: \Unitaries(2^n) \to \{0,1\}^n$,
    \[
    \norm{M(g;\UnitaryDist)}_F \le A \qquad\text{and}\qquad \norm{M(g;\UnitaryDist)}_{op} \le B.
    \]
    Then, any $m$-bit classical protocol for $\eps$-DXHOG with respect to $\StateDist$ and $\UnitaryDist$ satisfies
    \[
    \eps \le O\mparen{\max\mbrace{\sqrt{m}A, mB}}.
    \]
\end{restatable}

\Cref{thm:XEB_bound_from_norms} is proved in \Cref{sec:lower_bound}.
Importantly, the proof of \Cref{thm:XEB_bound_from_norms} allows us to compute the constant factors inside of the big-$O$.
The full expression for the bound is complicated but easy for computers to evaluate (see \Cref{sec:bound_summary} for the full statement).

Once \Cref{thm:XEB_bound_from_norms} is shown, proving \Cref{thm:main_informal} reduces to bounding the norms of the matrices $M(g;\UnitaryDist)$ in terms of the unitary distribution $\UnitaryDist$.
We establish such bounds in \Cref{sec:norms} by way of computing the moments of the underlying distribution $\UnitaryDist$.
We obtain the strongest communication bounds when $\UnitaryDist$ is the Haar measure over $\Unitaries(2^n)$, but we also obtain exponential lower bounds for approximate $t$-designs, the Clifford group, and tensor products of single-qubit Clifford unitaries.
See the summary in \Cref{sec:bound_summary} for further comparison and discussion of the various bounds.

\subsection{Concentration Inequalities}
\label{app:subexp}

Here we state some concentration inequalities that will be useful throughout the proofs.
First, a standard bound on the tail probability of a standard Gaussian:

\begin{lemma}%
    \label{prop:erf_bound}
    For any $x > 0$,
    \[
    1 - \erf(x) \le e^{-x^2}.
    \]
\end{lemma}

We will use the next inequality to bound the expected maximum of several random variables in terms of their tail bounds:

\begin{lemma}
\label{lem:max_tails_exp}
Let $X_1,\ldots,X_k$ be (not necessarily independent) random variables that each satisfy $X_i \ge r$ for some $r \in \Reals$.
Suppose there exists a function $p(t)$ such that for all $i$, 
\begin{equation}
\Pr\left[X_i \ge t\right] \le p(t)\label{eq:def_p(t)}.
\end{equation}
Then for any $t^* \ge r$,
\[
\E\left[\max_i X_i \right] \le t^* + k\int_{t^*}^\infty p(t) \mathrm{d}t.
\]
Moreover, if $p$ is decreasing, this bound is minimized by choosing $t^* = p^{-1}(1/k)$.
\end{lemma}

\begin{proof}
Assume without loss of generality that $r = 0$, as otherwise we can shift $X_i$, $p$, and $t^*$ by $r$. Assuming this is the case, we have:
\begin{align*}
\E\left[\max_i X_i \right] &= \int_0^\infty \Pr[\max_i X_i \ge t] \mathrm{d}t\\
&\le t^* + \int_{t^*}^\infty \Pr[\max_i X_i \ge t] \mathrm{d}t && (\Pr[\max_i X_i \ge t] \le 1)\\
&\le t^* + \int_{t^*}^\infty \sum_{i=1}^k \Pr[X_i \ge t] \mathrm{d}t && (\text{Union bound})\\
&\le t^* + k \int_{t^*}^\infty p(t) \mathrm{d}t. &&(\mathrm{\Cref{eq:def_p(t)}})
\end{align*}
The ``moreover'' part follows by simple calculus, as the derivative of the above expression with respect to $t^*$ is $1 - kp(t^*)$.
\end{proof}

Next, we move on to deriving a concentration inequality for sums of independent exponential random variables, based on a bound due to Pinelis~\cite{Pin22-subexp}.
We require the following notation from~\cite{Pin22-subexp}.
A function $f: [0,\infty) \to [0,\infty)$ is called an \textit{Orlicz function} if $f$ is convex, increasing, unbounded, and $f(0) = 0$. %
For an Orlicz function $f$, a random variable $X$, and $\gamma > 0$, define
\begin{equation}
\label{eq:orlicz_norm}
\norm{X}_{f; \gamma} \coloneqq \inf \{a > 0 : \E\mbracket{f(|X| / a)}  \le \gamma\}.
\end{equation}
$\norm{\cdot}_{f; \gamma}$ defines a norm on the vector space of real random variables $X$ for which $\norm{X}_{f; \gamma}$ is finite.\\

Define the Orlicz function:
\[
\psi_{11}(u) \coloneqq e^u - 1 - u.
\]

As noted by Pinelis~\cite{Pin22-subexp}, the significance of the Orlicz norm $\norm{\cdot}_{\psi_{11}; \gamma}$ is that it captures a notion of sub-exponentiality: if $\norm{X}_{\psi_{11}; \gamma} < \infty$, then the tails of $X$ decay at least as fast as those of an exponential random variable with mean proportional to $\norm{X}_{\psi_{11}; \gamma}$.
This Orlicz norm is used in the bound below of Pinelis, which takes a similar form as standard Bernstein-type concentration inequalities for sums of independent sub-exponential random variables, e.g.~\cite{Ver18-hdp}. %
Intuitively speaking, these bounds show that sums of independent sub-exponential random variables concentrate like Gaussian random variables near the mean, and like exponential random variables far from the mean.

\begin{lemma}[{\cite[Theorem $2.\eps$]{Pin22-subexp}}]
\label{lem:pinelis}
Let $X_1,\ldots,X_n$ be independent mean-zero random variables, and let
\[
A^2 \coloneqq \sum_{i=1}^N \norm{X_i}_{\psi_{11}; \gamma}^2 \qquad\text{and}\qquad B \coloneqq \max_{i=1}^N \norm{X_i}_{\psi_{11}; \gamma}.
\]
Then if $A$ and $B$ are finite,
\[
\Pr\mbracket{\sum_{i=1}^N X_i \ge t} \le 
\begin{cases}
    \exp\mparen{-\frac{t^2}{4\gamma A^2}} & t \le \frac{2\gamma A^2}{B},\\
    \exp\mparen{-\mparen{\frac{t}{B} - \frac{\gamma A^2}{B^2}}} & t \ge \frac{2\gamma A^2}{B}.
\end{cases}
\]
\end{lemma}

The bound in \Cref{lem:pinelis} is somewhat more complicated and less intuitive than the standard bounds~\cite{Ver18-hdp}, but we find it preferable because in practice it achieves much sharper constants in the exponent.
To make use of \Cref{lem:pinelis}, we compute the Orlicz norm $\norm{\cdot}_{\psi_{11}; \gamma}$ for exponential random variables.

\begin{lemma}
\label{lem:exp_orlicz_norm}
Fix $a > 1$, and define
\[
\gamma \coloneqq \frac{ae^{1/a}}{a+1} + \frac{2}{e(a^3 - a)} - 1.
\]
Let $Y$ be an exponential random variable with mean $\mu$, and let $X \coloneqq Y - \mu$ (i.e., $X$ is the shift of $Y$ that has mean $0$). Then
\[
\norm{X}_{\psi_{11}; \gamma} = a\mu.
\]
\end{lemma}
\begin{proof}
    An exponential random variable with mean $\mu$ is equivalent to $\mu$ times an exponential random variable with mean $1$. Thus, because $\norm{\cdot}_{\psi_{11}; \gamma}$ is a norm, it suffices to consider the case where $Y$ has mean $\mu = 1$. In that case, we have:
    \begin{align*}
    \E\mbracket{\psi_{11}(|X| / a)} &= \int_{-1}^{0} e^{-(x+1)}\psi_{11}\mparen{-x/a}\mathrm{d}x + \int_{0}^{\infty} e^{-(x+1)}\psi_{11}\mparen{x/a}\mathrm{d}x\\
    &= \int_{-1}^{0} e^{-(x+1)}\mparen{e^{-x/a} - 1 + x/a}\mathrm{d}x + \int_{0}^{\infty} e^{-(x+1)}\mparen{e^{x/a} - 1 - x/a}\mathrm{d}x\\
    &= \mparen{\frac{e^\frac{a+1}{a} - 1}{e\frac{a+1}{a}} + \mparen{\frac{1}{e} - 1} - \frac{1}{ea}} + \mparen{\frac{a}{e(a-1)} - \frac{1}{e} - \frac{1}{ea}}\\
    &= \frac{ae^{1/a}}{a+1} - \frac{a}{e(a+1)} - \frac{2}{ea} + \frac{a}{e(a-1)} - 1\\
    &= \frac{ae^{1/a}}{a+1} + \frac{2}{e(a^3 - a)} - 1.\\
    &= \gamma.
\end{align*}
$a' = a$ is clearly the smallest $a' > 0$ for which $\E\mbracket{f(|X| / a')}  \le \gamma$, because for $a' \le 1$ the above integral does not converge, and otherwise $\E\mbracket{f(|X| / a')}$ is a decreasing function of $a'$. It follows from the definition (\Cref{eq:orlicz_norm}) that $\norm{X}_{\psi_{11}; \gamma} = a$.
\end{proof}

We are now ready to prove the main concentration inequality that we will use for exponential random variables.
Compared to \Cref{lem:pinelis}, it specializes to the case of exponential random variables, and operates under the weaker assumption that $A$ and $B$ are merely upper bounds on (rather than exact values of) the quantities of interest.

\begin{lemma}
\label{lem:exp_sums}
Fix $a > 1$, and define
\[
\gamma \coloneqq \frac{ae^{1/a}}{a+1} + \frac{2}{e(a^3 - a)} - 1.
\]
Let $X_1,\ldots,X_N$ be independent exponentially-distributed random variables with means $\mu_1,\ldots,\mu_N$, and suppose that
\[
\sum_{i=1}^N \mu_i^2 \le A^2 \qquad\text{and}\qquad \max_{i=1}^N \mu_i \le B.
\]
Then
\[
\Pr\mbracket{\sum_{i=1}^N X_i \ge t + \sum_{i=1}^N \mu_i} \le 
\begin{cases}
    \exp\mparen{-\frac{t^2}{4\gamma a^2 A^2}} & t \le \frac{2\gamma a A^2}{B},\\
    \exp\mparen{-\mparen{\frac{t}{aB} - \frac{\gamma A^2}{B^2}}} & t \ge \frac{2\gamma a A^2}{B}.
\end{cases}
\]
\end{lemma}

\begin{proof}
    Let $X'_i \coloneqq X_i - \mu_i$ be the shift of $X_i$ with mean $0$.
    Define
    \[
    A' \coloneqq \sqrt{\sum_{i=1}^N \mu_i^2} \qquad\text{and}\qquad B' \coloneqq \max_{i=1}^N \mu_i,
    \]
    so that $A'^2 \le A^2$ and $B' \le B$.
    \Cref{lem:exp_orlicz_norm} shows that $\norm{X'_i}_{\psi_{11}; \gamma} = a\mu_i$, and plugging this into \Cref{lem:pinelis} gives:
    \begin{equation}
    \label{eq:A'_vs_A}
    \Pr\mbracket{\sum_{i=1}^N X_i \ge t + \sum_{i=1}^N \mu_i} \le 
    \begin{cases}
        \exp\mparen{-\frac{t^2}{4\gamma a^2 A'^2}} & t \le \frac{2\gamma a A'^2}{B'},\\
        \exp\mparen{-\mparen{\frac{t}{aB'} - \frac{\gamma A'^2}{B'^2}}} & t \ge \frac{2\gamma a A'^2}{B'}.
    \end{cases}
    \end{equation}
    Our only remaining goal, then, is to show that \Cref{eq:A'_vs_A} still holds when we replace $A'$ by $A$ and $B'$ by $B$.
    To see that this is possible, we observe that the right side of \Cref{eq:A'_vs_A} is equivalent to:\footnote{Expressing the piecewise function in terms of this minimization, 
    as done in Pinelis' proof~\cite[Theorem 2]{Pin22-subexp}, is convenient for establishing the validity of the replacement $A' \to A$ and $B' \to B$ in this context.}
    \[
    \min_{0 \le h \le \frac{1}{aB'}} \exp\mparen{-ht + h^2\gamma a^2A'^2},
    \]
    which is at most
    \[
    \min_{0 \le h \le \frac{1}{aB}} \exp\mparen{-ht + h^2\gamma a^2A^2} = \begin{cases}
    \exp\mparen{-\frac{t^2}{4\gamma a^2 A^2}} & t \le \frac{2\gamma a A^2}{B},\\
    \exp\mparen{-\mparen{\frac{t}{aB} - \frac{\gamma A^2}{B^2}}} & t \ge \frac{2\gamma a A^2}{B}.
\end{cases}\qedhere
    \]
\end{proof}

\subsection{Proof of \texorpdfstring{\Cref{thm:XEB_bound_from_norms}}{Main Theorem}}
\label{sec:lower_bound}

This section is dedicated to the proof of \Cref{thm:XEB_bound_from_norms}.
Within the proof, we will find it more convenient to work with the Gaussian $n$-qubit state distribution instead of the Haar measure over $n$-qubit states.
The next proposition establishes the validity of such substitution: lower bounds for DXHOG under these two distributions are in fact equivalent.
(Note that DXHOG remains well-defined even if the state $\ket{\psi} \sim \StateDist$ is not normalized.)

\begin{proposition}
    \label{prop:haar_gaussian_equivalence}
    Let $\StateDist$ be the Gaussian $n$-qubit state distribution, let $\overline{\StateDist}$ be the Haar measure over $n$-qubit states, and let $\UnitaryDist$ be any distribution over $\Unitaries(2^n)$.
    Then every $m$-bit classical protocol for DXHOG with respect to $(\StateDist, \UnitaryDist)$ achieves $\FXEB$ at most $\eps$ \emph{if and only if} every $m$-bit classical protocol for DXHOG with respect to $(\overline{\StateDist}, \UnitaryDist)$ achieves $\FXEB$ at most $\eps$.
\end{proposition}

\begin{proof}
    We prove the forward implication by showing that every protocol achieving $\FXEB = \eps$ with respect to the normalized distribution $(\overline{\StateDist}, \UnitaryDist)$ implies the existence of a protocol for the unnormalized distribution $(\StateDist, \UnitaryDist)$ achieving the same $\FXEB = \eps$.
    Suppose that the protocol $P$ for the normalized distribution $(\overline{\StateDist}, \UnitaryDist)$ computes a function $z = P(\ket{\psi}, U)$.
    Define a protocol $Q$ that first computes $\ket{\overline{\psi}} \coloneqq \frac{\ket{\psi}}{\norm{\ket{\psi}}}$ and then outputs $z = Q(\ket{\psi}, U) \coloneqq P\mparen{\ket{\overline{\psi}}, U}$ (i.e., Alice normalizes her input state and then runs $P$).
    Then this protocol $Q$ satisfies
    \begin{align*}
        \E_{\ket{\psi} \sim \StateDist, U \sim \UnitaryDist}\mbracket{\abs{\braket{Q(U, \ket{\psi})|U|\psi}}^2}
        &= \E_{\ket{\psi} \sim \StateDist, U \sim \UnitaryDist}\mbracket{\abs{\braket{P(\ket{\overline{\psi}}, U)|U|\overline{\psi}}}^2 \norm{\ket{\psi}}^2}\\
        &= \E_{\ket{\psi} \sim \StateDist, U \sim \UnitaryDist}\mbracket{\abs{\braket{P(\ket{\overline{\psi}}, U)|U|\overline{\psi}}}^2}\E_{\ket{\psi} \sim \StateDist}\mbracket{\norm{\ket{\psi}}^2}\\
        &= \E_{\ket{\overline{\psi}} \sim \overline{\StateDist}, U \sim \UnitaryDist}\mbracket{\abs{\braket{P(\ket{\overline{\psi}}, U)|U|\overline{\psi}}}^2}\E_{\ket{\psi} \sim \StateDist}\mbracket{\norm{\ket{\psi}}^2}\\
        &= \E_{\ket{\overline{\psi}} \sim \overline{\StateDist}, U \sim \UnitaryDist}\mbracket{\abs{\braket{P(\ket{\overline{\psi}}, U)|U|\overline{\psi}}}^2}\E_{\ket{\psi} \sim \StateDist}\mbracket{\sum_{x \in \{0,1\}^n} a_x^2 + b_x^2}\\
        &= \E_{\ket{\overline{\psi}} \sim \overline{\StateDist}, U \sim \UnitaryDist}\mbracket{\abs{\braket{P(\ket{\overline{\psi}}, U)|U|\overline{\psi}}}^2}\\
        &= \frac{1 + \eps}{2^n}.
    \end{align*}
    The first line substitutes the definition of $\ket{\overline{\psi}}$, the second line uses the independence between $\ket{\overline{\psi}}$ and $\norm{\ket{\psi}}$, the third line observes that letting $\ket{\overline{\psi}}$ be the normalization of $\ket{\psi} \sim \StateDist$ is the same as sampling $\ket{\overline{\psi}} \sim \overline{\StateDist}$ directly,
    the fourth and fifth lines apply the definition of $\StateDist$ (\Cref{eq:def_gaussian_state}),
    and the last line is by assumption that $P$ achieves $\FXEB = \eps$.
    This establishes the first part of the proof.

    For the reverse implication, suppose that every $m$-bit classical protocol for DXHOG with respect to the normalized $(\overline{S}, \UnitaryDist)$ achieves $\FXEB$ at most $\eps$.
    Now consider an $m$-bit classical protocol $Q(\ket{\psi}, U)$ for DXHOG with respect to the unnormalized $(S, \UnitaryDist)$.
    From $Q$, one may define an infinite family $\{P_r\}_{r \ge 0}$ of $m$-bit classical \textit{normalized} protocols by evaluating $Q$ on states of norm $r$:
    \[
    P_r(\ket{\psi}, U) \coloneqq Q(r\ket{\psi}, U)
    \]
    By assumption, each $P_r$ achieves $\FXEB$ at most $\eps$ with respect to $(\overline{S}, \UnitaryDist)$.
    This allows us to upper bound the performance of $Q$ (again using the notation $\ket{\overline{\psi}} \coloneqq \frac{\ket{\psi}}{\norm{\ket{\psi}}}$):
    \begin{align*}
    \E_{\ket{\psi} \sim \StateDist, U \sim \UnitaryDist}\mbracket{\abs{\braket{Q(U, \ket{\psi})|U|\psi}}^2}
    &= \E_{\ket{\psi} \sim \StateDist, U \sim \UnitaryDist}\mbracket{\abs{\braket{P_{\norm{\ket{\psi}}}(U, \ket{\overline{\psi}})|U|\psi}}^2}\\
    &= \E_{\ket{\psi} \sim \StateDist, U \sim \UnitaryDist}\mbracket{\abs{\braket{P_{\norm{\ket{\psi}}}(U, \ket{\overline{\psi}})|U|\overline{\psi}}}^2 \norm{\ket{\psi}}^2}\\
    &= \E_{\norm{\ket{\psi}} \sim \norm{\StateDist}}\mbracket{\E_{\ket{\overline{\psi}} \sim \overline{\StateDist}, U \sim \UnitaryDist}\mbracket{\abs{\braket{P_{\norm{\psi}}(U, \ket{\overline{\psi}})|U|\overline{\psi}}}^2} \norm{\ket{\psi}}^2}\\
    &\le \E_{\ket{\psi} \sim \StateDist}\mbracket{\frac{1 + \eps}{2^n} \norm{\ket{\psi}}^2}\\
    &= \frac{1 + \eps}{2^n} \E_{\ket{\psi} \sim \StateDist}\mbracket{\sum_{x \in \{0,1\}^n} a_x^2 + b_x^2}\\
    &= \frac{1 + \eps}{2^n}.
    \end{align*}
    The first two lines substitute the definition of $P_{\norm{\ket{\psi}}}$ and $\ket{\overline{\psi}}$, the third line observes that sampling $\ket{\psi} \sim \StateDist$ can be accomplished by independently sampling its radius $\norm{\ket{\psi}}$ (which is the meaning of $\norm{\ket{\psi}} \sim \norm{\StateDist}$) and normalization $\ket{\overline{\psi}}$, the fourth line uses the assumed bound on the $\FXEB$ of $P_{\norm{\ket{\psi}}}$, and the last two lines apply the definition of $\StateDist$ (\Cref{eq:def_gaussian_state}).
    This completes the proof in both directions.
\end{proof}

We now move onto the proof of \Cref{thm:XEB_bound_from_norms}, which is restated here for convenience:

\xebthm*

\begin{proof}
    In light of \Cref{prop:haar_gaussian_equivalence}, it suffices to instead consider the case where $\StateDist$ is the Gaussian $n$-qubit state distribution, which we assume henceforth.

    Let the communication protocol be parameterized by Bob's functions $\{g_x : x \in \{0,1\}^m\}$ as in \Cref{sec:proof_ideas}.
    In light of \Cref{eq:avg_over_max}, the matrices $M(g_x; \UnitaryDist)$ must satisfy
    \begin{equation}
    \label{eq:eps_at_most}
    \E_{\ket{\psi} \sim \StateDist} \mbracket{2^n \max_x \braket{\psi|M(g_x; \UnitaryDist)|\psi} - 1} \ge \eps,
    \end{equation}
    so we turn to upper bounding the left side of the above.    
    Observe that $M(g_x; \UnitaryDist)$ is positive semidefinite with trace $1$, because it is a convex combination of rank-1 projectors.
    Consider an eigendecomposition
    \[
    M(g_x; \UnitaryDist) = \sum_{i=1}^{2^n} \mu_i \ketbra{v_i}
    \]
    with eigenvalues $\{\mu_i\}$ and unit-norm eigenvectors $\{\ket{v_i}\}$.
    By the unitary invariance of $\StateDist$, $\braket{\psi|M(g_x; \UnitaryDist)|\psi}$ has the same distribution as the random variable
    \[
    \bra{\psi} \mparen{\sum_{i=1}^{2^n} \mu_i \ketbra{i} }\ket{\psi} = \sum_{i=1}^{2^n} \mu_i (a_i^2 + b_i^2),
    \]
    where $\{\ket{i}\}$ indexes the strings in $\{0,1\}^n$, and $a_i$ and $b_i$ are the real and imaginary components of $\braket{i|\psi}$.
    But this is simply a sum of of independent exponentially-distributed random variables with means $\frac{\mu_1}{2^n}, \ldots, \frac{\mu_{2^n}}{2^n}$, because the squared norm of a complex Gaussian random variable is an exponential random variable.
    Hence, \Cref{lem:exp_sums} and the assumed bounds on the Frobenius and operator norms of $M(g_x; \UnitaryDist)$ imply that for any $a > 1$ and
    \[
    \gamma \coloneqq \frac{ae^{1/a}}{a+1} + \frac{2}{e(a^3 - a)} - 1,
    \]
    we have
    \begin{equation}
    \label{eq:blah}
    \Pr_{\ket{\psi} \sim \StateDist}\mbracket{2^n\braket{\psi|M(g_x; \UnitaryDist)|\psi} \ge t + 1} \le p(t) \coloneqq 
    \begin{cases}
    \exp\mparen{-\frac{t^2}{4\gamma a^2 A^2}} & t \le \frac{2\gamma a A^2}{B},\\
    \exp\mparen{-\mparen{\frac{t}{aB} - \frac{\gamma A^2}{B^2}}} & t \ge \frac{2\gamma a A^2}{B}.
    \end{cases}
    \end{equation}
    We may pick $a$ to be any constant (say, $a$ = 2) so that $\gamma$ is also a constant.
    Set
    \[
    t^* \coloneqq \begin{cases}
    \sqrt{\ln(2)m \cdot 4\gamma a^2 A^2} & m \le \frac{\gamma A^2}{\ln(2) B^2},\\
        aB (\ln(2) m  + \gamma A^2 / B^2) & m > \frac{\gamma A^2}{ \ln(2) B^2},
    \end{cases}
    \]
    so that $p(t^*) = \frac{1}{2^m}$. Then \Cref{lem:max_tails_exp} combined with the bound from \Cref{eq:eps_at_most} gives
    \begin{align*}
        \eps \le \E_{\ket{\psi} \sim \StateDist} \mbracket{2^n \max_x \braket{\psi|M(g_x; \UnitaryDist)|\psi} - 1} \le t^* + 2^m \int_{t^*}^\infty p(t)  \mathrm{d}t.
    \end{align*}
    We bound $\eps$ by breaking into cases depending on $m$. If $m \le \frac{\gamma A^2}{\ln(2) B^2}$ then
    \begin{align}
    \eps
    &\le t^* + 2^m \int_{t^*}^{\frac{2\gamma aA^2}{B}} \exp\mparen{-\frac{t^2}{4\gamma a^2A^2}}\mathrm{d}t + 2^m \int_{\frac{2\gamma aA^2}{B}}^{\infty} \exp\mparen{-\mparen{\frac{t}{aB} - \frac{\gamma A^2}{B^2}}}\mathrm{d}t\nonumber\\
    &= t^* + 2^m \sqrt{\pi \gamma} aA \mparen{\erf\mparen{\frac{\sqrt{\gamma} A}{B}} - \erf\left(\frac{t^*}{2\sqrt{\gamma} aA}\right)} + 2^{m} aB\exp\left(- \frac{\gamma A^2}{B^2}\right)\label{eq:frob_bound_exact}\\
    &\le t^* + 2^m \sqrt{\pi \gamma} aA \mparen{1 - \erf\mparen{\frac{t^*}{2\sqrt{\gamma}aA}}} + aB\nonumber\\
    &\le t^* + 2^m \sqrt{\pi \gamma} aA \exp\mparen{-\frac{t^{*2}}{4\gamma a^2 A^2}} + aB\nonumber\\
    &= t^* + \sqrt{\pi \gamma}aA + aB\nonumber\\
    &= \mparen{2\sqrt{\ln(2) m} + \sqrt{\pi \gamma}}aA + aB\nonumber\\
    &\le O\mparen{\sqrt{m}A + B}\nonumber\\
    &\le O\mparen{\max\mbrace{\sqrt{m}A, mB}},\label{eq:xeb_bound_m_small}
    \end{align}
    where in the second line we apply basic calculus, in the third line we use the bound $\erf(x) \le 1$ for all $x$ along with $\exp(-\gamma A^2/B^2) \le \exp(-m \ln(2)) \le 2^{-m}$, in the fourth line we apply \Cref{prop:erf_bound}, and in the fifth and sixth lines we substitute the definition of $t^*$.

    In the complementary case, we suppose that $m > \frac{\gamma A^2}{\ln(2) B^2}$, which yields the bound
    \begin{align}
    \eps
    &\le t^* + 2^m \int_{t^*}^{\infty} \exp\mparen{-\mparen{\frac{t}{aB} - \frac{\gamma A^2}{B^2}}}\mathrm{d}t \nonumber\\
    &= t^* + 2^m aB \exp\mparen{-\mparen{\frac{t^*}{aB} - \frac{\gamma A^2}{B^2}}}\nonumber\\
    &= t^* + aB\label{eq:op_bound_exact}\\
    &= aB(\ln(2) m + \gamma A^2 / B^2 + 1)\nonumber\\
    &\le O\mparen{mB}\nonumber\\
    &\le O\mparen{\max\mbrace{\sqrt{m}A, mB}},\label{eq:xeb_bound_m_large}
    \end{align}
    where in the second line we apply basic calculus and in the third and fourth lines we substitute the definition of $t^*$.

    Combining \Cref{eq:xeb_bound_m_small,eq:xeb_bound_m_large} yields $\eps \le O\mparen{\max\mbrace{\sqrt{m}A, mB}}$, as desired.
\end{proof}

A simple rearrangement of parameters yields the following:

\begin{corollary}
    \label{cor:classical_lower_bound}
    Let $\StateDist$ be the Haar measure over $n$-qubit states. 
    Let $\UnitaryDist$ be a distribution over $\Unitaries(2^n)$ with the property that for every $g: \Unitaries(2^n) \to \{0,1\}^n$,
    \[
    \norm{M(g;\UnitaryDist)}_F \le A \qquad\text{and}\qquad \norm{M(g;\UnitaryDist)}_{op} \le B.
    \]
    Then, any classical protocol for $\eps$-DXHOG with respect to $\StateDist$ and $\UnitaryDist$ must use at least
    \[
    m = \Omega\mparen{\min\mbrace{
    \frac{\eps^2}{A^2}, \frac{\eps}{B}
    }}
    \]
    bits of communication.
\end{corollary}

\subsection{Bounding Matrix Norms}
\label{sec:norms}

In this section, we bound the Frobenius and operator norms of the matrices $M(g;\UnitaryDist)$ from \Cref{eq:def_MgU} for several natural unitary ensembles $\UnitaryDist$. To establish these bounds, we first prove a pair of lemmas that give simpler bounds on these norms.

\begin{lemma}
\label{lem:frob_norm_bound}
    Let $\UnitaryDist$ be a distribution over $\Unitaries(2^n)$ and let $g: \Unitaries(2^n) \to \{0,1\}^n$. Then\footnote{Note that $\max_{z_1, z_2 \in \{0,1\}^n} \abs{\braket{z_1|U_1U_2^\dagger|z_2}}^2$ is simply the largest squared magnitude of any entry in $U_1U_2^\dagger$.}
    \[\norm{M(g; \UnitaryDist)}_F^2 \le \E_{U_1,U_2 \sim \UnitaryDist}\mbracket{\max_{z_1, z_2 \in \{0,1\}^n} \abs{\braket{z_1|U_1U_2^\dagger|z_2}}^2 }.\]
\end{lemma}

\begin{proof}
    \begin{align*}
        \norm{M(g; \UnitaryDist)}_F^2 &= \Tr\mparen{M(g; \UnitaryDist)^2}\\
        &= \Tr\mparen{\E_{U_1 \sim \UnitaryDist}\mbracket{U_1^\dagger \ket{g(U_1)} \bra{g(U_1)} U_1} \E_{U_2 \sim \UnitaryDist}\mbracket{U_2^\dagger \ket{g(U_2)} \bra{g(U_2)} U_2}} && (\text{Definition from \Cref{eq:def_MgU}})\\
        &= \E_{U_1, U_2 \sim \UnitaryDist} \mbracket{\abs{\braket{g(U_1)|U_1U_2^\dagger|g(U_2)}}^2}\\
        &\le \E_{U_1, U_2 \sim \UnitaryDist} \mbracket{\max_{z_1,z_2 \in \{0,1\}^n} \abs{\braket{z_1|U_1U_2^\dagger|z_2}}^2}. && (\text{Replacing } g(U_i) \text{ with max over } z_i) &&\qedhere
    \end{align*}
\end{proof}

\begin{lemma}
\label{lem:op_norm_bound}
    Let $\UnitaryDist$ be a distribution over $\Unitaries(2^n)$ and let $g: \Unitaries(2^n) \to \{0,1\}^n$. Then
    \[\norm{M(g; \UnitaryDist)}_{op} \le \max_{\ket{\psi}} \E_{U \sim \UnitaryDist} \mbracket{\max_{z \in \{0,1\}^n} \abs{\braket{\psi|U^\dagger|z}}^2 }\]
\end{lemma}

\begin{proof}
    \begin{align*}
        \norm{M(g;\UnitaryDist)}_{op} &= \max_{\ket{\psi}} \braket{\psi|M(g;\UnitaryDist)|\psi} \\
        &= \max_{\ket{\psi}}\E_{U \sim \UnitaryDist} \mbracket{\abs{\braket{\psi|U^\dagger|g(U)}}^2 } && (\text{Definition from \Cref{eq:def_MgU}})\\
        &\le \max_{\ket{\psi}}\E_{U \sim \UnitaryDist} \mbracket{\max_{z \in \{0,1\}^n} \abs{\braket{\psi|U^\dagger|z}}^2 } &&(\text{Replacing } g(U) \text{ with max over } z). \qedhere
    \end{align*}
\end{proof}

\subsubsection{Clifford Product Unitaries}

In this section, we consider the case in which $\UnitaryDist$ is uniform over tensor products of random single-qubit Clifford unitaries.

\begin{lemma}
\label{lem:product_clifford_frob_norm}
    Let $\UnitaryDist$ be the uniform distribution over $n$-qubit tensor products of single-qubit Clifford unitaries, and let $g: \Unitaries(2^n) \to \{0,1\}^n$. Then
    \[\norm{M(g; \UnitaryDist)}_F^2 \le \mparen{\frac{2}{3}}^n.\]
\end{lemma}

\begin{proof}
For $U \sim \UnitaryDist$, write $U = U_1 \otimes \cdots \otimes U_n$ where each $U_i$ is a single-qubit Clifford unitary. Then:
    \begin{align*}
    \norm{M(g; \UnitaryDist)}_F^2
    &\le \E_{U_1,U_2 \sim \UnitaryDist}\mbracket{\max_{z_1, z_2 \in \{0,1\}^n} \abs{\braket{z_1|U_1U_2^\dagger|z_2}}^2 }\\
    &= \E_{U \sim \UnitaryDist}\mbracket{\max_{z_1, z_2 \in \{0,1\}^n} \abs{\braket{z_1|U|z_2}}^2 }\\
    &= \E_{U \sim \UnitaryDist}\mbracket{\max_{z \in \{0,1\}^n} \abs{\braket{z|U|0^n}}^2 }\\
    &= \E_{U \sim \UnitaryDist}\mbracket{\prod_{i=1}^n \max_{z_i \in \{0,1\}} \abs{\braket{z_i|U_i|0}}^2 }\\
    &= \prod_{i=1}^n \E_{U \sim \UnitaryDist}\mbracket{\max_{z_i \in \{0,1\}} \abs{\braket{z_i|U_i|0}}^2 }\\
    &= \prod_{i=1}^n \frac{1}{6}\mparen{\abs{\braket{0|0}}^2 + \abs{\braket{1|1}}^2 + \abs{\braket{0|+}}^2 + \abs{\braket{0|-}}^2 + \abs{\braket{0|i}}^2 + \abs{\braket{0|{-i}}}^2}\\
    &= \mparen{\frac{2}{3}}^n.
\end{align*}
Above, in the first line we apply \Cref{lem:frob_norm_bound}, in the second line we use the group structure of the Clifford group, and in the third line we use the fact that all nonzero entries in a Clifford unitary $U$ have the same magnitude. The fourth line uses the tensor product structure of the unitary $U$, and the fifth line is valid because the product of random variables within the expectation are all independent. The sixth line uses the fact that $U_i\ket{0}$ is uniformly distributed over the six single-qubit stabilizer states $\{\ket{0}, \ket{1}, \ket{+}, \ket{-}, \ket{i}, \ket{-i}\}$, and the last line is a straightforward computation.
\end{proof}

\begin{corollary}
\label{cor:product_clifford_op_norm}
    Let $\UnitaryDist$ be the uniform distribution over $n$-qubit tensor products of single-qubit Clifford unitaries, and let $g: \Unitaries(2^n) \to \{0,1\}^n$. Then
    \[\norm{M(g; \UnitaryDist)}_{op} \le \mparen{\frac{2}{3}}^{n/2}.\]
\end{corollary}

\begin{proof}
    This follows from \Cref{lem:product_clifford_frob_norm} and the fact that the Frobenius norm upper bounds the operator norm.
\end{proof}

\subsubsection{Clifford Group and Designs}

In this section, we consider the case in which $\UnitaryDist$ is uniform over the Clifford group, or alternatively drawn from an approximate unitary design.

\begin{lemma}
\label{lem:clifford_frob_norm}
    Let $\UnitaryDist$ be the uniform distribution over the $n$-qubit Clifford group and let $g: \Unitaries(2^n) \to \{0,1\}^n$. Then
    \[\norm{M(g; \UnitaryDist)}_F^2 \le \frac{2}{2^n + 1}.\]
\end{lemma}

\begin{proof}
\begin{align*}
    \norm{M(g; \UnitaryDist)}_F^2
    &\le \E_{U_1,U_2 \sim \UnitaryDist}\mbracket{\max_{z_1, z_2 \in \{0,1\}^n} \abs{\braket{z_1|U_1U_2^\dagger|z_2}}^2 }\\
    &= \E_{U \sim \UnitaryDist}\mbracket{\max_{z_1, z_2 \in \{0,1\}^n} \abs{\braket{z_1|U|z_2}}^2 }\\
    &= \E_{U \sim \UnitaryDist}\mbracket{\max_{z \in \{0,1\}^n} \abs{\braket{z|U|0^n}}^2 }\\
    &= \E_{U \sim \UnitaryDist}\mbracket{\E_{z \sim U\ket{0^n}} \abs{\braket{z|U|0^n}}^2 }\\
    &= \E_{U \sim \UnitaryDist}\mbracket{\sum_{z \in \{0,1\}^n} \abs{\braket{z|U|0^n}}^4 }\\
    &= \frac{2}{2^n + 1},
\end{align*}
where in the first line we apply \Cref{lem:frob_norm_bound}, in the second line we use the group structure of the Clifford group, in the third and fourth lines we use the fact that all nonzero entries in a Clifford unitary $U$ have the same magnitude ($z \sim U\ket{0^n}$ denotes that $z$ is sampled by measuring $U\ket{0^n}$ in the computational basis), and in the last line we use the fact that the Clifford group is a $2$-design~\cite{DCEL09-design,Web16-design}.
\end{proof}

To bound $\norm{M(g;\UnitaryDist)}_{op}$ for the Clifford group, we utilize a bound on the moments of output probabilities of random Clifford circuits.
The first three moments of these probabilities are readily computed, as Cliffords form exact unitary 3-designs~\cite{DCEL09-design,Web16-design}.
Although Cliffords do not form exact (or even approximate) unitary 4-designs \cite{Zhu16gracefully}, analytic expressions for higher moments of Cliffords and stabilizers are known \cite{Zhu16gracefully,Gross17SWClifford}, which we can use to prove the following.

\begin{lemma}
\label{lem:clifford_moments}
Let $\UnitaryDist$ be the uniform distribution over the $n$-qubit Clifford group, let $\ket{\psi}$ be an arbitrary $n$-qubit state, and let $\ket{z}$ be a fixed $n$-qubit computational basis state.
Then for $t \in \Naturals$, the $t$-th moments of the output probabilities are bounded by 
    \[
        \E_{U\sim \UnitaryDist} \big[ |\vev{\psi|U|z}|^{2t}\big] \leq \frac{1}{2^n}\prod_{i=0}^{t-2}\frac{2^i+1}{2^n+2^i}\,.
    \]
\end{lemma}

\begin{proof}
    We want to compute the $t$-th moments of the output probabilities of Clifford unitaries. Similar to the Weingarten formula for Haar-random unitaries and expressions for moments of Haar-random states, Schur-Weyl duality for the Clifford group~\cite{Gross17SWClifford,BELMO25-clifford,LOHEB25-clifford} gives us exact expressions for the moments of random stabilizer states. For some $t\geq 1$, the $t$-th moment of the uniform distribution over $n$-qubit stabilizer states, %
    denoted as $\CS$, is given as~\cite[Lemma 10]{LOHEB25-clifford}
    \[
        \E_{\ket{s}\sim \CS} \big[\ketbra{s}^{\otimes t}] = \frac{1}{2^n \prod_{i=0}^{t-2}(2^n+2^i)} \sum_{\Omega \in \mathcal{P}} \Omega,
    \]
    where $\mathcal{P}$ is the set of so-called ``reduced Pauli monomials'', which are operators on $\C^{2^{n t}}$.
    From \cite[Lemma 8]{LOHEB25-clifford}, we know that $|\mathcal{P}| = \prod_{i=0}^{t - 2}(2^i + 1)$ and $\Tr\mparen{\Omega \rho^{\otimes t}} \le 1$ for every $\Omega \in \mathcal{P}$ and $n$-qubit state $\rho$.
    Proceeding, we have
    \begin{align*}
        \E_{U\sim \UnitaryDist} \mbracket{ |\vev{\psi|U|z}|^{2t}}
        &= \E_{\ket{s}\sim \CS} \mbracket{|\vev{\psi|s}|^{2t}}\\
        &= \frac{1}{2^n \prod_{i=0}^{t-2}(2^n+2^i)} \sum_{\Omega \in \mathcal{P}} \bra{\psi}^{\otimes t}\Omega \ket{\psi}^{\otimes t}\\
        &\le \frac{1}{2^n}\prod_{i=0}^{t-2}\frac{2^i+1}{2^n+2^i}\,. \qedhere
    \end{align*}
\end{proof}

This in turn straightforwardly bounds the operator norm:

\begin{lemma}
\label{lem:clifford_op_norm}
    Let $\UnitaryDist$ be the uniform distribution over the $n$-qubit Clifford group and let $g: \Unitaries(2^n) \to \{0,1\}^n$. Then for any $t \in \Naturals$, we have
    \begin{align*}
    \norm{M(g; \UnitaryDist)}_{op} &\le \mparen{\prod_{i=0}^{t-2}\frac{2^i + 1}{2^n + 2^i}}^{1/t}\\
    &\le 2^{(t-1)/2 + n/t - n}.
    \end{align*}
\end{lemma}

\begin{proof}
    For any $n$-qubit state $\ket{\psi}$ we have
    \begin{align*}
        \E_{U \sim \UnitaryDist} \mbracket{\max_{z \in \{0,1\}^n} \abs{\braket{\psi|U^\dagger|z}}^2 }
        &\le \E_{U \sim \UnitaryDist} \mbracket{\mparen{\sum_{z \in \{0,1\}^n} \abs{\braket{\psi|U^\dagger|z}}^{2t}}^{1/t}}\\
        &\le \mparen{\E_{U \sim \UnitaryDist} \mbracket{\sum_{z \in \{0,1\}^n} \abs{\braket{\psi|U^\dagger|z}}^{2t}}}^{1/t},
    \end{align*}
    by the monotonicity of vector norms and Jensen's inequality. We can bound this using the Clifford moments from \Cref{lem:clifford_moments}, via
    \begin{align*}
        \E_{U \sim \UnitaryDist} \mbracket{\sum_{z \in \{0,1\}^n} \abs{\braket{\psi|U^\dagger|z}}^{2t}}
        &= \sum_{z \in \{0,1\}^n} \E_{U \sim \UnitaryDist} \mbracket{ \abs{\braket{\psi|U^\dagger|z}}^{2t}}\\
        &\le \prod_{i=0}^{t-2}\frac{2^i + 1}{2^n + 2^i}
    \end{align*}
    So, by \Cref{lem:op_norm_bound} we conclude that
    \begin{align*}
        \norm{M(g;\UnitaryDist)}_{op} &\le \max_{\ket{\psi}} \E_{U \sim \UnitaryDist} \mbracket{\max_{z \in \{0,1\}^n} \abs{\braket{\psi|U^\dagger|z}}^2 }\\
        &\le \max_{\ket{\psi}} \mparen{\E_{U \sim \UnitaryDist} \mbracket{\sum_{z \in \{0,1\}^n} \abs{\braket{\psi|U^\dagger|z}}^{2t}}}^{1/t}\\
        &\le \mparen{\prod_{i=0}^{t-2}\frac{2^i + 1}{2^n + 2^i}}^{1/t}\\
        &\le \mparen{\prod_{i=0}^{t-2}\frac{2^{i + 1}}{2^n}}^{1/t}\\
        &= 2^{\mparen{\binom{t}{2} - (t-1)n}/t}\\
        &= 2^{(t-1)/2 + n/t - n}.\qedhere
    \end{align*}
\end{proof}

Asymptotically, \Cref{lem:clifford_op_norm} is optimized by choosing $t \approx \sqrt{n}$, which yields the bound $2^{O(\sqrt{n}) - n}$.

The same proof strategy also gives a bound on the operator norm of any approximate unitary design.

\begin{lemma}
\label{lem:design_op_norm}
    Let $\UnitaryDist$ be a $\delta$-approximate $t$-design and let $g: \Unitaries(2^n) \to \{0,1\}^n$. Then
    \begin{align*}
    \norm{M(g; \UnitaryDist)}_{op}
    &\le \mparen{\frac{(1 + \delta)t!}{(2^n + 1)(2^n + 2) \cdots (2^n + t - 1)}}^{1/t}\\
    &\le \frac{\delta + t}{2^{n(t-1)/t}}.
    \end{align*}
\end{lemma}

\begin{proof}
    Following the same lines as the proof of \Cref{lem:clifford_op_norm}, we know that the moments of an approximate design satisfy:
    \begin{align*}
        \E_{U \sim \UnitaryDist} \mbracket{\sum_{z \in \{0,1\}^n} \abs{\braket{\psi|U^\dagger|z}}^{2t}}
        &= \sum_{z \in \{0,1\}^n} \E_{U \sim \UnitaryDist} \mbracket{ \abs{\braket{\psi|U^\dagger|z}}^{2t}}\\
        &\le (1 + \delta) 2^n\frac{t!}{(2^n)(2^n + 1) \cdots (2^n + t - 1)}. %
    \end{align*}

    This in turn implies:
    \begin{align*}
        \norm{M(g;\UnitaryDist)}_{op} &\le \max_{\ket{\psi}} \E_{U \sim \UnitaryDist} \mbracket{\max_{z \in \{0,1\}^n} \abs{\braket{\psi|U^\dagger|z}}^2 }\\
        &\le \max_{\ket{\psi}} \mparen{\E_{U \sim \UnitaryDist} \mbracket{\sum_{z \in \{0,1\}^n} \abs{\braket{\psi|U^\dagger|z}}^{2t}}}^{1/t}\\
        &\le \max_{\ket{\psi}} \mparen{(1 + \delta) \E_{U \sim \Unitaries(N)} \mbracket{\sum_{z \in \{0,1\}^n} \abs{\braket{\psi|U^\dagger|z}}^{2t}}}^{1/t}\\
        &= \max_{\ket{\psi}} \mparen{(1 + \delta) 2^n \E_{U \sim \Unitaries(N)} \mbracket{\abs{\braket{\psi|U^\dagger|0^n}}^{2t}}}^{1/t}\\
        &= \mparen{\frac{(1 + \delta)t!}{(2^n + 1)(2^n + 2) \cdots (2^n + t - 1)}}^{1/t}\\
        &\le \frac{(1 + \delta)^{1/t}t}{2^{(t-1)n/t}}\\
        &\le \frac{\delta + t}{2^{n(t-1)/t}}
        .
    \end{align*}
    Above, the third line follows from \Cref{def:unitary_t_design}, because the $t$th moments of a $\delta$-approximate $t$-design are within a $1 \pm \delta$ factor of the Haar moments (see, e.g., \cite[Lemma 22]{Kre21-pseudorandom}).
    The fourth line uses linearity of expectation and the invariance of the Haar measure, because $U^\dagger\ket{z}$ is distributed as a Haar-random state for any $z$.
    The fifth line is a standard Haar moment calculation; the invariance of the Haar measure means that the quantity in parentheses is the same for all $\ket{\psi}$.
\end{proof}

Note that if we choose $t = n$ and $\delta$ a small constant, the bound derived above is $O\mparen{\frac{n}{2^n}}$.

\subsubsection{Haar Measure}

Next we consider the case where $\UnitaryDist$ is the $n$-qubit Haar measure, which allows us to deduce a better upper bound on the operator norm. First, we prove a more general statement about Clifford-invariant unitary ensembles.

\begin{lemma}
\label{lem:clifford_invariant_frob_norm}
    Let $\UnitaryDist$ be any distribution over $\Unitaries(2^n)$ that is invariant under left multiplication by a random $n$-qubit Clifford, and let $g: \Unitaries(2^n) \to \{0,1\}^n$. Then
    \[\norm{M(g; \UnitaryDist)}_F^2 \le \frac{2}{2^n + 1}.\]
\end{lemma}

\begin{proof}
Let $\mathcal{C}$ denote the uniform distribution over the Clifford group, and for a Clifford unitary $C$ denote by $g_U(C) \coloneqq g(CU)$. Then observe that:
\begin{align*}
    M(g;\UnitaryDist)
    &= \E_{U \sim \UnitaryDist}\mbracket{U^\dagger \ket{g(U)} \bra{g(U)} U}\\
    &= \E_{U \sim \UnitaryDist, C \sim \mathcal{C}}\mbracket{(CU)^\dagger \ket{g(CU)} \bra{g(CU)} CU}\\
    &= \E_{U \sim \UnitaryDist}\mbracket{U^\dagger \E_{C \sim \mathcal{C}}\mbracket{C^\dagger \ket{g_U(C)} \bra{g_U(C)} C}U}\\
    &= \E_{U \sim \UnitaryDist}\mbracket{U^\dagger M(g_U ; \mathcal{C})U},
\end{align*}
where in the first and last lines we apply the definition of $M(\cdot)$ (\Cref{eq:def_MgU}), and in the second line we use the Clifford invariance of $\UnitaryDist$. Finally, notice that
\begin{align*}
    \norm{M(g;\UnitaryDist)}_F^2
    &= \norm{\E_{U \sim \UnitaryDist}\mbracket{U^\dagger M(g_U ; \mathcal{C})U}}_F^2\\
    &\le \E_{U \sim \UnitaryDist}\mbracket{\norm{U^\dagger M(g_U ; \mathcal{C})U}_F^2}\\
    &= \E_{U \sim \UnitaryDist}\mbracket{\norm{M(g_U ; \mathcal{C})}_F^2}\\
    &\le \frac{2}{2^n + 1},
\end{align*}
where the second line follows from Jensen's inequality and the convexity of the Frobenius norm, the third line holds by the unitary invariance of the Frobenius norm, and the last line appeals to \Cref{lem:clifford_frob_norm}.
\end{proof}

Because the Haar measure is invariant under multiplication by \textit{any} unitary, and therefore left-Clifford invariant, \Cref{lem:clifford_invariant_frob_norm} implies following:

\begin{corollary}
\label{cor:haar_frob_norm}
    Let $\UnitaryDist$ be the Haar measure over $\Unitaries(2^n)$ and let $g: \Unitaries(2^n) \to \{0,1\}^n$. Then
    \[\norm{M(g; \UnitaryDist)}_F^2 \le \frac{2}{2^n + 1}.\]
\end{corollary}

We next move to bounding the operator norm for the unitary Haar measure.
Obtaining the sharpest possible bound will make use of the following probability calculation, whose proof is reproduced from~\cite{Kab18-simplex}.

\begin{lemma}[{\cite{Kab18-simplex}}]
\label{lem:simplex_max}
    Let $X = (X_1,\ldots,X_N)$ be the vector of a random probability distribution on $N$ variables (i.e., $X$ is a random vector over the $(N-1)$-simplex, or equivalently a Dirichlet distribution $\mathrm{Dir}(1,\ldots,1)$).
    Then
    \[
    \E\mbracket{\max_i X_i} = \frac{H_N}{N},
    \]
    where $H_N \coloneqq \sum_{i=1}^N \frac{1}{i}$ is the $N$th harmonic number.
\end{lemma}

\begin{proof}
    It is well-known that $X$ has the same distribution as
    \[
    \mparen{\frac{Y_1}{Y_1 + \dots + Y_N}, \frac{Y_2}{Y_1 + \dots + Y_N}, \ldots, \frac{Y_N}{Y_1 + \dots + Y_N}},
    \]
    where $Y = (Y_1,\ldots,Y_N)$ is a vector of $N$ independent mean-$1$ exponential random variables.
    A result of R\'enyi~\cite{Ren53-order} shows that the vector $(B_N, B_{N-1},\ldots,B_1)$ of order statistics of $Y$ (meaning $B_i$ is the $i$th-largest coordinate of $Y$) has the same distribution as
    \[
    \mparen{\frac{Z_N}{N}, \frac{Z_N}{N} + \frac{Z_{N-1}}{N-1}, \ldots, \frac{Z_N}{N} + \frac{Z_{N-1}}{N-1} + \dots + \frac{Z_1}{1}},
    \]
    where $Z_1,\ldots,Z_N$ are independent mean-$1$ exponential random variables.
    Hence
    \[
    \E\mbracket{\max_i X_i} = \E\mbracket{\frac{\max_i Y_i}{Y_1 + \dots + Y_N}} = 
    \E\mbracket{\frac{\sum_{i=1}^N Z_i / i}{\sum_{i=1}^N Z_i}}.
    \]
    Notice that, because the variables $Z_i$ are independent and identically distributed, the above expression is invariant under permutations.
    That is, for any permutation $\pi: [N] \to [N]$ we have
    \[
    \E\mbracket{\frac{\sum_{i=1}^N Z_i / i}{\sum_{i=1}^N Z_i}} = \E\mbracket{\frac{\sum_{i=1}^N Z_{\pi(i)} / i}{\sum_{i=1}^N Z_i}}.
    \]
    Averaging over a uniformly random $\pi$ makes the coefficient in front of each $Z_i$ the same, which yields
    \begin{align*}
    \E\mbracket{\max_i X_i}
    &= \E_\pi\mbracket{\frac{\sum_{i=1}^N Z_{\pi(i)} / i}{\sum_{i=1}^N Z_i}}\\
    &= \mparen{\frac{\sum_{i=1}^N 1/i}{N}}\E\mbracket{\frac{\sum_{i=1}^N Z_{i}}{\sum_{i=1}^N Z_i}}\\
    &= \frac{H_N}{N}.\qedhere
    \end{align*}
\end{proof}

From the above lemma, an operator norm bound easily follows.

\begin{lemma}
\label{lem:haar_op_norm}
    Let $\UnitaryDist$ be the Haar measure over $\Unitaries(2^n)$ and let $g: \Unitaries(2^n) \to \{0,1\}^n$. Then
    \[\norm{M(g; \UnitaryDist)}_{op} \le \frac{H_{2^n}}{2^n},\]
    where $H_N \coloneqq \sum_{i=1}^N \frac{1}{i} \le \ln(N) + 1$ is the $N$th harmonic number.
\end{lemma}

\begin{proof}
    \begin{align*}
        \norm{M(g; \UnitaryDist)}_{op}
        &\le \max_{\ket{\psi}} \E_{U \sim \UnitaryDist} \mbracket{\max_{z \in \{0,1\}^n} \abs{\braket{\psi|U^\dagger|z}}^2 }\\
        &= \E_{U \sim \UnitaryDist} \mbracket{\max_{z \in \{0,1\}^n} \abs{\braket{0^n|U^\dagger|z}}^2 }\\
        &= \frac{H_{2^n}}{2^n}.
    \end{align*}
    Above, the first line is a restatement of \Cref{lem:op_norm_bound}, and the second line uses the invariance of the Haar measure---the choice of $\ket{\psi}$ is arbitrary, so we might as well set $\ket{\psi} = \ket{0^n}$. The last line follows from \Cref{lem:simplex_max} and the fact that the measurement distribution of a Haar-random quantum state in any fixed basis is uniform over the unit simplex.
\end{proof}

\section{Classical Upper Bound for DXHOG}\label{sec:classical_ub}

In this section, we obtain a reasonable upper bound on the number of classical bits needed to solve $\eps$-DXHOG.
Although this is not necessary for the interpretation of the experiment, it lets us evaluate the tightness of our lower bounds.
The protocol is most easily described as a \textit{randomized} classical communication protocol, though by convexity, there must exist some fixing of the randomness that yields a deterministic protocol achieving the same $\FXEB$.
We start by describing a protocol for DXHOG in the case when $\UnitaryDist$ is the Haar measure.
We explain later (in \Cref{cor:classical_upper_bound}) how to further randomize this protocol to make it achieve the same $\FXEB$ for \textit{any} measurement distribution $\UnitaryDist$.

The Haar-random protocol is defined as follows.
Assuming Alice and Bob have $m$ bits of communication, they initialize the randomness of the protocol by choosing $2^m$ independent states $\{\ket{\varphi_x}\}_{x \in \{0,1\}^m}$ from the Haar measure over $n$-qubit states.
Alice, given $\ket{\psi} \sim \StateDist$, communicates to Bob the string $x \in \{0,1\}^m$ that maximizes the corresponding fidelity $\abs{\braket{\varphi_x|\psi}}^2$.
Bob, given $x$ and $U \sim \UnitaryDist$, outputs the string $z \in \{0,1\}^n$ that maximizes the inner product $\abs{\braket{z|U|\varphi_x}}^2$ with respect to $U\ket{\varphi_x}$.

We analyze the $\FXEB$ achieved by this protocol in a similar fashion to the lower bound, by using the matrices $M(g;\UnitaryDist)$.
Define $g_x: \Unitaries(2^n) \to \{0,1\}^n$ by
\begin{equation}
\label{def:g_x}
g_x(U) \coloneqq \argmax_{z \in \{0,1\}^n}\  \abs{\braket{z|U|\varphi_x}}^2,
\end{equation}
which is the function Bob evaluates to output $z$.
Observe that after fixing the randomness $\{\ket{\varphi_x}\}_{x \in \{0,1\}^m}$, the state $\ket{\psi}$ given to Alice, and the message $x$ from Alice to Bob, the expected $\FXEB$ averaged over Bob's Haar-random unitary $U \sim \UnitaryDist$ is
\begin{align*}
2^n\E_{U\sim \UnitaryDist}\mbracket{\abs{\braket{z|U|\psi}}^2} - 1
&= 2^n\E_{U\sim \UnitaryDist}\mbracket{\abs{\braket{g_x(U)|U|\psi}}^2} - 1\\
&= 2^n \braket{\psi|M(g_x;\UnitaryDist)|\psi} - 1,
\end{align*}
recalling the definition of $M(g; \UnitaryDist)$ (\Cref{eq:def_MgU}).
Our first lemma shows that this quantity is a linear function of the fidelity $\abs{\braket{\varphi_x|\psi}}^2$:

\begin{lemma}
    \label{lem:fidelity_harmonic}
    Given some $n$-qubit state $\ket{\varphi}$, define $g: \Unitaries(2^n) \to \{0,1\}^n$ by
    \[
    g(U) \coloneqq \argmax_{z \in \{0,1\}^n}\  \abs{\braket{z|U|\varphi}}^2.
    \]
    If $\UnitaryDist$ is the Haar measure over $\Unitaries(2^n)$, then for any $n$-qubit state $\ket{\psi}$,
    \[
    \braket{\psi|M(g; \UnitaryDist)|\psi} = \frac{H_{2^n} - 1}{2^n - 1} \abs{\braket{\varphi|\psi}}^2 + \frac{2^n - H_{2^n}}{2^n(2^n - 1)},
    \]
    where $H_N \coloneqq \sum_{i=1}^N \frac{1}{i}$ is the $N$th harmonic number.
\end{lemma}

\begin{proof}
    Observe that $g$ is invariant under right multiplication by unitaries that fix $\ket{\varphi}$, meaning that for any $V \in \Unitaries(2^n)$ satisfying $V\ket{\varphi} = \ket{\varphi}$, we have $g(UV) = g(U)$.
    It follows that $M(g;\UnitaryDist)$ is also invariant under conjugation by such unitaries $V$:
    \begin{align*}
    V^\dagger M(g; \UnitaryDist) V
    &= V^\dagger \E_{U \sim \UnitaryDist}\mbracket{U^\dagger \ket{g(U)} \bra{g(U)} U}V && (\mathrm{\Cref{eq:def_MgU}})\\
    &= \E_{U \sim \UnitaryDist}\mbracket{V^\dagger U^\dagger \ket{g(U)} \bra{g(U)} UV} && (\text{Linearity of expectation})\\
    &= \E_{U \sim \UnitaryDist}\mbracket{V^\dagger U^\dagger \ket{g(UV)} \bra{g(UV)} UV} && (g(U) = g(UV))\\
    &= \E_{U \sim \UnitaryDist}\mbracket{U^\dagger \ket{g(U)} \bra{g(U)} U} && (\text{Invariance of Haar measure})\\
    &= M(g; \UnitaryDist). && (\mathrm{\Cref{eq:def_MgU}})
    \end{align*}

Because $M(g; \UnitaryDist)$ is invariant under conjugation by \textit{every} such $V$, $M(g; \UnitaryDist)$ must decompose as a linear combination:
\[
M(g; \UnitaryDist) = \alpha \ketbra{\varphi} + \beta I_\perp,
\]
where $\alpha, \beta > 0$ and $I_\perp$ is the identity matrix on the subspace orthogonal to $\ket{\varphi}$.
Since $M(g; \UnitaryDist)$ has trace $1$, we can also express it as a probabilistic mixture of trace-$1$ matrices:
\begin{equation}
M(g; \UnitaryDist) = p\ketbra{\varphi} + (1 - p)\frac{I}{2^n},
\label{eq:varphi_decomposition_M}
\end{equation}
where now $I$ is the identity on the entire $2^n$-dimensional space, and $p \in [0, 1]$.
To compute $p$, it suffices to compute the expectation of $M(g; \UnitaryDist)$ with respect to $\ket{\varphi}$:
\begin{align*}
    p + \frac{(1 - p)}{2^n} &= \braket{\varphi|M(g;\UnitaryDist)|\varphi} && (\mathrm{\Cref{eq:varphi_decomposition_M}})\\
    &= \E_{U \sim \UnitaryDist}\mbracket{\abs{\braket{g(U)|U|\varphi}}^2} && (\mathrm{\Cref{eq:def_MgU}})\\
    &= \E_{U \sim \UnitaryDist}\mbracket{\max_{z \in \{0,1\}^n}\abs{\braket{z|U|\varphi}}^2} && (\text{Definition of $g$})\\
    &= \frac{H_{2^n}}{2^n}. && (\mathrm{\Cref{lem:simplex_max}})
\end{align*}
Therefore,
\[
p = \frac{H_{2^n} - 1}{2^n - 1}.
\]
Using \Cref{eq:varphi_decomposition_M}, it follows that
\[
\braket{\psi|M(g;\UnitaryDist)|\psi} = \frac{H_{2^n} - 1}{2^n - 1} \abs{\braket{\varphi|\psi}}^2 + \frac{2^n - H_{2^n}}{2^n(2^n - 1)},
\]
as desired.
\end{proof}

We can use this to compute the $\FXEB$ achieved by our classical protocol for Haar-random measurements:

\begin{theorem}
    \label{thm:classical_upper_bound}
    Let $\StateDist$ be \emph{any} distribution over $n$-qubit states, and let $\UnitaryDist$ be the Haar measure over $\Unitaries(2^n)$.
    Then with $m$ bits of classical communication, the protocol above solves $\eps$-DXHOG for
    \[
    \eps = (H_{2^n} - 1) \mparen{1 - \frac{2^n}{2^n - 1}\int_0^1 (1 - u^{2^n - 1})^{2^m} \mathrm{d}u},
    \]
    where $H_N \coloneqq \sum_{i=1}^N \frac{1}{i}$ is the $N$th harmonic number.
\end{theorem}

\begin{proof}
    For convenience, let $N = 2^n$ and $M = 2^m$.
    Because Alice's message $x$ is chosen to be the string that maximizes $\abs{\braket{\varphi_x|\psi}}^2$, \Cref{lem:fidelity_harmonic} implies that the protocol achieves $\FXEB$ exactly
    \begin{align}
    N \E_{\substack{\ket{\psi} \sim \StateDist\\ \{\ket{\varphi_x}\}_{x \in \{0,1\}^m} \sim \StateDist}}\mbracket{\braket{\psi|M(g_x; \UnitaryDist)|\psi}} - 1
    &= N \E_{\substack{\ket{\psi} \sim \StateDist\\ \{\ket{\varphi_x}\}_{x \in \{0,1\}^m} \sim \StateDist}}\mbracket{\max_{x \in \{0,1\}^m} \frac{H_{N} - 1}{N - 1} \abs{\braket{\varphi|\psi}}^2 + \frac{N - H_{N}}{N(N - 1)}} - 1\nonumber\\
    &= N \frac{H_{N} - 1}{N - 1} \E_{\substack{\ket{\psi} \sim \StateDist\\ \{\ket{\varphi_x}\}_{x \in \{0,1\}^m} \sim \StateDist}}\mbracket{\max_{x \in \{0,1\}^m} \abs{\braket{\varphi|\psi}}^2} + \frac{1 - H_{N}}{N - 1}\nonumber\\
    &= \frac{H_{N} - 1}{N - 1}\mparen{ N \E_{\substack{\ket{\psi} \sim \StateDist\\ \{\ket{\varphi_x}\}_{x \in \{0,1\}^m} \sim \StateDist}}\mbracket{\max_{x \in \{0,1\}^m} \abs{\braket{\varphi_x|\psi}}^2} - 1}.\label{eq:classical_lb_intermediate}
    \end{align}
    
    We turn to computing the expectation inside of the expression above.
    For any fixed $\ket{\psi}$ (and therefore, for $\ket{\psi} \sim \StateDist$), the random variables $\{\abs{\braket{\varphi_x|\psi}}^2\}$ are i.i.d.\ with probability density function
    \[
    p(\abs{\braket{\varphi_x|\psi}}^2 = t) = \frac{(1 - t)^{N - 2}}{N - 1},
    \]
    i.e., they are $\mathrm{Beta}(1, N - 1)$ random variables.
    Their cumulative distribution function is
    \[
    \Pr[\abs{\braket{\varphi_x|\psi}}^2 \le t] = 1 - (1 - t)^{N - 1}.
    \]
    Because of independence, the variables $\{\ket{\varphi_x}\}$ further satisfy
    \[
    \Pr[\max_{x \in \{0,1\}^m} \abs{ \braket{\varphi_x|\psi}}^2 \le t] = (1 - (1 - t)^{N - 1})^M.
    \]
    Therefore,
    \begin{align*}
    \E\mbracket{\max_{x \in \{0,1\}^m} \abs{ \braket{\varphi_x|\psi}}^2}
    &= \int_0^1 1 - (1 - (1 - t)^{N - 1})^M \mathrm{d}t\\
    &= 1 - \int_0^1 (1 - u^{N - 1})^M \mathrm{d}u. && (u = 1 - t)
    \end{align*}
    Plugging back into \Cref{eq:classical_lb_intermediate}, we obtain:
    \begin{align*}
        \FXEB
        &= \frac{H_{N} - 1}{N - 1}\mparen{ N \mparen{1 - \int_0^1 (1 - u^{N - 1})^M \mathrm{d}u} - 1}\\
        &= (H_N - 1) \mparen{1 - \frac{N}{N - 1}\int_0^1 (1 - u^{N - 1})^M \mathrm{d}u}.\qedhere
    \end{align*}
\end{proof}

By a standard re-randomization trick, we can make $\StateDist$ Haar-random and $\UnitaryDist$ arbitrary:

\begin{corollary}
    \label{cor:classical_upper_bound}
    Let $\StateDist$ be the Haar measure over $n$-qubit states, and let $\UnitaryDist$ be \emph{any} distribution over $\Unitaries(2^n)$.
    Then with $m$ bits of classical communication, there is a protocol that solves $\eps$-DXHOG for
    \[
    \eps = (H_{2^n} - 1) \mparen{1 - \frac{2^n}{2^n - 1}\int_0^1 (1 - u^{2^n - 1})^{2^m} \mathrm{d}u},
    \]
    where $H_N \coloneqq \sum_{i=1}^N \frac{1}{i}$ is the $N$th harmonic number.
\end{corollary}

\begin{proof}
    The idea is for Alice and Bob to self-randomize their inputs to the protocol from \Cref{thm:classical_upper_bound}.
    Before execution of the protocol, Alice and Bob jointly agree on a random unitary $V$ drawn from the Haar measure over $\Unitaries(2^n)$.
    Alice, given input $\ket{\psi} \sim \StateDist$, executes the protocol from \Cref{thm:classical_upper_bound} as if her input were $V\ket{\psi}$.
    Bob, given input $U \sim \Unitaries(N)$, executes the protocol from \Cref{thm:classical_upper_bound} as if his input were $UV^\dagger$.
    By the invariance of the unitary Haar measure, $UV^\dagger$ is distributed as the Haar measure over $\Unitaries(2^n)$.
    Additionally, because $\StateDist$ is Haar-random, $V\ket{\psi}$ is distributed as the Haar measure over $n$-qubit states, and crucially $V\ket{\psi}$ is \textit{independent} of $UV^\dagger$.
    So, $V\ket{\psi}$ and $UV^\dagger$ are valid inputs for the protocol from \Cref{thm:classical_upper_bound} to apply.
    The output of \Cref{thm:classical_upper_bound} is a string $z$ satisfying, in expectation,
    \[
    \frac{1 + \eps}{2^n} = 
    \E_{\ket{\psi} \sim \StateDist, U \sim \UnitaryDist, V \sim \Unitaries(2^n)}\mbracket{\abs{\braket{z|UV^\dagger V|\psi}}^2} = \E_{\ket{\psi} \sim \StateDist, U \sim \UnitaryDist, V \sim \Unitaries(2^n)}\mbracket{\abs{\braket{z|U|\psi}}^2}.
    \]
    The right equality above holds because $V^\dagger V = I$, and this demonstrates that the re-randomization of the protocol achieves $\FXEB$ $\eps$ over $(\StateDist, \UnitaryDist)$.
\end{proof}

The quantity $\eps$ in \Cref{thm:classical_upper_bound} and \Cref{cor:classical_upper_bound} can be computed numerically.
In practice, however, the integrand $\mparen{1 - u^{2^n - 1}}^{2^m}$ is difficult to evaluate with sufficient numerical stability because $1 - u^{2^n - 1}$ can simultaneously be very close to $1$ when the exponent $2^m$ is extraordinarily large.
To circumvent this difficulty, one can instead compute a valid lower bound on $\eps$ by making use of the exponential inequality $1 - x \le e^{-x}$:
\begin{equation}
\label{eq:eps_lower_bound}
\eps \ge (H_{2^n} - 1) \mparen{1 - \frac{2^n}{2^n - 1}\int_0^1 \exp(-2^m u^{2^n - 1}) \mathrm{d}u}.
\end{equation}
The difference between the two sides of \Cref{eq:eps_lower_bound} becomes negligible as $m$ increases, because $(1 - x)^M \approx e^{-xM}$ for $x \in [0, 1]$ and sufficiently large $M$.

To give some intuition for the asymptotic behavior of the bound, note that the random variables $\abs{\braket{\varphi_x|x}}^2$ are close to exponential random variables with mean $\frac{1}{2^n}$.
Thus, when $m$ is not too large, the expectation of their maximum is approximately~\cite{Ren08-qkd}
\[
\E\mbracket{\max_{x \in \{0,1\}^m} \abs{ \braket{\varphi_x|\psi}}^2} \approx \frac{H_{2^m}}{2^n} \approx \frac{\ln(2) m}{2^n}.
\]
Substituting into the proof of \Cref{thm:classical_upper_bound}, one expects
\[
\eps \approx \frac{\ln(2)^2 mn}{2^n}
\]
for small $m$.
However, this approximation cannot be correct for large $m$, because the largest $\FXEB$ achievable for Haar-random states is $H_{2^n} - 1 \approx \ln(2) n$, which is obtained by always outputting the string $z$ that maximizes $\abs{\braket{z|U|\psi}}^2$ (cf.\ \Cref{lem:haar_op_norm}).
Indeed, the $\eps$ in \Cref{thm:classical_upper_bound} and \Cref{cor:classical_upper_bound} asymptotes towards $H_{2^n} - 1$ for large $m$.
One obtains a better bound for large $m$ by observing that the integrand $\mparen{1 - u^{2^n - 1}}^{2^m}$ drops sharply from $\sim 1$ to $\sim 0$ in a small neighborhood of $u = 2^{-m/(2^n - 1)} \approx 2^{-m/2^n}$.
In the regime of large $m$, one therefore expects
\[
\eps \approx (H_{2^n} - 1)\mparen{1 - \frac{2^n}{2^n - 1}2^{-m/2^n}} \approx (H_{2^n} - 1)\mparen{1 - 2^{-m/2^n}}.
\]
In either case, setting $m = \frac{2^n}{n}$ should suffice to obtain $\eps \ge \Omega(1)$, because
\[
(H_{2^n} - 1)\mparen{1 - 2^{-{1/n}}} = (H_{2^n} - 1)\mparen{1 - e^{-{\ln(2)/n}}} \approx (H_{2^n} - 1)\frac{\ln(2)}{n} \approx \ln(2)^2.
\]

\section{Summary of Classical Communication Complexity Bounds}
\label{sec:bound_summary}

In \Cref{sec:classical_lb}, we showed an exponential lower bound on the classical communication complexity of DXHOG for four different measurement ensembles, by way of bounding the matrix norms that appear in \Cref{thm:XEB_bound_from_norms}.
The resulting asymptotic bounds on classical communication complexity as a function of $\eps$ (\Cref{cor:classical_lower_bound}) are summarized in \Cref{tab:complexities}.

Haar-random measurements achieve the tightest quantum-classical separation in communication complexity, but asymptotically match approximate $n$-designs, which require exponentially fewer gates to implement.
Random Clifford measurements perform nearly as well as both designs and the Haar measure: all three give classical lower bounds of the form $2^{n - o(n)}$.
Each of these quantum-classical separations improve upon the separation for the Hidden Matching (HM) problem by roughly a quadratic factor~\cite{BJK08-comm}.
As a function of the fidelity $\eps$, the best-known classical lower bound for $n$-qubit noisy HM is precisely~\cite[Theorem 3.3]{BRSW12-bell}:\footnote{\cite{BRSW12-bell} expresses the communication bound in terms of the success probability $p$, which is related to the fidelity $\eps$ by $p = \frac{1 + \eps}{2}$ because the success probability of random guessing is $\frac{1}{2}$ for Hidden Matching.}
\[
\eps \frac{\sqrt{2^n} - 1}{2} - 1 = \Omega\mparen{\eps\sqrt{2^n}}.
\]
On the other hand, the separation we prove for product Cliffords is worse than the separation achievable by HM.

\renewcommand{\arraystretch}{1.4}
\setlength\tabcolsep{1ex}
\begin{table}[h]
\centering

\renewcommand{\thempfootnote}{\arabic{footnote}}
\addtocounter{footnote}{1}

\newcommand{\mycell}[2]{
\begin{tabular}{@{}c@{}}#1\\ #2\end{tabular}
}
\begin{tabular}{r||c|c|c}
\textbf{Ensemble} & \textbf{2Q Complexity} & \textbf{Bound, small $\eps$}
& \textbf{Bound, large $\eps$}\\\hline\hline
Random product Clifford
& \mycell{0}{\ }
& \mycell{$\Omega\mparen{\eps^2 (3/2)^n}$}{\Cref{lem:product_clifford_frob_norm}}
& \mycell{$\Omega\mparen{\eps (3/2)^{n/2}}$}{\Cref{cor:product_clifford_op_norm}}\\\hline
Random Clifford
& \mycell{$O\mparen{n^2 / \log n}$}{\cite{AG04-stabilizer}}
& \mycell{$\Omega\mparen{\eps^2 2^n}$}{\Cref{lem:clifford_frob_norm}}
& \mycell{$\eps 2^{n - O(\sqrt{n})}$}{\Cref{lem:clifford_op_norm}}\\\hline
Approximate $n$-design
& \mycell{$O\mparen{n^2 \polylog(n)}$}{\cite{SHH24-designs}}
& \mycell{$\Omega\mparen{\eps^2 2^n}$}{\Cref{lem:clifford_invariant_frob_norm}\Footnotemark}
& \mycell{$\Omega\mparen{\eps 2^n / n}$}{\Cref{lem:design_op_norm}}\\\hline
Haar measure
& \mycell{$O(4^n)$}{\cite{KA24-zxz}}
& \mycell{$\Omega\mparen{\eps^2 2^n}$}{\Cref{lem:clifford_invariant_frob_norm}}
& \mycell{$\Omega\mparen{\eps 2^n / n}$}{\Cref{lem:haar_op_norm}}\\
\end{tabular}
\caption{
Comparison of classical communication complexity lower bounds for $n$-qubit $\eps$-DXHOG, parameterized by Bob's ensemble of measurement unitaries and assuming Alice's state is Haar-random.
2Q complexity is the number of two-qubit gates required to implement Bob's measurement.
} \label{tab:complexities}
\end{table}
\addtocounter{footnote}{-1}\Footnotetext{This assumes that the approximate design is Clifford-invariant.
Any design can be transformed into a Clifford-invariant one, via left-multiplication by a random Clifford unitary.
This increases the gate complexity by at most an additive $O\mparen{n^2 / \log n}$~\cite{AG04-stabilizer}, while preserving the design property.} %

From an experimental standpoint, asymptotic bounds do not suffice to demonstrate a provable quantum advantage via DXHOG---we need explicit expressions with computable constants for system sizes of interest.
Fortunately, the proof of \Cref{thm:XEB_bound_from_norms} derives such expressions already.
For completeness, we reproduce them here.
Assume, as in \Cref{thm:XEB_bound_from_norms}, that we have norm bounds
\[
    \norm{M(g;\UnitaryDist)}_F \le A \qquad\text{and}\qquad \norm{M(g;\UnitaryDist)}_{op} \le B.
\]
Then \Cref{thm:XEB_bound_from_norms} upper bounds the largest $\FXEB$ $\eps$ achievable by any $m$-bit protocol for $\eps$-DXHOG.
The bound depends on an auxiliary quantity $a$, and is valid for all $a > 1$.
In practice, $a = 1.5$ works well, but the choice of $a$ can be optimized numerically to obtain the smallest bound on $\eps$ as a function of $m$.
The quantity $\gamma$ is defined as
\[
\gamma \coloneqq \frac{ae^{1/a}}{a+1} + \frac{2}{e(a^3 - a)} - 1.
\]
Then the quantity $t^*$ is computed via
\[
t^* \coloneqq \begin{cases}
\sqrt{\ln(2)m \cdot 4\gamma a^2 A^2} & m \le \frac{\gamma A^2}{\ln(2) B^2},\\
    aB (\ln(2) m  + \gamma A^2 / B^2) & m > \frac{\gamma A^2}{ \ln(2) B^2}.
\end{cases}
\]
Finally, the exact bound on $\eps$ derives from \Cref{eq:frob_bound_exact,eq:op_bound_exact}:
\[
\eps(m, a) =
\begin{cases}
    t^* + 2^m \sqrt{\pi \gamma} aA \mparen{\erf\mparen{\frac{\sqrt{\gamma} A}{B}} - \erf\left(\frac{t^*}{2\sqrt{\gamma} aA}\right)} + 2^{m} aB\exp\left(- \frac{\gamma A^2}{B^2}\right) & m \le \frac{\gamma A^2}{\ln(2) B^2},\\
    t^* + aB & m > \frac{\gamma A^2}{\ln(2)B^2}.
\end{cases}
\]

Let us briefly explain how to instantiate the bound with the specific parameters of our experiment.
Namely, to show that at least $62$ (resp.\ $78$) classical bits are required to solve $\eps$-DXHOG for $\eps = \hFXEB - 5\sigma = 0.362$ (resp.\ $\eps = \hFXEB = 0.427$) when $n = 12$ and Bob performs a random Clifford measurement.
From \Cref{lem:clifford_frob_norm} we have
$A = \sqrt{\frac{2}{2^{12} + 1}} \approx 2.2094 \times 10^{-2}$.
The bound from \Cref{lem:clifford_op_norm} is optimized at $t = 5$ and gives $B = \mparen{\prod_{i=0}^3 \frac{2^i + 1}{2^{12} + 2^i}}^{1/5}  \approx 3.9452 \times 10^{-3}$.
Plugging in $m = 61$ and $a = 1.53$ gives $\eps(m, a) < 0.360 < \hFXEB - 5\sigma$, whereas $m = 77$ and $a = 1.47$ gives $\eps(m, a) < 0.426 < \hFXEB$.
Hence, any classical protocol whose population mean $\FXEB$ equals the observed sample mean $\hFXEB$ of the experiment requires $m \ge 78$ bits, and even achieving $\FXEB$ five standard errors below the sample mean requires $m \ge 62$ bits.

\Cref{fig:ensemble_comparison} plots $\min_{a > 1} \eps(m, a)$ for $n = 12$ and each of the Clifford, design, and Haar-random ensembles, where the minimization over $a$ is performed numerically.
Even though the Clifford bound is somewhat weaker, it remains competitive with the others---for $n=12$ it is within a multiplicative factor of $1.4$ of the Haar bound.
This small quantitative difference in classical bounds, along with the relative simplicity of implementing Clifford unitaries in practice, explain the choice of using random Clifford measurements for our experimental demonstration of quantum information supremacy.
Indeed, either $t$-designs or the Haar measure would prove prohibitively expensive to implement in terms of gate count: a Haar measurement requires quadratically more gates ($\sim 4^n$) to implement than even worst-case state preparation ($\sim 2^n$), and the constant factors in the complexity of implementing approximate $t$-designs are astronomical even for small $t$.
In contrast, we will see in \Cref{sec:experiment} that the gate complexity of a Clifford measurement is quite small compared to the overall cost of the protocol, which is dominated by state preparation.

\begin{figure}
    \centering
    \includegraphics[width=0.7\textwidth]{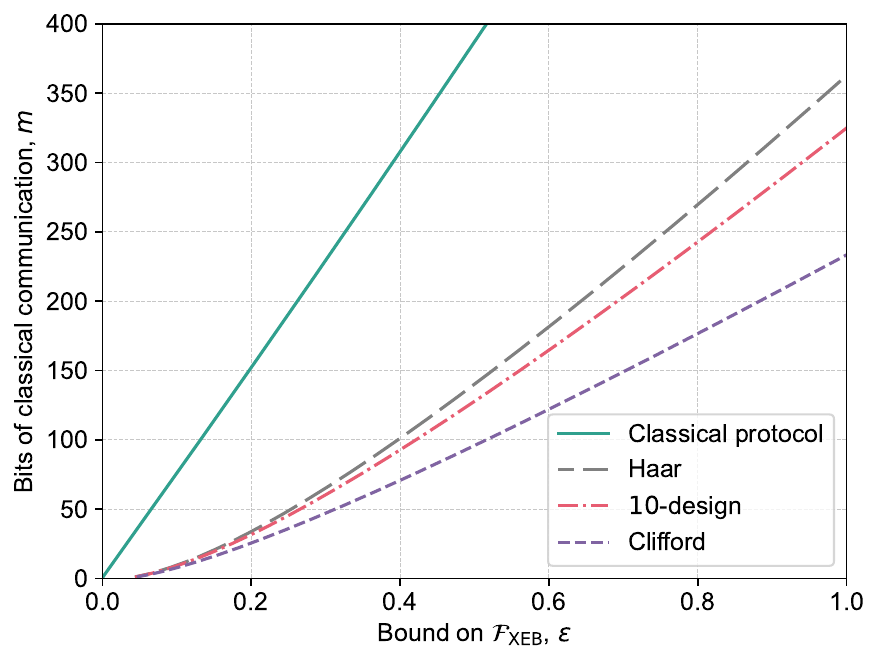}
    \caption{Comparison of upper bounds on $\FXEB$ derived from \Cref{thm:XEB_bound_from_norms} with $n = 12$ and three different measurement ensembles: random Clifford (purple), exact $t$-design (red), and the Haar measure (gray).
    For the $t$-design ensemble, $t = 10$ optimizes the bound from \Cref{lem:design_op_norm}.
    Also shown is the lower bound on $\FXEB$ known to be achievable by classical protocols (green), due to \Cref{cor:classical_upper_bound}.
    Although $\eps$ is the dependent variable (i.e., it is computed as a function of $m$), the figure is transposed to indicate the bounds on communication complexity as a function of the average linear cross-entropy benchmark.
    }
    \label{fig:ensemble_comparison}
\end{figure}

\Cref{fig:ensemble_comparison} also plots the largest $\FXEB$ known to be achievable by classical protocols, which we derived in \Cref{cor:classical_lower_bound}.\footnote{Technically, \Cref{fig:ensemble_comparison} displays the lower bound on this $\FXEB$ from \Cref{eq:eps_lower_bound}, but the difference between the two quantities is too small to be perceptible on the figure.}
We find it remarkable that the gap between our lower and upper bounds on communication is within an order of magnitude for reasonable parameters, despite the presumed sub-optimality of our proof techniques based on generic concentration inequalities and union bounding.
For example, our bounds imply that matching the noiseless $12$-qubit quantum protocol, which achieves $\FXEB \approx 1$, requires at least $234$ bits with Clifford measurements, $325$ bits with $10$-design measurements, or $363$ bits with Haar-random measurements, but never more than $801$ bits for any measurement ensemble.
Evaluating the bound from \Cref{cor:classical_lower_bound} at $m = 330$ gives $\eps > 0.428 > \hFXEB$, and $m = 382$ gives $\eps > 0.493 > \hFXEB + 5\sigma$.
Hence, there exists an $m = 330$-bit classical protocol whose population mean $\FXEB$ exceeds the observed sample mean $\hFXEB$ of the experiment, and $m = 382$ bits suffice to achieve $\FXEB$ five standard errors above the sample mean.
Combining our lower and upper bounds, we deduce with high confidence that the results of the experiment would require, at minimum, between $62$ and $382$ bits to spoof classically.

A natural question for future work is to close the gap between our lower and upper bounds.
This seems challenging, as our upper bound on classical communication (\Cref{thm:classical_upper_bound} and \Cref{cor:classical_lower_bound}) is essentially solving a packing problem in high dimensions: we want to find an arrangement of vectors $\{\ket{\varphi_x}\}_{x \in \{0,1\}^m}$ such that $\max_x \abs{\braket{\varphi_x|\psi}}^2$ is large for most $\ket{\psi}$.
Improving the randomized construction might require obtaining sharper bounds on the sizes of high-dimensional spherical codes, which is a major unresolved question in mathematics.

\section{Experimental Setup and Circuit Synthesis}
\label{sec:experiment}

The experiment in this work was performed on the Quantinuum H1-1 trapped-ion quantum computer, based on the quantum charge-coupled device (QCCD) architecture \cite{Pino2020, PhysRevX.13.041052}. The architecture consists of qubit and sympathetic coolant ions ($\mbox{}^{171}\text{Yb}^+$ and $\mbox{}^{138}\text{Ba}^+$, respectively) confined in a linear RF trap. Qubits are encoded in the hyperfine states $|0\rangle = |F = 0, m_F = 0\rangle$, $|1\rangle = |F = 1, m_F = 0\rangle$ of the $\mbox{}^{171}Yb^+$ ion and initialized in the $|0\rangle$ state by optical pumping \cite{PhysRevA.76.052314}. $20$ total qubits ($40$ total ions) are loaded in the trap at all times, though only $12$ were used in this experiment. Arbitrary connectivity between qubits is achieved by physically transporting ions in the trap, including operations that linearly shift the whole chain of ions, split or combine potential wells holding the ions, and swap the locations of qubits along with their accompanying coolant ions. Gating is performed in five distinct regions of the trap referred to as ``gate zones", supporting up to five single-qubit (1Q) or two-qubit (2Q) operations performed in parallel. 1Q rotations about an arbitrary axis in the XY-plane of the Bloch sphere are accomplished via stimulated Raman transitions \cite{Pino2020} using co-propagating laser beams. As 1Q rotations about the Z-axis of the Bloch sphere can be implemented fully in software by tracking the phases of each individual qubit, arbitrary single-qubit unitaries in $\SUnitaries(2)$,
\[
U3(\theta,\phi,\lambda) \coloneqq \begin{bmatrix}
    \cos \frac{\theta\phi}{2} & -e^{i\pi\lambda} \sin\frac{\pi\theta}{2}\\
    e^{i\pi\phi} \sin\frac{\pi\theta}{2} & e^{i\pi(\lambda+\phi)} \cos\frac{\pi\theta}{2}
\end{bmatrix},
\]
are implemented with only a single physical 1Q gate on hardware by decomposing into a Z-axis rotation and a rotation about an axis in the XY-plane. 2Q gates are implemented via a phase-sensitive M{\o}lmer-Sorenson interaction dressed with global 1Q wrapper pulses \cite{PhysRevA.62.022311, Lee_2005, PhysRevResearch.2.013317}, yielding the phase-insensitive parameterized $ZZ$ rotation gate
\[
ZZ(\theta) \coloneqq e^{-i\frac{\pi}{2}\theta (Z \otimes Z)} = \begin{bmatrix}
    e^{-\frac12 i \pi\theta} & 0 & 0 & 0 \\
    0 & e^{\frac12 i \pi\theta} & 0 & 0 \\
    0 & 0 & e^{\frac12 i \pi\theta} & 0 \\
    0 & 0 & 0 & e^{-\frac12 i \pi\theta}
\end{bmatrix}.
\]
Finally, measurement proceeds via a state-dependent fluorescence measurement read out into a multichannel PMT array. The long coherence times, low crosstalk achieved by physical separation of the qubits, and precisely controllable interactions in the QCCD architecture have enabled Quantinuum's trapped-ion quantum computers to achieve and maintain the highest reported commercial gate fidelities across all architectures to date \cite{github_spec}.

\subsection{Hardware Benchmarking}
\label{sec:benchmarking}

In the main text and in \Cref{sec:variational} below, we describe a variational state preparation protocol for approximating Haar-random states. This protocol must account both for the infidelity of the approximate state preparation, as well as for the expected infidelity introduced by the choice of gates and decoherence introduced over the time of circuit implementation. To that end, it is important to have an accurate model for the expected state fidelity given each choice of variational parameters that can be used in training those parameters. In this section we describe standard benchmarking protocols that were used to obtain such a model.

To estimate the fidelity of the state prepared by a fixed circuit ansatz we follow \cite{PhysRevX.15.021052} in using a gate-counting model in which the primary contributions to circuit infidelity are expected to be two-qubit gate error as well as memory errors (other errors incurred over the duration of the circuit that are not due to two-qubit gates). We compute the expected average infidelity of the two-qubit gate $ZZ(\theta)$ as a function of angle $\theta$ using 2Q parameterized gate randomized benchmarking (RB) as described in \cite{PhysRevX.13.041052}. We include a description of the procedure as described therein here for completeness.

In general, RB experiments proceed by applying sequences of random unitary operators of increasing sequence length to a set of qubits. The overall sequence of random unitary operators is inverted at the end so as to apply the identity operator in aggregate, up to a final Pauli operator applied to randomize the measurement basis to remove any possible measurement bias. For an RB experiment on $N$ qubits, and the fraction of qubits that end in the correct state across all shots is measured (the \emph{average survival}). The average survival is fit to a decay curve of the form
\begin{align}
    p(\ell) = Ar^{\ell} + 1/2^N
\end{align}
where $p(\ell)$ is the average survival at sequence length $\ell$, $A$ is a fit parameter corresponding to SPAM (state preparation and measurement) error, and $r$ is a fit parameter corresponding to the depolarizing rate, and $N$ is the number of qubits. Standard RB theory \cite{PhysRevLett.106.180504, PhysRevA.99.052350, NIELSEN2002249} relates $r$ to the average infidelity of the component operation being studied by
\begin{align}
    \eps = \frac{2^N - 1}{2^N} (1 - r)
\end{align}

In the case of 2Q parameterized gate RB, at each of the $\ell$ circuit layers a $ZZ(\theta)$ gate with $\theta$ some fixed angle is applied to a pair of qubits in each of the five gate zones, interleaved with Haar-random $\SUnitaries(2)$ gates on each qubit in between 2Q gate rounds. The overall $\SUnitaries(4)$ operator prepared on each pair of qubits is inverted at the end of the circuit using three $ZZ(\pi / 2)$ gates following the method of \cite{Hanneke2010}, after which a random Pauli operator is applied to each qubit. Since the sequences of random unitaries are applied to pairs of qubits, the decay curve fit to the average survival follows the $N=2$ case above, $p(\ell) = Ar^{\ell} + 1/4$ with $\eps = \frac{3}{4}(1-r)$ the corresponding average infidelity. We plot the average infidelity $\eps$ as a function of angle $\theta \in \{\frac{\pi}{100}, \frac{\pi}{8}, \frac{\pi}{4}, \frac{3\pi}{8}, \frac{\pi}{2}\}$ in \Cref{fig:arb_rb}.\footnote{For the near-zero angle $\theta = \frac{\pi}{100}$, the 2Q entangling gate is sufficiently weak that the overall unitary applied to each pair of qubits approximately lies in $\SUnitaries(2)\otimes \SUnitaries(2)$. This leads to a more complicated RB analysis in which a sum of multiple exponential decay curves must be fit, as can be derived from the representation theory of $\SUnitaries(2)\otimes \SUnitaries(2)$ \cite{PhysRevX.13.041052}.} Linearly fitting $\eps(\theta)$ we obtain a model for average 2Q gate infidelity as a function of $\theta$ in units of $\pi$, $\eps_{2Q}(\theta) \approx [(14.8 \pm 1.9)\theta + (2.7 \pm 0.6)]\times 10^{-4}$.

Memory errors were quantified using a transport 1Q RB experiment. In a transport 1Q RB experiment, rounds of Haar-random $\SUnitaries(2)$ unitaries on each qubit are interspersed with rounds of transport on all qubits in the trap. To accurately represent the memory errors that qubits would incur during transport required for the variational state preparation of \Cref{sec:variational}, each transport round consisted of the exact operations required to accomplish execution of a periodic 1D brickwork circuit, with no 2Q gate applied. A final round of 1Q unitary operators is applied to each qubit to invert the cumulative operation of the $\SUnitaries(2)$ unitaries applied over the course of the circuit, and a random bit flip is applied to each qubit to remove measurement bias. The average survival for each qubit is then aggregated across all shots as many copies of an $N=1$ RB experiment, from which one can compute an average infidelity per qubit per circuit layer due to memory error. Note that this estimate of memory error includes not only errors incurred during transport and idling but also the (relatively small) 1Q gate error introduced by the RB sequence, so 1Q gate errors do not need to be separately accounted for when modeling overall circuit fidelity.

As only $12$ qubits were used in this experiment out of the $20$ total loaded into the trap, and brickwork circuits require transport that is substantially easier to accomplish than that for random circuits, one expects the total memory error for $N=12$ brickwork circuits to be substantially better than the overall system specification quoted at \cite{github_spec} which is computed for $N=20$ random circuits. Indeed, preliminary time estimates obtained from the Quantinuum hardware compiler indicated that $N=12$ brickwork circuits can be executed about 42\% faster than $N=20$ random circuits on H1-1. The resulting average infidelity per qubit per circuit layer computed from the data displayed in \Cref{fig:brickwork_tsqrb} was $\eps_{\text{mem}} = (8 \pm 2)\times 10^{-5}$, substantially lower than the general system specification for $N=20$ random circuits of $(2.1 \pm 0.2)\times 10^{-4}$.

\begin{figure*}[!t]
\centering
\subfloat[\label{fig:arb_rb}]{\includegraphics[width = 0.45\textwidth]{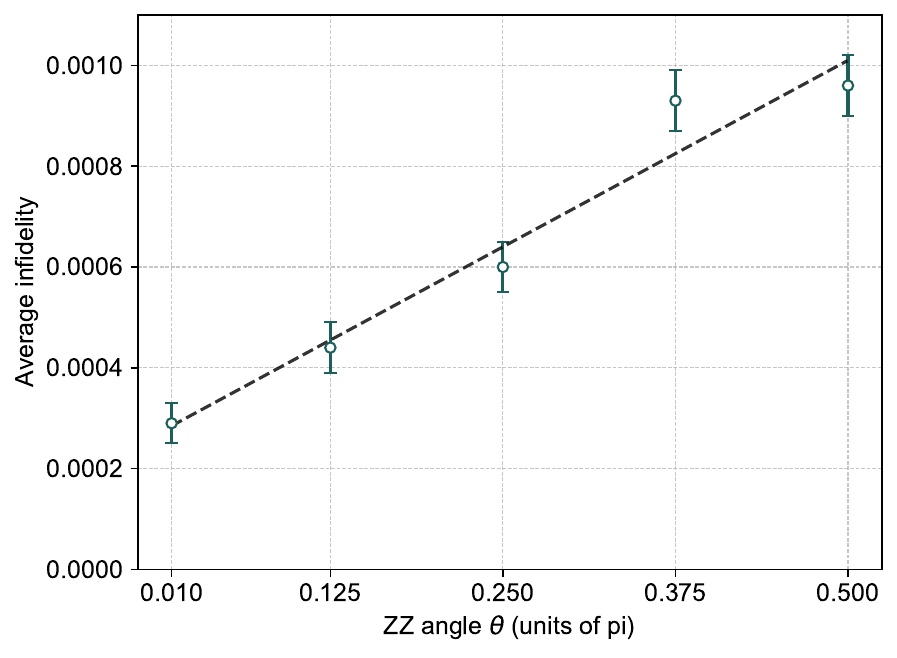}}
\subfloat[\label{fig:brickwork_tsqrb}]{\includegraphics[width = 0.45\textwidth]{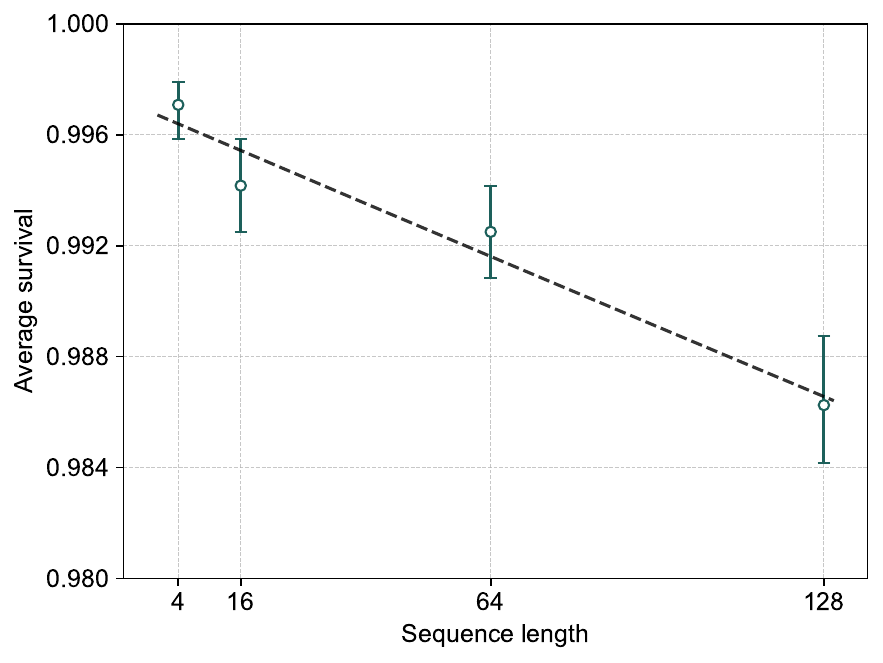}}
\caption{(a) Average infidelity as a function of ZZ angle $\theta$ in 2Q parameterized gate RB. Each experiment used sequence lengths $\ell = 4,60,120$, with 10 random circuits per sequence length and 100 shots per circuit \cite{arb_rb_expt}.
(b)  Average survival as a function of sequence length $\ell$ in a transport 1Q RB experiment for 12-qubit brickwork circuits. The experiment used sequence lengths $\ell = 4,16,64,128$ with 20 random circuits per sequence length and 10 shots per circuit.  $(8\pm 2) \times 10^{-5}$ memory error\label{benchmarking} }
\end{figure*}

\subsection{Variational State Preparation} \label{sec:variational}
The best-known provable upper bound on the number of CNOT gates needed to prepare a general $n$-qubit pure state $\ket{\psi}$ (with costless single-qubit gates) is
$\frac{23}{24}2^n - \frac{3}{2}2^\frac{n+1}{2} + 4/3$
for $n$ odd and $\frac{23}{24}2^n - \frac{3}{2}2^{\frac{n}{2}+1} + 5/3$
for $n$ even~\cite[Remark 5]{ICKHC16-isometries} (see also \cite{BVMS05-circuit,PB11-prep}).
This bound is impractical for our purposes: with $n=12$, the resulting circuit requires $3735$ CNOT gates.
Running this circuit on H1-1, %
we would expect to produce a state having fidelity roughly $0.0036$ with $\ket{\psi}$, which is too small to achieve a quantum advantage over our provable classical lower bounds.

We instead prepare $n$-qubit states using a heuristic approach inspired by variational quantum algorithms.
We take a quantum circuit ansatz $C(\Vec{\theta})$ parameterized by a vector $\Vec{\theta}$, and use classical numerical optimization to approximately maximize the fidelity of its output state with the desired target state $\ket{\psi}$.
Our circuit ansatz utilizes parameterized gates that can be implemented on H1-1 with one physical gate each.
Specifically, the ansatz consists of arbitrary 1Q unitaries in $\SUnitaries(2)$ of the form $U3(\theta,\phi,\lambda)$ above
as well as the $ZZ(\theta)$ parameterized entangling gates. Our chosen ansatz alternates between layers of $U3$ and $ZZ$ gates, with the $ZZ$ gates arranged in a brickwork architecture with periodic boundary conditions as depicted in \Cref{fig:main-circuit} in the main text. The variational parameters to be classically trained are the collection of all angles $\vec{\theta}$ that parameterize both the 1Q and 2Q gates. Although the H1-1 quantum computer can execute circuits with arbitrary gate connectivities, numerical experiments with circuits of more general connectivity did not display a clear advantage in terms of improved ansatz fidelity, likely due to the increased difficulty of optimizing the parameters in a more general ansatz.

\begin{figure*}[!ht]
\centering
\includegraphics[width = 0.7\textwidth]{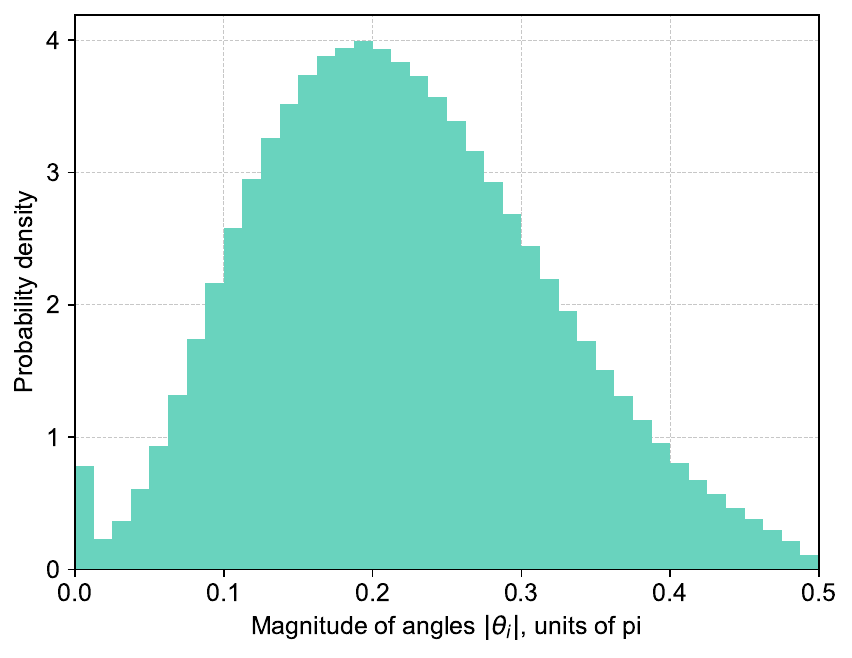}
\caption{Distribution of magnitudes of the variationally optimized angles $\theta_i$ of 2Q gates used to estimate the overall fidelity of state preparation, in units of $\pi$.\label{fig:angles} }
\end{figure*}

In the absence of noise, the fidelity of the ansatz circuit output with $\ket{\psi}$ is simply $\abs{\braket{\psi|C(\Vec{\theta})|0^n}}^2$.
We model the effect of experimental noise on the fidelity using a gate-counting / ``digital error" model as described in \cite{PhysRevX.15.021052}. In this model, each entangling gate with angle $\theta_i$ reduces the overall circuit fidelity by a factor of $1 - \frac{5}{4}\eps_{2Q}(\theta_i) - \frac{3}{2} \times 2 \eps_{\text{mem}}$. Here the numerical factors of $\frac54$ and $\frac32$ correspond to standard conversions between average and process fidelity \cite{Nielsen2002}, the memory error is doubled to account for the fact that there are two qubits per 2Q gate, and $\eps_{2Q}(\theta_i)$, $\eps_{\text{mem}}$ are the models for parameterized 2Q gate error and memory error computed in \Cref{sec:benchmarking}. Note that the memory error as measured by transport 1Q RB also includes the cumulative effect of 1Q gate errors introduced over the course of the circuit.

The overall fidelity estimate of state preparation is then given by
\[
F(\Vec{\theta}) = \abs{\braket{\psi|C(\Vec{\theta})|0^n}}^2 \prod_{i=1}^{\lfloor n/2 \rfloor d} \mparen{1 - \frac{5}{4}\eps_{2Q}(\theta_i) - 3 \eps_{\text{mem}}},
\]
Here, the product is taken over the index $i$ that parameterizes each 2Q gate angle $\theta_i$, $d$ is the number of two-qubit layers (i.e., circuit depth), $\eps_{\text{mem}} = 8\times 10^{-5}$ is the memory error, and $\eps_{2Q}(\theta) = (14.8\times |\theta| + 2.7)\times 10^{-4}$ is the two-qubit gate error as a function of the angle $\theta$ in $ZZ(\theta)$.
Based on internal testing, we take the error model for 2Q gates as a function of gate angle to be identical for positive and negative angles; on hardware, the angles $\theta_i$ are implemented in the range $[-\pi/2, \pi/2]$.\footnote{The gate $ZZ(\theta)$ na\"ively has periodicity $4\pi$. However, the angle can be taken modulo $2\pi$ and shifted into the range $[-\pi, \pi]$ at the expense of only a global phase. Further reduction to the range $[-\pi/2, \pi/2]$ is accomplished using only $Z$-axis 1Q rotations that are performed in software.} In \Cref{fig:angles} we plot the distribution of magnitudes of variationally trained angles obtained by this procedure, the median of which is approximately $0.213\pi \approx 0.67$ (slightly below half the maximal entangler). This median value corresponds to a partial entangler fidelity of $5.9(7)\times 10^{-4}$.

We numerically optimize the function $F(\Vec{\theta})$ using the L-BGFS-B method~\cite{BLNZ95-l-bfgs-b,ZBLN97-l-bfgs-b} in \texttt{scipy}~\cite{SciPy} with a maximum of $10000$ iterations.
For this optimization procedure, we initialize the $U3(\theta,\phi,\lambda)$ parameters pseudorandomly from the Haar measure over $\SUnitaries(2)$ using \texttt{numpy}'s default random number generator~\cite{NumPy,oneill:pcg2014}, and initialize all $ZZ(\theta)$ parameters to $\theta = 0$.
$F(\Vec{\theta})$ and its gradient are computed using a jit-compiled \texttt{jax} implementation~\cite{jax2018github} with 32-bit floating point precision. Notably, \Cref{fig:angles} shows a small peak near zero angle that is likely a relic of the initial conditions of the optimizer; due to the high dimensionality of this parameter space we expect that some directions in parameter space are nearly flat and difficult to optimize numerically in finite time.

For $n=12$, we chose depth $d = 86$ for the ansatz circuit.
The average of $F(\Vec{\theta})$ for the $10000$ circuits we optimized for the experiment (with corresponding Haar-random states $\ket{\psi}$) was $0.4638$, with a standard deviation of $3.1\times 10^{-3}$.

While \Cref{thm:main_informal} shows that any quantum protocol that can achieve $\eps$-DXHOG demonstrates an asymptotic separation in communication complexity over the corresponding classical protocol, it is important to note that a very high-fidelity quantum computer is required to achieve the $\FXEB \sim \Theta(1)$ cross-entropy fidelities that are required to witness the separation in a practical experiment. This is essentially because the cross-entropy fidelity $\FXEB \ge \eps$ in the $\eps$-DXHOG protocol is expected to be a good estimator for the overall circuit fidelity~\cite{AAB+19-google-supremacy,GKC+21-xeb,DHJB24-noise}, or at least strongly correlated with overall circuit fidelity. At a fixed circuit depth, however, this assumption is known to fail for sufficiently weakly entangling 2Q gates~\cite{Boixo2018}. In \Cref{fig:svsims} we compare the gate-counting estimates of fidelity with $\hFXEB$ and the true fidelity $F$ for $N=12$ qubit circuits by executing statevector simulations at a range of circuit depths, where the 2Q gate angle is taken to be $0.2 \pi$, slightly below the median of the variationally trained angles. The three metrics generally agree for depths above $d\sim 60$, well below the $d = 86$ used in our brickwork circuits, showing that for a quantum computer to achieve $\eps$-DXHOG it is required to prepare the state $|\psi\rangle$ with fidelity approximately $\varepsilon$.

\begin{figure*}[!t]
\centering
\includegraphics[width = 0.7\textwidth]{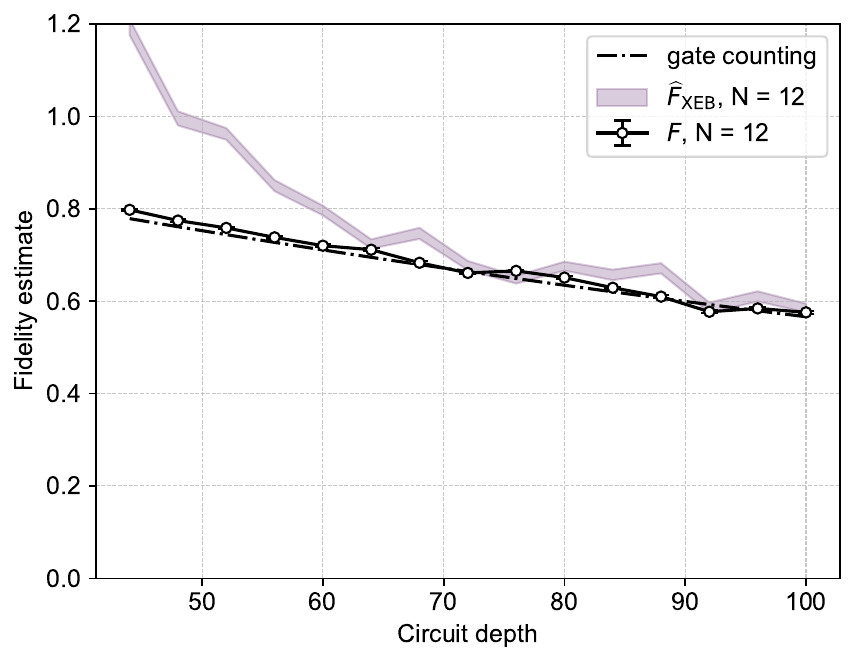}
\caption{Comparing fidelity estimators for $N=12$ brickwork circuits as a function of circuit depth in noisy simulation. The three metrics are comparable for circuit depths roughly exceeding $d\sim 60$. Each circuit consisted of a brickwork circuit as described above, alternating layers of $ZZ(\theta)$ entangling gates with $\theta = 0.2 \pi$ with random $\SUnitaries(2)$ unitaries on each qubit, and was simulated using $2^{14} = 16384$ shots. Uncertainties on the fidelity $F$ are computed via a non-parametric bootstrap resampling \cite{EfroTibs1993} with $100$ resamples, and are plotted above but are smaller than the plotted point size.\label{fig:svsims} }
\end{figure*}

\subsection{Clifford Measurement}
\label{sec:clifford_meas}

It is well-known that every $n$-qubit Clifford unitary admits a decomposition into at most $O(n^2 / \log n)$ one- and two-qubit gates~\cite{AG04-stabilizer}.
The method we choose for performing a Clifford measurement is asymptotically less efficient, requiring $\Theta(n^2)$ gates, but it yields a circuit in a particularly simple form that has favorable constants in practice.

A first observation that reduces the complexity of implementation is to notice that a Clifford measurement can be performed with fewer gates than a general Clifford unitary, because a Clifford unitary contains extra degrees of freedom.
In particular, if $U$ is a Clifford unitary and we wish to measure in the basis $B \coloneqq \{U\ket{x} : x \in \{0,1\}^n\}$, then the basis is invariant under right-multiplication of $U$ by any diagonal or permutation operator.
So, the information in $U$ is a redundant way to specify $B$.
In fact, if $U'$ is \textit{any} Clifford that maps the all-zeros state into $B$ (i.e., $U'\ket{0^n} \in B$), then $U'$ also maps the entire computational basis into $B$ (i.e., $B = \{U'\ket{x} : x \in \{0,1\}^n\}$).
This is a simple consequence of the fact that all states in any stabilizer basis share the same stabilizer group, modulo phase.
Hence, to perform a measurement in a random Clifford basis, it instead suffices to sample a circuit $U$ such that $U\ket{0^n}$ is a random stabilizer state, and then measure in the basis $B = \{U\ket{x} : x \in \{0,1\}^n\}$.
So, henceforth in this section, we focus on the task of generating a circuit to produce a random stabilizer state $\ket{\psi}$ from $\ket{0^n}$.

Every stabilizer state has an expression in the following canonical form: 
\begin{equation}
\label{eq:stab_canonical_form}
\sum_{x \in y + V} (-1)^{q(x)} i^{\ell(x)} \ket{x}, 
\end{equation}
where $V$ is a linear subspace of $\F_2^n$, $y \in \F_2^n / V$ is an affine shift, $q$ is an $\F_2$-quadratic polynomial over $y + V$, and $\ell$ is an $\F_2$-linear polynomial over $y + V$~\cite{DDM03-quadratic,GV08-LU-LC,VDN10-gk}.
Aaronson and Gottesman~\cite[Lemma 6]{AG04-stabilizer} further observe that any stabilizer state may be easily transformed into one where the subspace $V$ is all of $\F_2^n$.
In particular, they show that for any stabilizer state $\ket{\psi}$, there is some subset $\overline{\PivotSet}$ of qubits for which applying Hadamard gates to the subset yields a different stabilizer state that has nonzero support over all $2^n$ basis states.
Using the above expansion of a stabilizer state in terms of a quadratic and linear forms in the phases, it follows that there exist $s, z \in \{0,1\}^n$ and an upper-triangular matrix $M \in \{0,1\}^{n \times n}$ with zero diagonal such that\footnote{According to~\cite{DDM03-quadratic,GV08-LU-LC,VDN10-gk}, the exponent $x \cdot s$ in the term $i^{x \cdot s}$ should be computed mod $2$.
However, our analysis below shows that it is always possible (and arguably more natural) to use a representation where the exponent is taken mod $4$.
}
\begin{align*}
H^{\otimes \overline{\PivotSet}}\ket{\psi} &= \frac{1}{\sqrt{2^n}}\sum_{x \in \{0,1\}^n} i^{x \cdot s} (-1)^{x^\top M x + x \cdot z} \ket{x}\\
&= \mparen{\prod_{i < j : M_{i,j} = 1} CZ_{i,j}} \mparen{\prod_{i: z_i = 1} Z_i} \mparen{\prod_{i: s_i = 1} S_i} \frac{1}{\sqrt{2^n}}\sum_{x\in\{0,1\}^n}\ket{x}\\
&= \mparen{\prod_{i < j : M_{i,j} = 1} CZ_{i,j}} \mparen{\prod_{i: z_i = 1} Z_i} \mparen{\prod_{i: s_i = 1} S_i} H^{\otimes n} \ket{0^n}\\
&= \mparen{\prod_{i < j : M_{i,j} = 1} CZ_{i,j}} \mparen{\prod_{i: s_i = 1} S_i} H^{\otimes n} \mparen{\prod_{i: z_i = 1} X_i}\ket{0^n}.
\end{align*}
In the last line, we used the identity $ZH = HX$.
Thus, any stabilizer state may be prepared from $\ket{0^n}$ by a circuit of the form $X$--$H$--$S$--$CZ$--$H$, where $X$, $H$, $S$, and $CZ$ represent respectively layers of Pauli-$X$, Hadamard, Phase, and controlled-$Z$ gates.
This is the circuit form with which we perform our random stabilizer measurements.\footnote{Actually, the $X$ layer is redundant when performing a stabilizer measurement because it simply permutes the basis, as noted above.
Our implementation did not make use of this optimization.}

A few algorithmic approaches are possible for the task of sampling a random stabilizer state in $X$--$H$--$S$--$CZ$--$H$ form.
For example, Bravyi and Gosset~\cite{BG16-clifford} give an $O(n^2)$ time algorithm for sampling a random stabilizer state in the canonical form of \Cref{eq:stab_canonical_form}, which can then be transformed into the above form using standard techniques for manipulating stabilizer tableaux~\cite{AG04-stabilizer}.
However, we find that there is a particularly elegant algorithm (\Cref{alg:random_stab} below) to \textit{directly} sample a random stabilizer state in the form of an $X$--$H$--$S$--$CZ$--$H$ circuit that prepares it.
Technically, \Cref{alg:random_stab} produces a circuit of the form $H$--$CNOT$--$X$--$S$--$Z$--$CZ$, but we show in \Cref{prop:random_stab_transform} below the algorithm how to easily transform into an $X$--$H$--$S$--$CZ$--$H$ circuit.

\normalem %
\begin{algorithm}[H]

\DontPrintSemicolon
\caption{\label{alg:random_stab} Sampling a random stabilizer state}
\KwInput{$n \in \Naturals$}

Initialize the state to $\ket{0^n}$.

Sample $d \in \{0,1,\ldots,n\}$ according to the probability distribution
\[
P(d) \coloneqq \frac{\eta(d)}{\sum_{m=0}^n \eta(m)},
\]
where $\eta(0) \coloneqq 1$ and
\[
\eta(d) \coloneqq 2^{-d(d+1)/2} \cdot \prod_{a=1}^d \frac{1 - 2^{d - n - a}}{1 - 2^{-a}}.
\]

Sample a random $k$-dimensional linear subspace $V$ of $\mathbb{F}_2^n$

Let $M \in \mathbb{F}_2^{k \times n}$ be a generator matrix for $V$ in reduced row echelon form 

Let $\PivotSet \subseteq [n]$ be the set of column indices of $M$ that contain the leading $1$s (i.e., the pivot columns) \tcc*{$|\PivotSet| = k$}

Apply Hadamard gates to all qubits in $\PivotSet$

\ForEach{$i \in \PivotSet$}{
    \For{$j = i + 1$ \KwTo $n$}{
        \If{$M_{ij} = 1$}{
            Apply $CNOT_{i \to j}$
        }
    }
}\tcc*{Now the state is a uniform superposition over $V$}

Apply $X$ gates to each qubit in $\overline{\PivotSet}$ independently with probability $\frac{1}{2}$ \tcc*{Add a random affine shift to $V$}
\label{line:affine_superposition}

Apply $S$ gates to each qubit in $\PivotSet$ independently with probability $\frac{1}{2}$
\label{line:s_gates} \tcc*{Randomize the imaginary phase}

Apply $Z$ gates to each qubit in $\PivotSet$ independently with probability $\frac{1}{2}$
\label{line:z_gates} \tcc*{Randomize the real linear phase}

Apply $CZ$ gates to each pair of qubits in $\PivotSet$ independently with probability $\frac{1}{2}$
\label{line:cz_gates}
\tcc*{Randomize the quadratic phase}

\end{algorithm}
\ULforem %

Before proving the correctness of our algorithm, let us prove that it can indeed be modified into an $X$--$H$--$S$--$CZ$--$H$ circuit.

\begin{proposition}\label{prop:random_stab_transform}
    The output of \Cref{alg:random_stab} can be made into an $X$--$H$--$S$--$CZ$--$H$ circuit.
\end{proposition}
\begin{proof}
   The output of \Cref{alg:random_stab} has the form 
\[
\Qcircuit @C=1em @R=1em {
\lstick{\PivotSet}          & \gate{H} & \ctrl{1} & \gate{S} & \gate{Z} & \gate{CZ} & \qw \\
\lstick{\overline{\PivotSet}} & \qw      & \targ    & \gate{X} & \qw      & \qw & \qw
}
\]   
where the two wires represent qubits in $\PivotSet$ and $\overline{\PivotSet}$, respectively, and a gate on one of the wires represents a layer of those gates applied to the qubits in that set.
The $CNOT$ gates have controls in $\PivotSet$ and targets in $\overline{\PivotSet}$, due to the structure of the reduced row echelon form: each pivot column contains a single nonzero entry.

We begin by moving the $X$ layer to the front of the circuit, because it is disjoint from the $H$ layer and commutes with the $CNOT$ gates, whose targets are only on $\overline{\PivotSet}$.
This puts the circuit in the form:
\[
\Qcircuit @C=1em @R=1em {
\lstick{\PivotSet} & \gate{H} & \ctrl{1} & \gate{S} & \gate{Z} & \gate{CZ} & \qw \\
\lstick{\overline{\PivotSet}} & \gate{X}      & \targ    & \qw & \qw   & \qw   & \qw
}
\]  

Next, we move the $Z$ layer to the front of the circuit.
These gates commute with the $S$ and $CNOT$ layers, as the $Z$ gates are confined to $\PivotSet$ while the $CNOT$ targets are confined to $\overline{\PivotSet}$. Finally, using $HZ = XH$, we pull the $Z$ layer through $H$, transforming it into $X$. The circuit then becomes:
\[
\Qcircuit @C=1em @R=1em {
\lstick{\PivotSet} & \gate{X} & \gate{H} & \ctrl{1} & \gate{S} & \gate{CZ} & \qw \\
\lstick{\overline{\PivotSet}} & \gate{X} & \qw      & \targ    & \qw & \qw      & \qw
}
\]  

We then rewrite each $CNOT$ using the identity $CNOT = (I \otimes H)CZ(I \otimes H)$, which transforms the circuit into: 
\[
\Qcircuit @C=1em @R=1em {
\lstick{\PivotSet} & \gate{X} & \gate{H} &\multigate{1}{CZ} & \gate{S} & \gate{CZ} & \qw \\
\lstick{\overline{\PivotSet}} & \gate{X} & \gate{H}      & \ghost{CZ}    & \gate{H} & \qw      & \qw
}
\] 

Finally, we move the $CZ$ and rightmost $H$ layers to the end of the circuit, which is valid because $CZ$ and $S$ gates are diagonal (and thus commute), and because the $H$ layer acts on disjoint qubits from $S$ and $CZ$.
The final circuit has the form
\[
\Qcircuit @C=1em @R=1em {
\lstick{\PivotSet} & \gate{X} & \gate{H} & \gate{S} &\multigate{1}{CZ} & \qw & \qw\\
\lstick{\overline{\PivotSet}} & \gate{X} & \gate{H} & \qw     & \ghost{CZ}    & \gate{H} & \qw   
}
\qedhere
\] 

\end{proof}

We now prove the correctness of the algorithm.

\begin{theorem}
    The output of \Cref{alg:random_stab} is a uniformly random $n$-qubit stabilizer state.
\end{theorem}
\begin{proof}
    It is clear that the state produced at the end of \Cref{alg:random_stab} is a stabilizer state, because it is obtained from $\ket{0^n}$ by applying a sequence of Clifford gates.
    Hence, we just need to show that the distribution of states sampled is uniform over all stabilizer states.

    Recall from \Cref{eq:stab_canonical_form} that for every $n$-qubit stabilizer state $\ket{\psi}$, there is a linear subspace $V \subseteq \mathbb{F}_2^n$ and an affine shift $x \in \mathbb{F}_2^n / V$ such that $\ket{\psi}$ has uniform support over the affine subspace $x + V$.
    For a randomly chosen stabilizer state, Bravyi and Gosset~\cite[Proof of Lemma 5]{BG16-clifford} show that the dimension $d$ of $V$ is distributed according to $P(d)$, and the choice of $x + V$ is uniformly random over all dimension-$d$ affine subspaces.
    Furthermore, there are exactly $2^{2d} \cdot 2^{d(d-1)/2}$ stabilizer states with support over $x + V$~\cite{BG16-clifford}.

    Observe that each step of \Cref{alg:random_stab} (indirectly) randomizes one of the components in the canonical form from \Cref{eq:stab_canonical_form}.
    After \Cref{line:affine_superposition} and until the end of the algorithm, the support of the state is a random affine subspace $x + V$ with the correct distribution.
    Moreover, notice that each choice of the randomization in \Cref{line:s_gates,line:z_gates,line:cz_gates} yields a distinct state supported over $x + V$, because the restriction of $x + V$ to $\PivotSet$ contains all $2^k$ possible strings.
    \Cref{line:s_gates,line:z_gates,line:cz_gates} use exactly $2^{2d} \cdot 2^{d(d-1)/2}$ bits of randomness, so conditioned on the choice of affine subspace $x + V$, the output distribution is uniform over the $2^{2d} \cdot 2^{d(d-1)/2}$ stabilizer states supported over $x + V$.
    Therefore, the distribution sampled by the algorithm is uniform over all stabilizer states.
\end{proof}

\bibliography{MainBibliography}

\begin{thebibliography}{97}%
\makeatletter
\providecommand \@ifxundefined [1]{%
 \@ifx{#1\undefined}
}%
\providecommand \@ifnum [1]{%
 \ifnum #1\expandafter \@firstoftwo
 \else \expandafter \@secondoftwo
 \fi
}%
\providecommand \@ifx [1]{%
 \ifx #1\expandafter \@firstoftwo
 \else \expandafter \@secondoftwo
 \fi
}%
\providecommand \natexlab [1]{#1}%
\providecommand \enquote  [1]{``#1''}%
\providecommand \bibnamefont  [1]{#1}%
\providecommand \bibfnamefont [1]{#1}%
\providecommand \citenamefont [1]{#1}%
\providecommand \href@noop [0]{\@secondoftwo}%
\providecommand \href [0]{\begingroup \@sanitize@url \@href}%
\providecommand \@href[1]{\@@startlink{#1}\@@href}%
\providecommand \@@href[1]{\endgroup#1\@@endlink}%
\providecommand \@sanitize@url [0]{\catcode `\\12\catcode `\$12\catcode
  `\&12\catcode `\#12\catcode `\^12\catcode `\_12\catcode `\%12\relax}%
\providecommand \@@startlink[1]{}%
\providecommand \@@endlink[0]{}%
\providecommand \url  [0]{\begingroup\@sanitize@url \@url }%
\providecommand \@url [1]{\endgroup\@href {#1}{\urlprefix }}%
\providecommand \urlprefix  [0]{URL }%
\providecommand \Eprint [0]{\href }%
\providecommand \doibase [0]{https://doi.org/}%
\providecommand \selectlanguage [0]{\@gobble}%
\providecommand \bibinfo  [0]{\@secondoftwo}%
\providecommand \bibfield  [0]{\@secondoftwo}%
\providecommand \translation [1]{[#1]}%
\providecommand \BibitemOpen [0]{}%
\providecommand \bibitemStop [0]{}%
\providecommand \bibitemNoStop [0]{.\EOS\space}%
\providecommand \EOS [0]{\spacefactor3000\relax}%
\providecommand \BibitemShut  [1]{\csname bibitem#1\endcsname}%
\let\auto@bib@innerbib\@empty
\bibitem [{\citenamefont {Shor}(1997)}]{Sho99-factoring}%
  \BibitemOpen
  \bibfield  {author} {\bibinfo {author} {\bibfnamefont {P.~W.}\ \bibnamefont
  {Shor}},\ }\bibfield  {title} {\bibinfo {title} {Polynomial-time algorithms
  for prime factorization and discrete logarithms on a quantum computer},\
  }\href {https://doi.org/10.1137/S0097539795293172} {\bibfield  {journal}
  {\bibinfo  {journal} {SIAM Journal on Computing}\ }\textbf {\bibinfo {volume}
  {26}},\ \bibinfo {pages} {1484} (\bibinfo {year} {1997})}\BibitemShut
  {NoStop}%
\bibitem [{\citenamefont {Levin}(2003)}]{levin2003tale}%
  \BibitemOpen
  \bibfield  {author} {\bibinfo {author} {\bibfnamefont {L.~A.}\ \bibnamefont
  {Levin}},\ }\bibfield  {title} {\bibinfo {title} {{The Tale of One-Way
  Functions}},\ }\href {https://doi.org/10.1023/A:1023634616182} {\bibfield
  {journal} {\bibinfo  {journal} {Problems of Information Transmission}\
  }\textbf {\bibinfo {volume} {39}},\ \bibinfo {pages} {92} (\bibinfo {year}
  {2003})}\BibitemShut {NoStop}%
\bibitem [{\citenamefont {Dyakonov}(2019)}]{dyakonov2019}%
  \BibitemOpen
  \bibfield  {author} {\bibinfo {author} {\bibfnamefont {M.}~\bibnamefont
  {Dyakonov}},\ }\bibfield  {title} {\bibinfo {title} {When will useful quantum
  computers be constructed? not in the foreseeable future, this physicist
  argues. here's why: The case against: Quantum computing},\ }\href
  {https://doi.org/10.1109/MSPEC.2019.8651931} {\bibfield  {journal} {\bibinfo
  {journal} {IEEE Spectrum}\ }\textbf {\bibinfo {volume} {56}},\ \bibinfo
  {pages} {24} (\bibinfo {year} {2019})}\BibitemShut {NoStop}%
\bibitem [{\citenamefont {Kalai}(2020)}]{kalai2020argument}%
  \BibitemOpen
  \bibfield  {author} {\bibinfo {author} {\bibfnamefont {G.}~\bibnamefont
  {Kalai}},\ }\bibfield  {title} {\bibinfo {title} {The argument against
  quantum computers},\ }in\ \href
  {https://doi.org/10.1007/978-3-030-34316-3_18} {\emph {\bibinfo {booktitle}
  {Quantum, probability, logic: The work and influence of Itamar Pitowsky}}}\
  (\bibinfo  {publisher} {Springer},\ \bibinfo {year} {2020})\ pp.\ \bibinfo
  {pages} {399--422}\BibitemShut {NoStop}%
\bibitem [{\citenamefont {'t~Hooft}(1999)}]{Hooft_1999}%
  \BibitemOpen
  \bibfield  {author} {\bibinfo {author} {\bibfnamefont {G.}~\bibnamefont
  {'t~Hooft}},\ }\bibfield  {title} {\bibinfo {title} {Quantum gravity as a
  dissipative deterministic system},\ }\href
  {https://doi.org/10.1088/0264-9381/16/10/316} {\bibfield  {journal} {\bibinfo
   {journal} {Classical and Quantum Gravity}\ }\textbf {\bibinfo {volume}
  {16}},\ \bibinfo {pages} {3263} (\bibinfo {year} {1999})}\BibitemShut
  {NoStop}%
\bibitem [{\citenamefont {Wolfram}(2002)}]{Wolfram2002}%
  \BibitemOpen
  \bibfield  {author} {\bibinfo {author} {\bibfnamefont {S.}~\bibnamefont
  {Wolfram}},\ }\href@noop {} {{\selectlanguage {English}\emph {\bibinfo
  {title} {A New Kind of Science}}}}\ (\bibinfo  {publisher} {Wolfram Media},\
  \bibinfo {year} {2002})\BibitemShut {NoStop}%
\bibitem [{\citenamefont {Fuchs}(2011)}]{Fuchs_2011}%
  \BibitemOpen
  \bibfield  {author} {\bibinfo {author} {\bibfnamefont {C.~A.}\ \bibnamefont
  {Fuchs}},\ }\href {https://doi.org/10.1017/CBO9780511762789} {\emph {\bibinfo
  {title} {Coming of Age With Quantum Information: Notes on a Paulian Idea}}}\
  (\bibinfo  {publisher} {Cambridge University Press},\ \bibinfo {year}
  {2011})\BibitemShut {NoStop}%
\bibitem [{\citenamefont {Freedman}\ and\ \citenamefont
  {Clauser}(1972)}]{Clauser-BellInequalityViolation}%
  \BibitemOpen
  \bibfield  {author} {\bibinfo {author} {\bibfnamefont {S.~J.}\ \bibnamefont
  {Freedman}}\ and\ \bibinfo {author} {\bibfnamefont {J.~F.}\ \bibnamefont
  {Clauser}},\ }\bibfield  {title} {\bibinfo {title} {Experimental test of
  local hidden-variable theories},\ }\href
  {https://doi.org/10.1103/PhysRevLett.28.938} {\bibfield  {journal} {\bibinfo
  {journal} {Phys. Rev. Lett.}\ }\textbf {\bibinfo {volume} {28}},\ \bibinfo
  {pages} {938} (\bibinfo {year} {1972})}\BibitemShut {NoStop}%
\bibitem [{\citenamefont {Aspect}\ \emph {et~al.}(1982)\citenamefont {Aspect},
  \citenamefont {Grangier},\ and\ \citenamefont
  {Roger}}]{Aspect-BellInequalityViolation}%
  \BibitemOpen
  \bibfield  {author} {\bibinfo {author} {\bibfnamefont {A.}~\bibnamefont
  {Aspect}}, \bibinfo {author} {\bibfnamefont {P.}~\bibnamefont {Grangier}},\
  and\ \bibinfo {author} {\bibfnamefont {G.}~\bibnamefont {Roger}},\ }\bibfield
   {title} {\bibinfo {title} {Experimental realization of
  einstein-podolsky-rosen-bohm gedankenexperiment: A new violation of bell's
  inequalities},\ }\href {https://doi.org/10.1103/PhysRevLett.49.91} {\bibfield
   {journal} {\bibinfo  {journal} {Phys. Rev. Lett.}\ }\textbf {\bibinfo
  {volume} {49}},\ \bibinfo {pages} {91} (\bibinfo {year} {1982})}\BibitemShut
  {NoStop}%
\bibitem [{\citenamefont {Weihs}\ \emph {et~al.}(1998)\citenamefont {Weihs},
  \citenamefont {Jennewein}, \citenamefont {Simon}, \citenamefont
  {Weinfurter},\ and\ \citenamefont
  {Zeilinger}}]{Zeilinger-BellInequalityViolation}%
  \BibitemOpen
  \bibfield  {author} {\bibinfo {author} {\bibfnamefont {G.}~\bibnamefont
  {Weihs}}, \bibinfo {author} {\bibfnamefont {T.}~\bibnamefont {Jennewein}},
  \bibinfo {author} {\bibfnamefont {C.}~\bibnamefont {Simon}}, \bibinfo
  {author} {\bibfnamefont {H.}~\bibnamefont {Weinfurter}},\ and\ \bibinfo
  {author} {\bibfnamefont {A.}~\bibnamefont {Zeilinger}},\ }\bibfield  {title}
  {\bibinfo {title} {Violation of bell's inequality under strict einstein
  locality conditions},\ }\href {https://doi.org/10.1103/PhysRevLett.81.5039}
  {\bibfield  {journal} {\bibinfo  {journal} {Phys. Rev. Lett.}\ }\textbf
  {\bibinfo {volume} {81}},\ \bibinfo {pages} {5039} (\bibinfo {year}
  {1998})}\BibitemShut {NoStop}%
\bibitem [{\citenamefont {Arute}\ \emph {et~al.}(2019)\citenamefont {Arute}
  \emph {et~al.}}]{AAB+19-google-supremacy}%
  \BibitemOpen
  \bibfield  {author} {\bibinfo {author} {\bibfnamefont {F.}~\bibnamefont
  {Arute}} \emph {et~al.},\ }\bibfield  {title} {\bibinfo {title} {Quantum
  supremacy using a programmable superconducting processor},\ }\href
  {https://doi.org/10.1038/s41586-019-1666-5} {\bibfield  {journal} {\bibinfo
  {journal} {Nature}\ }\textbf {\bibinfo {volume} {574}},\ \bibinfo {pages}
  {505} (\bibinfo {year} {2019})}\BibitemShut {NoStop}%
\bibitem [{\citenamefont {Wu}\ \emph {et~al.}(2021)\citenamefont {Wu} \emph
  {et~al.}}]{WBC+21-ustc-superconducting}%
  \BibitemOpen
  \bibfield  {author} {\bibinfo {author} {\bibfnamefont {Y.}~\bibnamefont {Wu}}
  \emph {et~al.},\ }\bibfield  {title} {\bibinfo {title} {Strong quantum
  computational advantage using a superconducting quantum processor},\ }\href
  {https://doi.org/10.1103/PhysRevLett.127.180501} {\bibfield  {journal}
  {\bibinfo  {journal} {Phys. Rev. Lett.}\ }\textbf {\bibinfo {volume} {127}},\
  \bibinfo {pages} {180501} (\bibinfo {year} {2021})}\BibitemShut {NoStop}%
\bibitem [{\citenamefont {Zhu}\ \emph {et~al.}(2022)\citenamefont {Zhu} \emph
  {et~al.}}]{ZCZ21}%
  \BibitemOpen
  \bibfield  {author} {\bibinfo {author} {\bibfnamefont {Q.}~\bibnamefont
  {Zhu}} \emph {et~al.},\ }\bibfield  {title} {\bibinfo {title} {Quantum
  computational advantage via 60-qubit 24-cycle random circuit sampling},\
  }\href {https://doi.org/10.1016/j.scib.2021.10.017} {\bibfield  {journal}
  {\bibinfo  {journal} {Science Bulletin}\ }\textbf {\bibinfo {volume} {67}},\
  \bibinfo {pages} {240} (\bibinfo {year} {2022})}\BibitemShut {NoStop}%
\bibitem [{\citenamefont {Morvan}\ \emph {et~al.}(2024)\citenamefont {Morvan}
  \emph {et~al.}}]{Sycamore70}%
  \BibitemOpen
  \bibfield  {author} {\bibinfo {author} {\bibfnamefont {A.}~\bibnamefont
  {Morvan}} \emph {et~al.},\ }\bibfield  {title} {\bibinfo {title} {Phase
  transitions in random circuit sampling},\ }\href
  {https://doi.org/10.1038/s41586-024-07998-6} {\bibfield  {journal} {\bibinfo
  {journal} {Nature}\ }\textbf {\bibinfo {volume} {634}},\ \bibinfo {pages}
  {328} (\bibinfo {year} {2024})}\BibitemShut {NoStop}%
\bibitem [{\citenamefont {Gao}\ \emph {et~al.}(2025)\citenamefont {Gao} \emph
  {et~al.}}]{ZCZ24}%
  \BibitemOpen
  \bibfield  {author} {\bibinfo {author} {\bibfnamefont {D.}~\bibnamefont
  {Gao}} \emph {et~al.},\ }\bibfield  {title} {\bibinfo {title} {Establishing a
  new benchmark in quantum computational advantage with 105-qubit zuchongzhi
  3.0 processor},\ }\href {https://doi.org/10.1103/PhysRevLett.134.090601}
  {\bibfield  {journal} {\bibinfo  {journal} {Phys. Rev. Lett.}\ }\textbf
  {\bibinfo {volume} {134}},\ \bibinfo {pages} {090601} (\bibinfo {year}
  {2025})}\BibitemShut {NoStop}%
\bibitem [{\citenamefont {DeCross}\ \emph {et~al.}(2025)\citenamefont {DeCross}
  \emph {et~al.}}]{PhysRevX.15.021052}%
  \BibitemOpen
  \bibfield  {author} {\bibinfo {author} {\bibfnamefont {M.}~\bibnamefont
  {DeCross}} \emph {et~al.},\ }\bibfield  {title} {\bibinfo {title}
  {Computational power of random quantum circuits in arbitrary geometries},\
  }\href {https://doi.org/10.1103/PhysRevX.15.021052} {\bibfield  {journal}
  {\bibinfo  {journal} {Phys. Rev. X}\ }\textbf {\bibinfo {volume} {15}},\
  \bibinfo {pages} {021052} (\bibinfo {year} {2025})}\BibitemShut {NoStop}%
\bibitem [{\citenamefont {Madsen}\ \emph {et~al.}(2022)\citenamefont {Madsen}
  \emph {et~al.}}]{madsen2022quantum}%
  \BibitemOpen
  \bibfield  {author} {\bibinfo {author} {\bibfnamefont {L.}~\bibnamefont
  {Madsen}} \emph {et~al.},\ }\bibfield  {title} {\bibinfo {title} {Quantum
  computational advantage with a programmable photonic processor},\ }\href
  {https://doi.org/10.1038/s41586-022-04725-x} {\bibfield  {journal} {\bibinfo
  {journal} {Nature}\ }\textbf {\bibinfo {volume} {606}},\ \bibinfo {pages}
  {75} (\bibinfo {year} {2022})}\BibitemShut {NoStop}%
\bibitem [{\citenamefont {Aaronson}\ and\ \citenamefont
  {Arkhipov}(2013)}]{AA13-boson-sampling}%
  \BibitemOpen
  \bibfield  {author} {\bibinfo {author} {\bibfnamefont {S.}~\bibnamefont
  {Aaronson}}\ and\ \bibinfo {author} {\bibfnamefont {A.}~\bibnamefont
  {Arkhipov}},\ }\bibfield  {title} {\bibinfo {title} {The computational
  complexity of linear optics},\ }\href
  {https://doi.org/10.4086/toc.2013.v009a004} {\bibfield  {journal} {\bibinfo
  {journal} {Theory of Computing}\ }\textbf {\bibinfo {volume} {9}},\ \bibinfo
  {pages} {143} (\bibinfo {year} {2013})}\BibitemShut {NoStop}%
\bibitem [{\citenamefont {Aaronson}\ and\ \citenamefont
  {Gunn}(2020)}]{AG20-xeb}%
  \BibitemOpen
  \bibfield  {author} {\bibinfo {author} {\bibfnamefont {S.}~\bibnamefont
  {Aaronson}}\ and\ \bibinfo {author} {\bibfnamefont {S.}~\bibnamefont
  {Gunn}},\ }\bibfield  {title} {\bibinfo {title} {On the classical hardness of
  spoofing linear cross-entropy benchmarking},\ }\href
  {https://doi.org/10.4086/toc.2020.v016a011} {\bibfield  {journal} {\bibinfo
  {journal} {Theory of Computing}\ }\textbf {\bibinfo {volume} {16}},\ \bibinfo
  {pages} {1} (\bibinfo {year} {2020})}\BibitemShut {NoStop}%
\bibitem [{\citenamefont {Aaronson}\ \emph {et~al.}(2024)\citenamefont
  {Aaronson}, \citenamefont {Buhrman},\ and\ \citenamefont
  {Kretschmer}}]{ABK23-relation}%
  \BibitemOpen
  \bibfield  {author} {\bibinfo {author} {\bibfnamefont {S.}~\bibnamefont
  {Aaronson}}, \bibinfo {author} {\bibfnamefont {H.}~\bibnamefont {Buhrman}},\
  and\ \bibinfo {author} {\bibfnamefont {W.}~\bibnamefont {Kretschmer}},\
  }\bibfield  {title} {\bibinfo {title} {{A Qubit, a Coin, and an Advice String
  Walk into a Relational Problem}},\ }in\ \href
  {https://doi.org/10.4230/LIPIcs.ITCS.2024.1} {\emph {\bibinfo {booktitle}
  {15th Innovations in Theoretical Computer Science Conference (ITCS 2024)}}},\
  \bibinfo {series} {Leibniz International Proceedings in Informatics
  (LIPIcs)}, Vol.\ \bibinfo {volume} {287},\ \bibinfo {editor} {edited by\
  \bibinfo {editor} {\bibfnamefont {V.}~\bibnamefont {Guruswami}}}\ (\bibinfo
  {publisher} {Schloss Dagstuhl -- Leibniz-Zentrum f{\"u}r Informatik},\
  \bibinfo {address} {Dagstuhl, Germany},\ \bibinfo {year} {2024})\ pp.\
  \bibinfo {pages} {1:1--1:24}\BibitemShut {NoStop}%
\bibitem [{\citenamefont {Bar-Yossef}\ \emph {et~al.}(2008)\citenamefont
  {Bar-Yossef}, \citenamefont {Jayram},\ and\ \citenamefont
  {Kerenidis}}]{BJK08-comm}%
  \BibitemOpen
  \bibfield  {author} {\bibinfo {author} {\bibfnamefont {Z.}~\bibnamefont
  {Bar-Yossef}}, \bibinfo {author} {\bibfnamefont {T.~S.}\ \bibnamefont
  {Jayram}},\ and\ \bibinfo {author} {\bibfnamefont {I.}~\bibnamefont
  {Kerenidis}},\ }\bibfield  {title} {\bibinfo {title} {Exponential separation
  of quantum and classical one-way communication complexity},\ }\href
  {https://doi.org/10.1137/060651835} {\bibfield  {journal} {\bibinfo
  {journal} {SIAM Journal on Computing}\ }\textbf {\bibinfo {volume} {38}},\
  \bibinfo {pages} {366} (\bibinfo {year} {2008})}\BibitemShut {NoStop}%
\bibitem [{\citenamefont {Gavinsky}\ \emph {et~al.}(2009)\citenamefont
  {Gavinsky}, \citenamefont {Kempe}, \citenamefont {Kerenidis}, \citenamefont
  {Raz},\ and\ \citenamefont {de~Wolf}}]{GKK+09-comm}%
  \BibitemOpen
  \bibfield  {author} {\bibinfo {author} {\bibfnamefont {D.}~\bibnamefont
  {Gavinsky}}, \bibinfo {author} {\bibfnamefont {J.}~\bibnamefont {Kempe}},
  \bibinfo {author} {\bibfnamefont {I.}~\bibnamefont {Kerenidis}}, \bibinfo
  {author} {\bibfnamefont {R.}~\bibnamefont {Raz}},\ and\ \bibinfo {author}
  {\bibfnamefont {R.}~\bibnamefont {de~Wolf}},\ }\bibfield  {title} {\bibinfo
  {title} {Exponential separation for one-way quantum communication complexity,
  with applications to cryptography},\ }\href
  {https://doi.org/10.1137/070706550} {\bibfield  {journal} {\bibinfo
  {journal} {SIAM Journal on Computing}\ }\textbf {\bibinfo {volume} {38}},\
  \bibinfo {pages} {1695} (\bibinfo {year} {2009})}\BibitemShut {NoStop}%
\bibitem [{\citenamefont {Montanaro}(2019)}]{Mon19-transmit}%
  \BibitemOpen
  \bibfield  {author} {\bibinfo {author} {\bibfnamefont {A.}~\bibnamefont
  {Montanaro}},\ }\bibfield  {title} {\bibinfo {title} {Quantum states cannot
  be transmitted efficiently classically},\ }\href
  {https://doi.org/10.22331/q-2019-06-28-154} {\bibfield  {journal} {\bibinfo
  {journal} {{Quantum}}\ }\textbf {\bibinfo {volume} {3}},\ \bibinfo {pages}
  {154} (\bibinfo {year} {2019})}\BibitemShut {NoStop}%
\bibitem [{\citenamefont {Buhrman}\ \emph {et~al.}(2012)\citenamefont
  {Buhrman}, \citenamefont {Regev}, \citenamefont {Scarpa},\ and\ \citenamefont
  {Wolf}}]{BRSW12-bell}%
  \BibitemOpen
  \bibfield  {author} {\bibinfo {author} {\bibfnamefont {H.}~\bibnamefont
  {Buhrman}}, \bibinfo {author} {\bibfnamefont {O.}~\bibnamefont {Regev}},
  \bibinfo {author} {\bibfnamefont {G.}~\bibnamefont {Scarpa}},\ and\ \bibinfo
  {author} {\bibfnamefont {R.~d.}\ \bibnamefont {Wolf}},\ }\bibfield  {title}
  {\bibinfo {title} {Near-optimal and explicit bell inequality violations},\
  }\href {https://doi.org/10.4086/toc.2012.v008a027} {\bibfield  {journal}
  {\bibinfo  {journal} {Theory of Computing}\ }\textbf {\bibinfo {volume}
  {8}},\ \bibinfo {pages} {623} (\bibinfo {year} {2012})}\BibitemShut {NoStop}%
\bibitem [{\citenamefont {Quantinuum\mbox{\:}H1-1}(2025)}]{expt_data_full}%
  \BibitemOpen
  \bibfield  {author} {\bibinfo {author} {\bibnamefont
  {Quantinuum\mbox{\:}H1-1}},\ }\href {https://www.quantinuum.com/} {}
  (\bibinfo {year} {March 28, 2025 - April 17, 2025})\BibitemShut {NoStop}%
\bibitem [{\citenamefont {Pino}\ \emph {et~al.}(2021)\citenamefont {Pino} \emph
  {et~al.}}]{Pino2020}%
  \BibitemOpen
  \bibfield  {author} {\bibinfo {author} {\bibfnamefont {J.~M.}\ \bibnamefont
  {Pino}} \emph {et~al.},\ }\bibfield  {title} {\bibinfo {title} {Demonstration
  of the trapped-ion quantum ccd computer architecture},\ }\href
  {https://doi.org/10.1038/s41586-021-03318-4} {\bibfield  {journal} {\bibinfo
  {journal} {Nature}\ }\textbf {\bibinfo {volume} {592}},\ \bibinfo {pages}
  {209} (\bibinfo {year} {2021})}\BibitemShut {NoStop}%
\bibitem [{\citenamefont {Ryan-Anderson}\ \emph {et~al.}(2021)\citenamefont
  {Ryan-Anderson} \emph {et~al.}}]{RyanAnderson2021}%
  \BibitemOpen
  \bibfield  {author} {\bibinfo {author} {\bibfnamefont {C.}~\bibnamefont
  {Ryan-Anderson}} \emph {et~al.},\ }\bibfield  {title} {\bibinfo {title}
  {Realization of real-time fault-tolerant quantum error correction},\ }\href
  {https://link.aps.org/doi/10.1103/PhysRevX.11.041058} {\bibfield  {journal}
  {\bibinfo  {journal} {Phys. Rev. X}\ }\textbf {\bibinfo {volume} {11}},\
  \bibinfo {pages} {041058} (\bibinfo {year} {2021})}\BibitemShut {NoStop}%
\bibitem [{\citenamefont {Ryan-Anderson}\ \emph {et~al.}(2022)\citenamefont
  {Ryan-Anderson} \emph {et~al.}}]{RyanAnderson2022}%
  \BibitemOpen
  \bibfield  {author} {\bibinfo {author} {\bibfnamefont {C.}~\bibnamefont
  {Ryan-Anderson}} \emph {et~al.},\ }\bibfield  {title} {\bibinfo {title}
  {Implementing fault-tolerant entangling gates on the five-qubit code and the
  color code},\ }\href {https://arxiv.org/abs/2208.01863} {\bibfield  {journal}
  {\bibinfo  {journal} {arxiv:2208.01863}\ } (\bibinfo {year}
  {2022})}\BibitemShut {NoStop}%
\bibitem [{git(2025)}]{github_spec}%
  \BibitemOpen
  \href {https://github.com/CQCL/quantinuum-hardware-specifications} {\bibinfo
  {title} {Quantinuum hardware specificiations}} (\bibinfo {year}
  {2025})\BibitemShut {NoStop}%
\bibitem [{\citenamefont {Magesan}\ \emph
  {et~al.}(2011{\natexlab{a}})\citenamefont {Magesan}, \citenamefont
  {Gambetta},\ and\ \citenamefont {Emerson}}]{Magesan2011}%
  \BibitemOpen
  \bibfield  {author} {\bibinfo {author} {\bibfnamefont {E.}~\bibnamefont
  {Magesan}}, \bibinfo {author} {\bibfnamefont {J.~M.}\ \bibnamefont
  {Gambetta}},\ and\ \bibinfo {author} {\bibfnamefont {J.}~\bibnamefont
  {Emerson}},\ }\bibfield  {title} {\bibinfo {title} {Scalable and robust
  randomized benchmarking of quantum processes},\ }\href
  {https://doi.org/10.1103/PhysRevLett.106.180504} {\bibfield  {journal}
  {\bibinfo  {journal} {Phys. Rev. Lett.}\ }\textbf {\bibinfo {volume} {106}},\
  \bibinfo {pages} {180504} (\bibinfo {year} {2011}{\natexlab{a}})}\BibitemShut
  {NoStop}%
\bibitem [{\citenamefont {Proctor}\ \emph {et~al.}(2019)\citenamefont {Proctor}
  \emph {et~al.}}]{Proctor2019}%
  \BibitemOpen
  \bibfield  {author} {\bibinfo {author} {\bibfnamefont {T.~J.}\ \bibnamefont
  {Proctor}} \emph {et~al.},\ }\bibfield  {title} {\bibinfo {title} {Direct
  randomized benchmarking for multiqubit devices},\ }\href
  {https://doi.org/10.1103/PhysRevLett.123.030503} {\bibfield  {journal}
  {\bibinfo  {journal} {Phys. Rev. Lett.}\ }\textbf {\bibinfo {volume} {123}},\
  \bibinfo {pages} {030503} (\bibinfo {year} {2019})}\BibitemShut {NoStop}%
\bibitem [{\citenamefont {Iten}\ \emph {et~al.}(2016)\citenamefont {Iten},
  \citenamefont {Colbeck}, \citenamefont {Kukuljan}, \citenamefont {Home},\
  and\ \citenamefont {Christandl}}]{ICKHC16-isometries}%
  \BibitemOpen
  \bibfield  {author} {\bibinfo {author} {\bibfnamefont {R.}~\bibnamefont
  {Iten}}, \bibinfo {author} {\bibfnamefont {R.}~\bibnamefont {Colbeck}},
  \bibinfo {author} {\bibfnamefont {I.}~\bibnamefont {Kukuljan}}, \bibinfo
  {author} {\bibfnamefont {J.}~\bibnamefont {Home}},\ and\ \bibinfo {author}
  {\bibfnamefont {M.}~\bibnamefont {Christandl}},\ }\bibfield  {title}
  {\bibinfo {title} {Quantum circuits for isometries},\ }\href
  {https://doi.org/10.1103/PhysRevA.93.032318} {\bibfield  {journal} {\bibinfo
  {journal} {Phys. Rev. A}\ }\textbf {\bibinfo {volume} {93}},\ \bibinfo
  {pages} {032318} (\bibinfo {year} {2016})}\BibitemShut {NoStop}%
\bibitem [{\citenamefont {Haw}\ \emph {et~al.}(2015)\citenamefont {Haw} \emph
  {et~al.}}]{Haw-true-randomness}%
  \BibitemOpen
  \bibfield  {author} {\bibinfo {author} {\bibfnamefont {J.~Y.}\ \bibnamefont
  {Haw}} \emph {et~al.},\ }\bibfield  {title} {\bibinfo {title} {Maximization
  of extractable randomness in a quantum random-number generator},\ }\href
  {https://doi.org/10.1103/PhysRevApplied.3.054004} {\bibfield  {journal}
  {\bibinfo  {journal} {Phys. Rev. Appl.}\ }\textbf {\bibinfo {volume} {3}},\
  \bibinfo {pages} {054004} (\bibinfo {year} {2015})}\BibitemShut {NoStop}%
\bibitem [{\citenamefont {Kumar}\ \emph {et~al.}(2019)\citenamefont {Kumar},
  \citenamefont {Kerenidis},\ and\ \citenamefont {Diamanti}}]{KKD19-comm}%
  \BibitemOpen
  \bibfield  {author} {\bibinfo {author} {\bibfnamefont {N.}~\bibnamefont
  {Kumar}}, \bibinfo {author} {\bibfnamefont {I.}~\bibnamefont {Kerenidis}},\
  and\ \bibinfo {author} {\bibfnamefont {E.}~\bibnamefont {Diamanti}},\
  }\bibfield  {title} {\bibinfo {title} {Experimental demonstration of quantum
  advantage for one-way communication complexity surpassing best-known
  classical protocol},\ }\href {https://doi.org/10.1038/s41467-019-12139-z}
  {\bibfield  {journal} {\bibinfo  {journal} {Nature Communications}\ }\textbf
  {\bibinfo {volume} {10}},\ \bibinfo {pages} {4152} (\bibinfo {year}
  {2019})}\BibitemShut {NoStop}%
\bibitem [{\citenamefont {Scheidl}\ \emph {et~al.}(2010)\citenamefont {Scheidl}
  \emph {et~al.}}]{scheidl-freedom-of-choice-2010}%
  \BibitemOpen
  \bibfield  {author} {\bibinfo {author} {\bibfnamefont {T.}~\bibnamefont
  {Scheidl}} \emph {et~al.},\ }\bibfield  {title} {\bibinfo {title} {Violation
  of local realism with freedom of choice},\ }\href
  {https://doi.org/10.1073/pnas.1002780107} {\bibfield  {journal} {\bibinfo
  {journal} {Proceedings of the National Academy of Sciences}\ }\textbf
  {\bibinfo {volume} {107}},\ \bibinfo {pages} {19708} (\bibinfo {year}
  {2010})}\BibitemShut {NoStop}%
\bibitem [{\citenamefont {Brunner}\ \emph {et~al.}(2014)\citenamefont
  {Brunner}, \citenamefont {Cavalcanti}, \citenamefont {Pironio}, \citenamefont
  {Scarani},\ and\ \citenamefont {Wehner}}]{bell-nonlocality-brunner-2014}%
  \BibitemOpen
  \bibfield  {author} {\bibinfo {author} {\bibfnamefont {N.}~\bibnamefont
  {Brunner}}, \bibinfo {author} {\bibfnamefont {D.}~\bibnamefont {Cavalcanti}},
  \bibinfo {author} {\bibfnamefont {S.}~\bibnamefont {Pironio}}, \bibinfo
  {author} {\bibfnamefont {V.}~\bibnamefont {Scarani}},\ and\ \bibinfo {author}
  {\bibfnamefont {S.}~\bibnamefont {Wehner}},\ }\bibfield  {title} {\bibinfo
  {title} {Bell nonlocality},\ }\href
  {https://doi.org/10.1103/RevModPhys.86.419} {\bibfield  {journal} {\bibinfo
  {journal} {Rev. Mod. Phys.}\ }\textbf {\bibinfo {volume} {86}},\ \bibinfo
  {pages} {419} (\bibinfo {year} {2014})}\BibitemShut {NoStop}%
\bibitem [{\citenamefont {Kretschmer}(2021{\natexlab{a}})}]{Kre21-tsirelson}%
  \BibitemOpen
  \bibfield  {author} {\bibinfo {author} {\bibfnamefont {W.}~\bibnamefont
  {Kretschmer}},\ }\bibfield  {title} {\bibinfo {title} {The {Q}uantum
  {S}upremacy {T}sirelson {I}nequality},\ }\href
  {https://doi.org/10.22331/q-2021-10-07-560} {\bibfield  {journal} {\bibinfo
  {journal} {{Quantum}}\ }\textbf {\bibinfo {volume} {5}},\ \bibinfo {pages}
  {560} (\bibinfo {year} {2021}{\natexlab{a}})}\BibitemShut {NoStop}%
\bibitem [{\citenamefont {Aaronson}\ and\ \citenamefont
  {Hung}(2023)}]{AH23-random}%
  \BibitemOpen
  \bibfield  {author} {\bibinfo {author} {\bibfnamefont {S.}~\bibnamefont
  {Aaronson}}\ and\ \bibinfo {author} {\bibfnamefont {S.-H.}\ \bibnamefont
  {Hung}},\ }\bibfield  {title} {\bibinfo {title} {Certified randomness from
  quantum supremacy},\ }in\ \href {https://doi.org/10.1145/3564246.3585145}
  {\emph {\bibinfo {booktitle} {Proceedings of the 55th Annual ACM Symposium on
  Theory of Computing}}},\ \bibinfo {series and number} {STOC 2023}\ (\bibinfo
  {publisher} {Association for Computing Machinery},\ \bibinfo {address} {New
  York, NY, USA},\ \bibinfo {year} {2023})\ pp.\ \bibinfo {pages}
  {933--944}\BibitemShut {NoStop}%
\bibitem [{\citenamefont {Gao}\ \emph {et~al.}(2024)\citenamefont {Gao},
  \citenamefont {Kalinowski}, \citenamefont {Chou}, \citenamefont {Lukin},
  \citenamefont {Barak},\ and\ \citenamefont {Choi}}]{GKC+21-xeb}%
  \BibitemOpen
  \bibfield  {author} {\bibinfo {author} {\bibfnamefont {X.}~\bibnamefont
  {Gao}}, \bibinfo {author} {\bibfnamefont {M.}~\bibnamefont {Kalinowski}},
  \bibinfo {author} {\bibfnamefont {C.-N.}\ \bibnamefont {Chou}}, \bibinfo
  {author} {\bibfnamefont {M.~D.}\ \bibnamefont {Lukin}}, \bibinfo {author}
  {\bibfnamefont {B.}~\bibnamefont {Barak}},\ and\ \bibinfo {author}
  {\bibfnamefont {S.}~\bibnamefont {Choi}},\ }\bibfield  {title} {\bibinfo
  {title} {Limitations of linear cross-entropy as a measure for quantum
  advantage},\ }\href {https://doi.org/10.1103/PRXQuantum.5.010334} {\bibfield
  {journal} {\bibinfo  {journal} {PRX Quantum}\ }\textbf {\bibinfo {volume}
  {5}},\ \bibinfo {pages} {010334} (\bibinfo {year} {2024})}\BibitemShut
  {NoStop}%
\bibitem [{\citenamefont {Dalzell}\ \emph {et~al.}(2024)\citenamefont
  {Dalzell}, \citenamefont {Hunter-Jones},\ and\ \citenamefont
  {Brand{\~a}o}}]{DHJB24-noise}%
  \BibitemOpen
  \bibfield  {author} {\bibinfo {author} {\bibfnamefont {A.~M.}\ \bibnamefont
  {Dalzell}}, \bibinfo {author} {\bibfnamefont {N.}~\bibnamefont
  {Hunter-Jones}},\ and\ \bibinfo {author} {\bibfnamefont {F.~G. S.~L.}\
  \bibnamefont {Brand{\~a}o}},\ }\bibfield  {title} {\bibinfo {title} {Random
  quantum circuits transform local noise into global white noise},\ }\href
  {https://doi.org/10.1007/s00220-024-04958-z} {\bibfield  {journal} {\bibinfo
  {journal} {Communications in Mathematical Physics}\ }\textbf {\bibinfo
  {volume} {405}},\ \bibinfo {pages} {78} (\bibinfo {year} {2024})}\BibitemShut
  {NoStop}%
\bibitem [{\citenamefont {Aaronson}\ and\ \citenamefont
  {Gottesman}(2004)}]{AG04-stabilizer}%
  \BibitemOpen
  \bibfield  {author} {\bibinfo {author} {\bibfnamefont {S.}~\bibnamefont
  {Aaronson}}\ and\ \bibinfo {author} {\bibfnamefont {D.}~\bibnamefont
  {Gottesman}},\ }\bibfield  {title} {\bibinfo {title} {Improved simulation of
  stabilizer circuits},\ }\href {https://doi.org/10.1103/PhysRevA.70.052328}
  {\bibfield  {journal} {\bibinfo  {journal} {Phys. Rev. A}\ }\textbf {\bibinfo
  {volume} {70}},\ \bibinfo {pages} {052328} (\bibinfo {year}
  {2004})}\BibitemShut {NoStop}%
\bibitem [{\citenamefont {Montanaro}(2011)}]{Mon11-comm}%
  \BibitemOpen
  \bibfield  {author} {\bibinfo {author} {\bibfnamefont {A.}~\bibnamefont
  {Montanaro}},\ }\bibfield  {title} {\bibinfo {title} {A new exponential
  separation between quantum and classical one-way communication complexity},\
  }\href {https://doi.org/10.26421/QIC11.7-8-3} {\bibfield  {journal} {\bibinfo
   {journal} {Quantum Info. Comput.}\ }\textbf {\bibinfo {volume} {11}},\
  \bibinfo {pages} {574} (\bibinfo {year} {2011})}\BibitemShut {NoStop}%
\bibitem [{\citenamefont {Verbin}\ and\ \citenamefont
  {Yu}(2011)}]{VY11-stream}%
  \BibitemOpen
  \bibfield  {author} {\bibinfo {author} {\bibfnamefont {E.}~\bibnamefont
  {Verbin}}\ and\ \bibinfo {author} {\bibfnamefont {W.}~\bibnamefont {Yu}},\
  }\bibfield  {title} {\bibinfo {title} {The streaming complexity of cycle
  counting, sorting by reversals, and other problems},\ }in\ \href
  {https://doi.org/10.1137/1.9781611973082.2} {\emph {\bibinfo {booktitle}
  {Proceedings of the Twenty-Second Annual ACM-SIAM Symposium on Discrete
  Algorithms}}},\ \bibinfo {series and number} {SODA '11}\ (\bibinfo
  {publisher} {Society for Industrial and Applied Mathematics},\ \bibinfo
  {address} {USA},\ \bibinfo {year} {2011})\ pp.\ \bibinfo {pages}
  {11--25}\BibitemShut {NoStop}%
\bibitem [{\citenamefont {Shi}\ \emph {et~al.}(2012)\citenamefont {Shi},
  \citenamefont {Wu},\ and\ \citenamefont {Yu}}]{SWY12-hyper}%
  \BibitemOpen
  \bibfield  {author} {\bibinfo {author} {\bibfnamefont {Y.}~\bibnamefont
  {Shi}}, \bibinfo {author} {\bibfnamefont {X.}~\bibnamefont {Wu}},\ and\
  \bibinfo {author} {\bibfnamefont {W.}~\bibnamefont {Yu}},\ }\href@noop {}
  {\bibinfo {title} {Limits of quantum one-way communication by matrix
  hypercontractive inequality}},\ \bibinfo {howpublished} {Online} (\bibinfo
  {year} {2012})\BibitemShut {NoStop}%
\bibitem [{\citenamefont {Doriguello}\ and\ \citenamefont
  {Montanaro}(2020)}]{DM20-bhm}%
  \BibitemOpen
  \bibfield  {author} {\bibinfo {author} {\bibfnamefont {J.~F.}\ \bibnamefont
  {Doriguello}}\ and\ \bibinfo {author} {\bibfnamefont {A.}~\bibnamefont
  {Montanaro}},\ }\bibfield  {title} {\bibinfo {title} {{Exponential Quantum
  Communication Reductions from Generalizations of the Boolean Hidden Matching
  Problem}},\ }in\ \href {https://doi.org/10.4230/LIPIcs.TQC.2020.1} {\emph
  {\bibinfo {booktitle} {15th Conference on the Theory of Quantum Computation,
  Communication and Cryptography (TQC 2020)}}},\ \bibinfo {series} {Leibniz
  International Proceedings in Informatics (LIPIcs)}, Vol.\ \bibinfo {volume}
  {158},\ \bibinfo {editor} {edited by\ \bibinfo {editor} {\bibfnamefont
  {S.~T.}\ \bibnamefont {Flammia}}}\ (\bibinfo  {publisher} {Schloss Dagstuhl
  -- Leibniz-Zentrum f{\"u}r Informatik},\ \bibinfo {address} {Dagstuhl,
  Germany},\ \bibinfo {year} {2020})\ pp.\ \bibinfo {pages}
  {1:1--1:16}\BibitemShut {NoStop}%
\bibitem [{\citenamefont {Bennett}\ and\ \citenamefont
  {Wiesner}(1992)}]{Bennett-superdense}%
  \BibitemOpen
  \bibfield  {author} {\bibinfo {author} {\bibfnamefont {C.~H.}\ \bibnamefont
  {Bennett}}\ and\ \bibinfo {author} {\bibfnamefont {S.~J.}\ \bibnamefont
  {Wiesner}},\ }\bibfield  {title} {\bibinfo {title} {Communication via one-
  and two-particle operators on einstein-podolsky-rosen states},\ }\href
  {https://doi.org/10.1103/PhysRevLett.69.2881} {\bibfield  {journal} {\bibinfo
   {journal} {Phys. Rev. Lett.}\ }\textbf {\bibinfo {volume} {69}},\ \bibinfo
  {pages} {2881} (\bibinfo {year} {1992})}\BibitemShut {NoStop}%
\bibitem [{\citenamefont {Ambainis}\ \emph {et~al.}(1999)\citenamefont
  {Ambainis}, \citenamefont {Nayak}, \citenamefont {Ta-Shma},\ and\
  \citenamefont {Vazirani}}]{ANTV99-qrac}%
  \BibitemOpen
  \bibfield  {author} {\bibinfo {author} {\bibfnamefont {A.}~\bibnamefont
  {Ambainis}}, \bibinfo {author} {\bibfnamefont {A.}~\bibnamefont {Nayak}},
  \bibinfo {author} {\bibfnamefont {A.}~\bibnamefont {Ta-Shma}},\ and\ \bibinfo
  {author} {\bibfnamefont {U.}~\bibnamefont {Vazirani}},\ }\bibfield  {title}
  {\bibinfo {title} {Dense quantum coding and a lower bound for 1-way quantum
  automata},\ }in\ \href {https://doi.org/10.1145/301250.301347} {\emph
  {\bibinfo {booktitle} {Proceedings of the Thirty-First Annual ACM Symposium
  on Theory of Computing}}},\ \bibinfo {series and number} {STOC '99}\
  (\bibinfo  {publisher} {Association for Computing Machinery},\ \bibinfo
  {address} {New York, NY, USA},\ \bibinfo {year} {1999})\ pp.\ \bibinfo
  {pages} {376--383}\BibitemShut {NoStop}%
\bibitem [{\citenamefont {Nayak}(1999)}]{Nay99-qrac}%
  \BibitemOpen
  \bibfield  {author} {\bibinfo {author} {\bibfnamefont {A.}~\bibnamefont
  {Nayak}},\ }\bibfield  {title} {\bibinfo {title} {Optimal lower bounds for
  quantum automata and random access codes},\ }in\ \href
  {https://doi.org/10.1109/SFFCS.1999.814608} {\emph {\bibinfo {booktitle}
  {40th Annual Symposium on Foundations of Computer Science (Cat.
  No.99CB37039)}}}\ (\bibinfo {year} {1999})\ pp.\ \bibinfo {pages}
  {369--376}\BibitemShut {NoStop}%
\bibitem [{\citenamefont {Ambainis}\ \emph {et~al.}(2002)\citenamefont
  {Ambainis}, \citenamefont {Nayak}, \citenamefont {Ta-Shma},\ and\
  \citenamefont {Vazirani}}]{ANTV02-qrac-combined}%
  \BibitemOpen
  \bibfield  {author} {\bibinfo {author} {\bibfnamefont {A.}~\bibnamefont
  {Ambainis}}, \bibinfo {author} {\bibfnamefont {A.}~\bibnamefont {Nayak}},
  \bibinfo {author} {\bibfnamefont {A.}~\bibnamefont {Ta-Shma}},\ and\ \bibinfo
  {author} {\bibfnamefont {U.}~\bibnamefont {Vazirani}},\ }\bibfield  {title}
  {\bibinfo {title} {Dense quantum coding and quantum finite automata},\ }\href
  {https://doi.org/10.1145/581771.581773} {\bibfield  {journal} {\bibinfo
  {journal} {J. ACM}\ }\textbf {\bibinfo {volume} {49}},\ \bibinfo {pages}
  {496} (\bibinfo {year} {2002})}\BibitemShut {NoStop}%
\bibitem [{\citenamefont {Bennett}\ \emph {et~al.}(2002)\citenamefont
  {Bennett}, \citenamefont {Shor}, \citenamefont {Smolin},\ and\ \citenamefont
  {Thapliyal}}]{BSST02-capacity}%
  \BibitemOpen
  \bibfield  {author} {\bibinfo {author} {\bibfnamefont {C.}~\bibnamefont
  {Bennett}}, \bibinfo {author} {\bibfnamefont {P.}~\bibnamefont {Shor}},
  \bibinfo {author} {\bibfnamefont {J.}~\bibnamefont {Smolin}},\ and\ \bibinfo
  {author} {\bibfnamefont {A.}~\bibnamefont {Thapliyal}},\ }\bibfield  {title}
  {\bibinfo {title} {Entanglement-assisted capacity of a quantum channel and
  the reverse {S}hannon theorem},\ }\href
  {https://doi.org/10.1109/TIT.2002.802612} {\bibfield  {journal} {\bibinfo
  {journal} {IEEE Transactions on Information Theory}\ }\textbf {\bibinfo
  {volume} {48}},\ \bibinfo {pages} {2637} (\bibinfo {year}
  {2002})}\BibitemShut {NoStop}%
\bibitem [{\citenamefont {Lin}\ and\ \citenamefont
  {de~Wolf}(2025)}]{LdW25-qrac}%
  \BibitemOpen
  \bibfield  {author} {\bibinfo {author} {\bibfnamefont {H.-H.}\ \bibnamefont
  {Lin}}\ and\ \bibinfo {author} {\bibfnamefont {R.}~\bibnamefont {de~Wolf}},\
  }\href@noop {} {\bibinfo {title} {Getting almost all the bits from a quantum
  random access code}} (\bibinfo {year} {2025}),\ \Eprint
  {https://arxiv.org/abs/2506.01903} {arXiv:2506.01903 [quant-ph]} \BibitemShut
  {NoStop}%
\bibitem [{\citenamefont {Centrone}\ \emph {et~al.}(2021)\citenamefont
  {Centrone}, \citenamefont {Kumar}, \citenamefont {Diamanti},\ and\
  \citenamefont {Kerenidis}}]{CKDK21-np-ver}%
  \BibitemOpen
  \bibfield  {author} {\bibinfo {author} {\bibfnamefont {F.}~\bibnamefont
  {Centrone}}, \bibinfo {author} {\bibfnamefont {N.}~\bibnamefont {Kumar}},
  \bibinfo {author} {\bibfnamefont {E.}~\bibnamefont {Diamanti}},\ and\
  \bibinfo {author} {\bibfnamefont {I.}~\bibnamefont {Kerenidis}},\ }\bibfield
  {title} {\bibinfo {title} {Experimental demonstration of quantum advantage
  for {NP} verification with limited information},\ }\href
  {https://doi.org/10.1038/s41467-021-21119-1} {\bibfield  {journal} {\bibinfo
  {journal} {Nature Communications}\ }\textbf {\bibinfo {volume} {12}},\
  \bibinfo {pages} {850} (\bibinfo {year} {2021})}\BibitemShut {NoStop}%
\bibitem [{\citenamefont {Aaronson}\ \emph {et~al.}(2009)\citenamefont
  {Aaronson}, \citenamefont {Beigi}, \citenamefont {Drucker}, \citenamefont
  {Fefferman},\ and\ \citenamefont {Shor}}]{ABDFS09-unentanglement}%
  \BibitemOpen
  \bibfield  {author} {\bibinfo {author} {\bibfnamefont {S.}~\bibnamefont
  {Aaronson}}, \bibinfo {author} {\bibfnamefont {S.}~\bibnamefont {Beigi}},
  \bibinfo {author} {\bibfnamefont {A.}~\bibnamefont {Drucker}}, \bibinfo
  {author} {\bibfnamefont {B.}~\bibnamefont {Fefferman}},\ and\ \bibinfo
  {author} {\bibfnamefont {P.}~\bibnamefont {Shor}},\ }\bibfield  {title}
  {\bibinfo {title} {The power of unentanglement},\ }\href
  {https://doi.org/10.4086/toc.2009.v005a001} {\bibfield  {journal} {\bibinfo
  {journal} {Theory of Computing}\ }\textbf {\bibinfo {volume} {5}},\ \bibinfo
  {pages} {1} (\bibinfo {year} {2009})}\BibitemShut {NoStop}%
\bibitem [{\citenamefont {Huang}\ \emph {et~al.}(2021)\citenamefont {Huang},
  \citenamefont {Kueng},\ and\ \citenamefont {Preskill}}]{HKP21-info-bounds}%
  \BibitemOpen
  \bibfield  {author} {\bibinfo {author} {\bibfnamefont {H.-Y.}\ \bibnamefont
  {Huang}}, \bibinfo {author} {\bibfnamefont {R.}~\bibnamefont {Kueng}},\ and\
  \bibinfo {author} {\bibfnamefont {J.}~\bibnamefont {Preskill}},\ }\bibfield
  {title} {\bibinfo {title} {Information-theoretic bounds on quantum advantage
  in machine learning},\ }\href
  {https://doi.org/10.1103/PhysRevLett.126.190505} {\bibfield  {journal}
  {\bibinfo  {journal} {Phys. Rev. Lett.}\ }\textbf {\bibinfo {volume} {126}},\
  \bibinfo {pages} {190505} (\bibinfo {year} {2021})}\BibitemShut {NoStop}%
\bibitem [{\citenamefont {Aharonov}\ \emph {et~al.}(2022)\citenamefont
  {Aharonov}, \citenamefont {Cotler},\ and\ \citenamefont
  {Qi}}]{ACQ22-measure}%
  \BibitemOpen
  \bibfield  {author} {\bibinfo {author} {\bibfnamefont {D.}~\bibnamefont
  {Aharonov}}, \bibinfo {author} {\bibfnamefont {J.}~\bibnamefont {Cotler}},\
  and\ \bibinfo {author} {\bibfnamefont {X.-L.}\ \bibnamefont {Qi}},\
  }\bibfield  {title} {\bibinfo {title} {Quantum algorithmic measurement},\
  }\href {https://doi.org/10.1038/s41467-021-27922-0} {\bibfield  {journal}
  {\bibinfo  {journal} {Nature Communications}\ }\textbf {\bibinfo {volume}
  {13}},\ \bibinfo {pages} {887} (\bibinfo {year} {2022})}\BibitemShut
  {NoStop}%
\bibitem [{\citenamefont {Chen}\ \emph {et~al.}(2022)\citenamefont {Chen},
  \citenamefont {Cotler}, \citenamefont {Huang},\ and\ \citenamefont
  {Li}}]{CCHL22-memory}%
  \BibitemOpen
  \bibfield  {author} {\bibinfo {author} {\bibfnamefont {S.}~\bibnamefont
  {Chen}}, \bibinfo {author} {\bibfnamefont {J.}~\bibnamefont {Cotler}},
  \bibinfo {author} {\bibfnamefont {H.-Y.}\ \bibnamefont {Huang}},\ and\
  \bibinfo {author} {\bibfnamefont {J.}~\bibnamefont {Li}},\ }\bibfield
  {title} {\bibinfo {title} {Exponential separations between learning with and
  without quantum memory},\ }in\ \href
  {https://doi.org/10.1109/FOCS52979.2021.00063} {\emph {\bibinfo {booktitle}
  {2021 IEEE 62nd Annual Symposium on Foundations of Computer Science
  (FOCS)}}}\ (\bibinfo {year} {2022})\ pp.\ \bibinfo {pages}
  {574--585}\BibitemShut {NoStop}%
\bibitem [{\citenamefont {Huang}\ \emph {et~al.}(2022)\citenamefont {Huang}
  \emph {et~al.}}]{HBC+22-experiments}%
  \BibitemOpen
  \bibfield  {author} {\bibinfo {author} {\bibfnamefont {H.-Y.}\ \bibnamefont
  {Huang}} \emph {et~al.},\ }\bibfield  {title} {\bibinfo {title} {Quantum
  advantage in learning from experiments},\ }\href
  {https://doi.org/10.1126/science.abn7293} {\bibfield  {journal} {\bibinfo
  {journal} {Science}\ }\textbf {\bibinfo {volume} {376}},\ \bibinfo {pages}
  {1182} (\bibinfo {year} {2022})}\BibitemShut {NoStop}%
\bibitem [{\citenamefont {Brand{\~a}o}\ \emph {et~al.}(2016)\citenamefont
  {Brand{\~a}o}, \citenamefont {Harrow},\ and\ \citenamefont
  {Horodecki}}]{BHH16-designs}%
  \BibitemOpen
  \bibfield  {author} {\bibinfo {author} {\bibfnamefont {F.~G. S.~L.}\
  \bibnamefont {Brand{\~a}o}}, \bibinfo {author} {\bibfnamefont {A.~W.}\
  \bibnamefont {Harrow}},\ and\ \bibinfo {author} {\bibfnamefont
  {M.}~\bibnamefont {Horodecki}},\ }\bibfield  {title} {\bibinfo {title} {Local
  random quantum circuits are approximate polynomial-designs},\ }\href
  {https://doi.org/10.1007/s00220-016-2706-8} {\bibfield  {journal} {\bibinfo
  {journal} {Communications in Mathematical Physics}\ }\textbf {\bibinfo
  {volume} {346}},\ \bibinfo {pages} {397} (\bibinfo {year}
  {2016})}\BibitemShut {NoStop}%
\bibitem [{\citenamefont {Vershynin}(2018)}]{Ver18-hdp}%
  \BibitemOpen
  \bibfield  {author} {\bibinfo {author} {\bibfnamefont {R.}~\bibnamefont
  {Vershynin}},\ }\href {https://doi.org/10.1017/9781108231596} {\emph
  {\bibinfo {title} {High-Dimensional Probability: An Introduction with
  Applications in Data Science}}},\ Cambridge Series in Statistical and
  Probabilistic Mathematics\ (\bibinfo  {publisher} {Cambridge University
  Press},\ \bibinfo {year} {2018})\BibitemShut {NoStop}%
\bibitem [{\citenamefont {Pinelis}(2022)}]{Pin22-subexp}%
  \BibitemOpen
  \bibfield  {author} {\bibinfo {author} {\bibfnamefont {I.}~\bibnamefont
  {Pinelis}},\ }\bibfield  {title} {\bibinfo {title} {Improved concentration
  bounds for sums of independent sub-exponential random variables},\ }\href
  {https://doi.org/10.1016/j.spl.2022.109666} {\bibfield  {journal} {\bibinfo
  {journal} {Statistics \& Probability Letters}\ }\textbf {\bibinfo {volume}
  {191}},\ \bibinfo {pages} {109666} (\bibinfo {year} {2022})}\BibitemShut
  {NoStop}%
\bibitem [{\citenamefont {Dankert}\ \emph {et~al.}(2009)\citenamefont
  {Dankert}, \citenamefont {Cleve}, \citenamefont {Emerson},\ and\
  \citenamefont {Livine}}]{DCEL09-design}%
  \BibitemOpen
  \bibfield  {author} {\bibinfo {author} {\bibfnamefont {C.}~\bibnamefont
  {Dankert}}, \bibinfo {author} {\bibfnamefont {R.}~\bibnamefont {Cleve}},
  \bibinfo {author} {\bibfnamefont {J.}~\bibnamefont {Emerson}},\ and\ \bibinfo
  {author} {\bibfnamefont {E.}~\bibnamefont {Livine}},\ }\bibfield  {title}
  {\bibinfo {title} {Exact and approximate unitary 2-designs and their
  application to fidelity estimation},\ }\href
  {https://doi.org/10.1103/PhysRevA.80.012304} {\bibfield  {journal} {\bibinfo
  {journal} {Phys. Rev. A}\ }\textbf {\bibinfo {volume} {80}},\ \bibinfo
  {pages} {012304} (\bibinfo {year} {2009})}\BibitemShut {NoStop}%
\bibitem [{\citenamefont {Webb}(2016)}]{Web16-design}%
  \BibitemOpen
  \bibfield  {author} {\bibinfo {author} {\bibfnamefont {Z.}~\bibnamefont
  {Webb}},\ }\bibfield  {title} {\bibinfo {title} {The {C}lifford group forms a
  unitary 3-design},\ }\href {https://doi.org/10.26421/QIC16.15-16-8}
  {\bibfield  {journal} {\bibinfo  {journal} {Quantum Info. Comput.}\ }\textbf
  {\bibinfo {volume} {16}},\ \bibinfo {pages} {1379} (\bibinfo {year}
  {2016})}\BibitemShut {NoStop}%
\bibitem [{\citenamefont {{Zhu}}\ \emph {et~al.}(2016)\citenamefont {{Zhu}},
  \citenamefont {{Kueng}}, \citenamefont {{Grassl}},\ and\ \citenamefont
  {{Gross}}}]{Zhu16gracefully}%
  \BibitemOpen
  \bibfield  {author} {\bibinfo {author} {\bibfnamefont {H.}~\bibnamefont
  {{Zhu}}}, \bibinfo {author} {\bibfnamefont {R.}~\bibnamefont {{Kueng}}},
  \bibinfo {author} {\bibfnamefont {M.}~\bibnamefont {{Grassl}}},\ and\
  \bibinfo {author} {\bibfnamefont {D.}~\bibnamefont {{Gross}}},\ }\href@noop
  {} {\bibinfo {title} {{The Clifford group fails gracefully to be a unitary
  4-design}}} (\bibinfo {year} {2016}),\ \Eprint
  {https://arxiv.org/abs/1609.08172} {arXiv:1609.08172 [quant-ph]} \BibitemShut
  {NoStop}%
\bibitem [{\citenamefont {{Gross}}\ \emph {et~al.}(2021)\citenamefont
  {{Gross}}, \citenamefont {{Nezami}},\ and\ \citenamefont
  {{Walter}}}]{Gross17SWClifford}%
  \BibitemOpen
  \bibfield  {author} {\bibinfo {author} {\bibfnamefont {D.}~\bibnamefont
  {{Gross}}}, \bibinfo {author} {\bibfnamefont {S.}~\bibnamefont {{Nezami}}},\
  and\ \bibinfo {author} {\bibfnamefont {M.}~\bibnamefont {{Walter}}},\
  }\bibfield  {title} {\bibinfo {title} {{Schur-Weyl Duality for the Clifford
  Group with Applications: Property Testing, a Robust Hudson Theorem, and de
  Finetti Representations}},\ }\href
  {https://doi.org/10.1007/s00220-021-04118-7} {\bibfield  {journal} {\bibinfo
  {journal} {Communications in Mathematical Physics}\ }\textbf {\bibinfo
  {volume} {385}},\ \bibinfo {pages} {1325} (\bibinfo {year}
  {2021})}\BibitemShut {NoStop}%
\bibitem [{\citenamefont {Bittel}\ \emph {et~al.}(2025)\citenamefont {Bittel},
  \citenamefont {Eisert}, \citenamefont {Leone}, \citenamefont {Mele},\ and\
  \citenamefont {Oliviero}}]{BELMO25-clifford}%
  \BibitemOpen
  \bibfield  {author} {\bibinfo {author} {\bibfnamefont {L.}~\bibnamefont
  {Bittel}}, \bibinfo {author} {\bibfnamefont {J.}~\bibnamefont {Eisert}},
  \bibinfo {author} {\bibfnamefont {L.}~\bibnamefont {Leone}}, \bibinfo
  {author} {\bibfnamefont {A.~A.}\ \bibnamefont {Mele}},\ and\ \bibinfo
  {author} {\bibfnamefont {S.~F.~E.}\ \bibnamefont {Oliviero}},\ }\href@noop {}
  {\bibinfo {title} {A complete theory of the clifford commutant}} (\bibinfo
  {year} {2025}),\ \Eprint {https://arxiv.org/abs/2504.12263} {arXiv:2504.12263
  [quant-ph]} \BibitemShut {NoStop}%
\bibitem [{\citenamefont {Leone}\ \emph {et~al.}(2025)\citenamefont {Leone},
  \citenamefont {Oliviero}, \citenamefont {Hamma}, \citenamefont {Eisert},\
  and\ \citenamefont {Bittel}}]{LOHEB25-clifford}%
  \BibitemOpen
  \bibfield  {author} {\bibinfo {author} {\bibfnamefont {L.}~\bibnamefont
  {Leone}}, \bibinfo {author} {\bibfnamefont {S.~F.~E.}\ \bibnamefont
  {Oliviero}}, \bibinfo {author} {\bibfnamefont {A.}~\bibnamefont {Hamma}},
  \bibinfo {author} {\bibfnamefont {J.}~\bibnamefont {Eisert}},\ and\ \bibinfo
  {author} {\bibfnamefont {L.}~\bibnamefont {Bittel}},\ }\href@noop {}
  {\bibinfo {title} {The non-{C}lifford cost of random unitaries}} (\bibinfo
  {year} {2025}),\ \Eprint {https://arxiv.org/abs/2505.10110} {arXiv:2505.10110
  [quant-ph]} \BibitemShut {NoStop}%
\bibitem [{\citenamefont
  {Kretschmer}(2021{\natexlab{b}})}]{Kre21-pseudorandom}%
  \BibitemOpen
  \bibfield  {author} {\bibinfo {author} {\bibfnamefont {W.}~\bibnamefont
  {Kretschmer}},\ }\bibfield  {title} {\bibinfo {title} {{Quantum
  Pseudorandomness and Classical Complexity}},\ }in\ \href
  {https://doi.org/10.4230/LIPIcs.TQC.2021.2} {\emph {\bibinfo {booktitle}
  {16th Conference on the Theory of Quantum Computation, Communication and
  Cryptography (TQC 2021)}}},\ \bibinfo {series} {Leibniz International
  Proceedings in Informatics (LIPIcs)}, Vol.\ \bibinfo {volume} {197},\
  \bibinfo {editor} {edited by\ \bibinfo {editor} {\bibfnamefont {M.-H.}\
  \bibnamefont {Hsieh}}}\ (\bibinfo  {publisher} {Schloss Dagstuhl --
  Leibniz-Zentrum f{\"u}r Informatik},\ \bibinfo {address} {Dagstuhl,
  Germany},\ \bibinfo {year} {2021})\ pp.\ \bibinfo {pages} {2:1--2:20},\
  \Eprint {https://arxiv.org/abs/2103.09320v5} {2103.09320v5} \BibitemShut
  {NoStop}%
\bibitem [{\citenamefont
  {{Kaban-5\mbox{\:}(https://mathoverflow.net/users/126017/kaban-5)}}(2018)}]{Kab18-simplex}%
  \BibitemOpen
  \bibfield  {author} {\bibinfo {author} {\bibnamefont
  {{Kaban-5\mbox{\:}(https://mathoverflow.net/users/126017/kaban-5)}}},\
  }\href@noop {} {\bibinfo {title} {What is the probability distribution of the
  $k$th largest coordinate chosen over a simplex?}},\ \bibinfo {howpublished}
  {MathOverflow} (\bibinfo {year} {2018}),\ \Eprint
  {https://arxiv.org/abs/https://mathoverflow.net/q/312594}
  {https://mathoverflow.net/q/312594} \BibitemShut {NoStop}%
\bibitem [{\citenamefont {R{\'e}nyi}(1953)}]{Ren53-order}%
  \BibitemOpen
  \bibfield  {author} {\bibinfo {author} {\bibfnamefont {A.}~\bibnamefont
  {R{\'e}nyi}},\ }\bibfield  {title} {\bibinfo {title} {On the theory of order
  statistics},\ }\href {https://doi.org/10.1007/BF02127580} {\bibfield
  {journal} {\bibinfo  {journal} {Acta Mathematica Academiae Scientiarum
  Hungarica}\ }\textbf {\bibinfo {volume} {4}},\ \bibinfo {pages} {191}
  (\bibinfo {year} {1953})}\BibitemShut {NoStop}%
\bibitem [{\citenamefont {Renner}(2008)}]{Ren08-qkd}%
  \BibitemOpen
  \bibfield  {author} {\bibinfo {author} {\bibfnamefont {R.}~\bibnamefont
  {Renner}},\ }\bibfield  {title} {\bibinfo {title} {Security of quantum key
  distribution},\ }\href {https://doi.org/10.1142/S0219749908003256} {\bibfield
   {journal} {\bibinfo  {journal} {International Journal of Quantum
  Information}\ }\textbf {\bibinfo {volume} {06}},\ \bibinfo {pages} {1}
  (\bibinfo {year} {2008})}\BibitemShut {NoStop}%
\bibitem [{\citenamefont {Schuster}\ \emph {et~al.}(2025)\citenamefont
  {Schuster}, \citenamefont {Haferkamp},\ and\ \citenamefont
  {Huang}}]{SHH24-designs}%
  \BibitemOpen
  \bibfield  {author} {\bibinfo {author} {\bibfnamefont {T.}~\bibnamefont
  {Schuster}}, \bibinfo {author} {\bibfnamefont {J.}~\bibnamefont
  {Haferkamp}},\ and\ \bibinfo {author} {\bibfnamefont {H.-Y.}\ \bibnamefont
  {Huang}},\ }\bibfield  {title} {\bibinfo {title} {Random unitaries in
  extremely low depth},\ }\href {https://doi.org/10.1126/science.adv8590}
  {\bibfield  {journal} {\bibinfo  {journal} {Science}\ }\textbf {\bibinfo
  {volume} {389}},\ \bibinfo {pages} {92} (\bibinfo {year} {2025})}\BibitemShut
  {NoStop}%
\bibitem [{\citenamefont {Krol}\ and\ \citenamefont {Al-Ars}(2024)}]{KA24-zxz}%
  \BibitemOpen
  \bibfield  {author} {\bibinfo {author} {\bibfnamefont {A.~M.}\ \bibnamefont
  {Krol}}\ and\ \bibinfo {author} {\bibfnamefont {Z.}~\bibnamefont {Al-Ars}},\
  }\bibfield  {title} {\bibinfo {title} {Beyond quantum {S}hannon
  decomposition: Circuit construction for $n$-qubit gates based on
  block-{$ZXZ$} decomposition},\ }\href
  {https://doi.org/10.1103/PhysRevApplied.22.034019} {\bibfield  {journal}
  {\bibinfo  {journal} {Phys. Rev. Appl.}\ }\textbf {\bibinfo {volume} {22}},\
  \bibinfo {pages} {034019} (\bibinfo {year} {2024})}\BibitemShut {NoStop}%
\bibitem [{\citenamefont {Moses}\ \emph {et~al.}(2023)\citenamefont {Moses}
  \emph {et~al.}}]{PhysRevX.13.041052}%
  \BibitemOpen
  \bibfield  {author} {\bibinfo {author} {\bibfnamefont {S.~A.}\ \bibnamefont
  {Moses}} \emph {et~al.},\ }\bibfield  {title} {\bibinfo {title} {A race-track
  trapped-ion quantum processor},\ }\href
  {https://doi.org/10.1103/PhysRevX.13.041052} {\bibfield  {journal} {\bibinfo
  {journal} {Phys. Rev. X}\ }\textbf {\bibinfo {volume} {13}},\ \bibinfo
  {pages} {041052} (\bibinfo {year} {2023})}\BibitemShut {NoStop}%
\bibitem [{\citenamefont {Olmschenk}\ \emph {et~al.}(2007)\citenamefont
  {Olmschenk} \emph {et~al.}}]{PhysRevA.76.052314}%
  \BibitemOpen
  \bibfield  {author} {\bibinfo {author} {\bibfnamefont {S.}~\bibnamefont
  {Olmschenk}} \emph {et~al.},\ }\bibfield  {title} {\bibinfo {title}
  {Manipulation and detection of a trapped ${\mathrm{yb}}^{+}$ hyperfine
  qubit},\ }\href {https://doi.org/10.1103/PhysRevA.76.052314} {\bibfield
  {journal} {\bibinfo  {journal} {Phys. Rev. A}\ }\textbf {\bibinfo {volume}
  {76}},\ \bibinfo {pages} {052314} (\bibinfo {year} {2007})}\BibitemShut
  {NoStop}%
\bibitem [{\citenamefont {S\o{}rensen}\ and\ \citenamefont
  {M\o{}lmer}(2000)}]{PhysRevA.62.022311}%
  \BibitemOpen
  \bibfield  {author} {\bibinfo {author} {\bibfnamefont {A.}~\bibnamefont
  {S\o{}rensen}}\ and\ \bibinfo {author} {\bibfnamefont {K.}~\bibnamefont
  {M\o{}lmer}},\ }\bibfield  {title} {\bibinfo {title} {Entanglement and
  quantum computation with ions in thermal motion},\ }\href
  {https://doi.org/10.1103/PhysRevA.62.022311} {\bibfield  {journal} {\bibinfo
  {journal} {Phys. Rev. A}\ }\textbf {\bibinfo {volume} {62}},\ \bibinfo
  {pages} {022311} (\bibinfo {year} {2000})}\BibitemShut {NoStop}%
\bibitem [{\citenamefont {Lee}\ \emph {et~al.}(2005)\citenamefont {Lee} \emph
  {et~al.}}]{Lee_2005}%
  \BibitemOpen
  \bibfield  {author} {\bibinfo {author} {\bibfnamefont {P.~J.}\ \bibnamefont
  {Lee}} \emph {et~al.},\ }\bibfield  {title} {\bibinfo {title} {Phase control
  of trapped ion quantum gates},\ }\href
  {https://doi.org/10.1088/1464-4266/7/10/025} {\bibfield  {journal} {\bibinfo
  {journal} {Journal of Optics B: Quantum and Semiclassical Optics}\ }\textbf
  {\bibinfo {volume} {7}},\ \bibinfo {pages} {S371} (\bibinfo {year}
  {2005})}\BibitemShut {NoStop}%
\bibitem [{\citenamefont {Baldwin}\ \emph {et~al.}(2020)\citenamefont
  {Baldwin}, \citenamefont {Bjork}, \citenamefont {Gaebler}, \citenamefont
  {Hayes},\ and\ \citenamefont {Stack}}]{PhysRevResearch.2.013317}%
  \BibitemOpen
  \bibfield  {author} {\bibinfo {author} {\bibfnamefont {C.~H.}\ \bibnamefont
  {Baldwin}}, \bibinfo {author} {\bibfnamefont {B.~J.}\ \bibnamefont {Bjork}},
  \bibinfo {author} {\bibfnamefont {J.~P.}\ \bibnamefont {Gaebler}}, \bibinfo
  {author} {\bibfnamefont {D.}~\bibnamefont {Hayes}},\ and\ \bibinfo {author}
  {\bibfnamefont {D.}~\bibnamefont {Stack}},\ }\bibfield  {title} {\bibinfo
  {title} {Subspace benchmarking high-fidelity entangling operations with
  trapped ions},\ }\href {https://doi.org/10.1103/PhysRevResearch.2.013317}
  {\bibfield  {journal} {\bibinfo  {journal} {Phys. Rev. Res.}\ }\textbf
  {\bibinfo {volume} {2}},\ \bibinfo {pages} {013317} (\bibinfo {year}
  {2020})}\BibitemShut {NoStop}%
\bibitem [{\citenamefont {Magesan}\ \emph
  {et~al.}(2011{\natexlab{b}})\citenamefont {Magesan}, \citenamefont
  {Gambetta},\ and\ \citenamefont {Emerson}}]{PhysRevLett.106.180504}%
  \BibitemOpen
  \bibfield  {author} {\bibinfo {author} {\bibfnamefont {E.}~\bibnamefont
  {Magesan}}, \bibinfo {author} {\bibfnamefont {J.~M.}\ \bibnamefont
  {Gambetta}},\ and\ \bibinfo {author} {\bibfnamefont {J.}~\bibnamefont
  {Emerson}},\ }\bibfield  {title} {\bibinfo {title} {Scalable and robust
  randomized benchmarking of quantum processes},\ }\href
  {https://doi.org/10.1103/PhysRevLett.106.180504} {\bibfield  {journal}
  {\bibinfo  {journal} {Phys. Rev. Lett.}\ }\textbf {\bibinfo {volume} {106}},\
  \bibinfo {pages} {180504} (\bibinfo {year} {2011}{\natexlab{b}})}\BibitemShut
  {NoStop}%
\bibitem [{\citenamefont {Harper}\ \emph {et~al.}(2019)\citenamefont {Harper},
  \citenamefont {Hincks}, \citenamefont {Ferrie}, \citenamefont {Flammia},\
  and\ \citenamefont {Wallman}}]{PhysRevA.99.052350}%
  \BibitemOpen
  \bibfield  {author} {\bibinfo {author} {\bibfnamefont {R.}~\bibnamefont
  {Harper}}, \bibinfo {author} {\bibfnamefont {I.}~\bibnamefont {Hincks}},
  \bibinfo {author} {\bibfnamefont {C.}~\bibnamefont {Ferrie}}, \bibinfo
  {author} {\bibfnamefont {S.~T.}\ \bibnamefont {Flammia}},\ and\ \bibinfo
  {author} {\bibfnamefont {J.~J.}\ \bibnamefont {Wallman}},\ }\bibfield
  {title} {\bibinfo {title} {Statistical analysis of randomized benchmarking},\
  }\href {https://doi.org/10.1103/PhysRevA.99.052350} {\bibfield  {journal}
  {\bibinfo  {journal} {Phys. Rev. A}\ }\textbf {\bibinfo {volume} {99}},\
  \bibinfo {pages} {052350} (\bibinfo {year} {2019})}\BibitemShut {NoStop}%
\bibitem [{\citenamefont {Nielsen}(2002{\natexlab{a}})}]{NIELSEN2002249}%
  \BibitemOpen
  \bibfield  {author} {\bibinfo {author} {\bibfnamefont {M.~A.}\ \bibnamefont
  {Nielsen}},\ }\bibfield  {title} {\bibinfo {title} {A simple formula for the
  average gate fidelity of a quantum dynamical operation},\ }\href
  {https://doi.org/https://doi.org/10.1016/S0375-9601(02)01272-0} {\bibfield
  {journal} {\bibinfo  {journal} {Physics Letters A}\ }\textbf {\bibinfo
  {volume} {303}},\ \bibinfo {pages} {249} (\bibinfo {year}
  {2002}{\natexlab{a}})}\BibitemShut {NoStop}%
\bibitem [{\citenamefont {Hanneke}\ \emph {et~al.}(2010)\citenamefont {Hanneke}
  \emph {et~al.}}]{Hanneke2010}%
  \BibitemOpen
  \bibfield  {author} {\bibinfo {author} {\bibfnamefont {D.}~\bibnamefont
  {Hanneke}} \emph {et~al.},\ }\bibfield  {title} {\bibinfo {title}
  {Realization of a programmable two-qubit quantum processor},\ }\href
  {https://doi.org/10.1038/nphys1453} {\bibfield  {journal} {\bibinfo
  {journal} {Nature Physics}\ }\textbf {\bibinfo {volume} {6}},\ \bibinfo
  {pages} {13} (\bibinfo {year} {2010})}\BibitemShut {NoStop}%
\bibitem [{\citenamefont {Quantinuum\mbox{\:}H1-1}(2024)}]{arb_rb_expt}%
  \BibitemOpen
  \bibfield  {author} {\bibinfo {author} {\bibnamefont
  {Quantinuum\mbox{\:}H1-1}},\ }\href {https://www.quantinuum.com/} {}
  (\bibinfo {year} {August 26, 2024})\BibitemShut {NoStop}%
\bibitem [{\citenamefont {Bergholm}\ \emph {et~al.}(2005)\citenamefont
  {Bergholm}, \citenamefont {Vartiainen}, \citenamefont {M\"ott\"onen},\ and\
  \citenamefont {Salomaa}}]{BVMS05-circuit}%
  \BibitemOpen
  \bibfield  {author} {\bibinfo {author} {\bibfnamefont {V.}~\bibnamefont
  {Bergholm}}, \bibinfo {author} {\bibfnamefont {J.~J.}\ \bibnamefont
  {Vartiainen}}, \bibinfo {author} {\bibfnamefont {M.}~\bibnamefont
  {M\"ott\"onen}},\ and\ \bibinfo {author} {\bibfnamefont {M.~M.}\ \bibnamefont
  {Salomaa}},\ }\bibfield  {title} {\bibinfo {title} {Quantum circuits with
  uniformly controlled one-qubit gates},\ }\href
  {https://doi.org/10.1103/PhysRevA.71.052330} {\bibfield  {journal} {\bibinfo
  {journal} {Phys. Rev. A}\ }\textbf {\bibinfo {volume} {71}},\ \bibinfo
  {pages} {052330} (\bibinfo {year} {2005})}\BibitemShut {NoStop}%
\bibitem [{\citenamefont {Plesch}\ and\ \citenamefont
  {Brukner}(2011)}]{PB11-prep}%
  \BibitemOpen
  \bibfield  {author} {\bibinfo {author} {\bibfnamefont {M.}~\bibnamefont
  {Plesch}}\ and\ \bibinfo {author} {\bibfnamefont {{\v C}.}~\bibnamefont
  {Brukner}},\ }\bibfield  {title} {\bibinfo {title} {Quantum-state preparation
  with universal gate decompositions},\ }\href
  {https://doi.org/10.1103/PhysRevA.83.032302} {\bibfield  {journal} {\bibinfo
  {journal} {Phys. Rev. A}\ }\textbf {\bibinfo {volume} {83}},\ \bibinfo
  {pages} {032302} (\bibinfo {year} {2011})}\BibitemShut {NoStop}%
\bibitem [{\citenamefont {Nielsen}(2002{\natexlab{b}})}]{Nielsen2002}%
  \BibitemOpen
  \bibfield  {author} {\bibinfo {author} {\bibfnamefont {M.~A.}\ \bibnamefont
  {Nielsen}},\ }\bibfield  {title} {\bibinfo {title} {A simple formula for the
  average gate fidelity of a quantum dynamical operation},\ }\href
  {https://doi.org/https://doi.org/10.1016/S0375-9601(02)01272-0} {\bibfield
  {journal} {\bibinfo  {journal} {Physics Letters A}\ }\textbf {\bibinfo
  {volume} {303}},\ \bibinfo {pages} {249} (\bibinfo {year}
  {2002}{\natexlab{b}})}\BibitemShut {NoStop}%
\bibitem [{\citenamefont {Byrd}\ \emph {et~al.}(1995)\citenamefont {Byrd},
  \citenamefont {Lu}, \citenamefont {Nocedal},\ and\ \citenamefont
  {Zhu}}]{BLNZ95-l-bfgs-b}%
  \BibitemOpen
  \bibfield  {author} {\bibinfo {author} {\bibfnamefont {R.~H.}\ \bibnamefont
  {Byrd}}, \bibinfo {author} {\bibfnamefont {P.}~\bibnamefont {Lu}}, \bibinfo
  {author} {\bibfnamefont {J.}~\bibnamefont {Nocedal}},\ and\ \bibinfo {author}
  {\bibfnamefont {C.}~\bibnamefont {Zhu}},\ }\bibfield  {title} {\bibinfo
  {title} {A limited memory algorithm for bound constrained optimization},\
  }\href {https://doi.org/10.1137/0916069} {\bibfield  {journal} {\bibinfo
  {journal} {SIAM Journal on Scientific Computing}\ }\textbf {\bibinfo {volume}
  {16}},\ \bibinfo {pages} {1190} (\bibinfo {year} {1995})}\BibitemShut
  {NoStop}%
\bibitem [{\citenamefont {Zhu}\ \emph {et~al.}(1997)\citenamefont {Zhu},
  \citenamefont {Byrd}, \citenamefont {Lu},\ and\ \citenamefont
  {Nocedal}}]{ZBLN97-l-bfgs-b}%
  \BibitemOpen
  \bibfield  {author} {\bibinfo {author} {\bibfnamefont {C.}~\bibnamefont
  {Zhu}}, \bibinfo {author} {\bibfnamefont {R.~H.}\ \bibnamefont {Byrd}},
  \bibinfo {author} {\bibfnamefont {P.}~\bibnamefont {Lu}},\ and\ \bibinfo
  {author} {\bibfnamefont {J.}~\bibnamefont {Nocedal}},\ }\bibfield  {title}
  {\bibinfo {title} {Algorithm 778: L-bfgs-b: Fortran subroutines for
  large-scale bound-constrained optimization},\ }\href
  {https://doi.org/10.1145/279232.279236} {\bibfield  {journal} {\bibinfo
  {journal} {ACM Trans. Math. Softw.}\ }\textbf {\bibinfo {volume} {23}},\
  \bibinfo {pages} {550} (\bibinfo {year} {1997})}\BibitemShut {NoStop}%
\bibitem [{\citenamefont {Virtanen}\ \emph {et~al.}(2020)\citenamefont
  {Virtanen} \emph {et~al.}}]{SciPy}%
  \BibitemOpen
  \bibfield  {author} {\bibinfo {author} {\bibfnamefont {P.}~\bibnamefont
  {Virtanen}} \emph {et~al.},\ }\bibfield  {title} {\bibinfo {title} {{{SciPy}
  1.0: Fundamental Algorithms for Scientific Computing in Python}},\ }\href
  {https://doi.org/10.1038/s41592-019-0686-2} {\bibfield  {journal} {\bibinfo
  {journal} {Nature Methods}\ }\textbf {\bibinfo {volume} {17}},\ \bibinfo
  {pages} {261} (\bibinfo {year} {2020})}\BibitemShut {NoStop}%
\bibitem [{\citenamefont {Harris}\ \emph {et~al.}(2020)\citenamefont {Harris}
  \emph {et~al.}}]{NumPy}%
  \BibitemOpen
  \bibfield  {author} {\bibinfo {author} {\bibfnamefont {C.~R.}\ \bibnamefont
  {Harris}} \emph {et~al.},\ }\bibfield  {title} {\bibinfo {title} {Array
  programming with {NumPy}},\ }\href
  {https://doi.org/10.1038/s41586-020-2649-2} {\bibfield  {journal} {\bibinfo
  {journal} {Nature}\ }\textbf {\bibinfo {volume} {585}},\ \bibinfo {pages}
  {357} (\bibinfo {year} {2020})}\BibitemShut {NoStop}%
\bibitem [{\citenamefont {O'Neill}(2014)}]{oneill:pcg2014}%
  \BibitemOpen
  \bibfield  {author} {\bibinfo {author} {\bibfnamefont {M.~E.}\ \bibnamefont
  {O'Neill}},\ }\href@noop {} {\emph {\bibinfo {title} {PCG: A Family of Simple
  Fast Space-Efficient Statistically Good Algorithms for Random Number
  Generation}}},\ \bibinfo {type} {Tech. Rep.}\ \bibinfo {number}
  {HMC-CS-2014-0905}\ (\bibinfo  {institution} {Harvey Mudd College},\ \bibinfo
  {address} {Claremont, CA},\ \bibinfo {year} {2014})\BibitemShut {NoStop}%
\bibitem [{\citenamefont {Bradbury}\ \emph {et~al.}(2018)\citenamefont
  {Bradbury} \emph {et~al.}}]{jax2018github}%
  \BibitemOpen
  \bibfield  {author} {\bibinfo {author} {\bibfnamefont {J.}~\bibnamefont
  {Bradbury}} \emph {et~al.},\ }\href {http://github.com/jax-ml/jax} {\bibinfo
  {title} {{JAX}: composable transformations of {P}ython+{N}um{P}y programs}}
  (\bibinfo {year} {2018})\BibitemShut {NoStop}%
\bibitem [{\citenamefont {Boixo}\ \emph {et~al.}(2018)\citenamefont {Boixo}
  \emph {et~al.}}]{Boixo2018}%
  \BibitemOpen
  \bibfield  {author} {\bibinfo {author} {\bibfnamefont {S.}~\bibnamefont
  {Boixo}} \emph {et~al.},\ }\bibfield  {title} {\bibinfo {title}
  {Characterizing quantum supremacy in near-term devices},\ }\href
  {https://doi.org/10.1038/s41567-018-0124-x} {\bibfield  {journal} {\bibinfo
  {journal} {Nature Physics}\ }\textbf {\bibinfo {volume} {14}},\ \bibinfo
  {pages} {595} (\bibinfo {year} {2018})}\BibitemShut {NoStop}%
\bibitem [{\citenamefont {Efron}\ and\ \citenamefont
  {Tibshirani}(1993)}]{EfroTibs1993}%
  \BibitemOpen
  \bibfield  {author} {\bibinfo {author} {\bibfnamefont {B.}~\bibnamefont
  {Efron}}\ and\ \bibinfo {author} {\bibfnamefont {R.~J.}\ \bibnamefont
  {Tibshirani}},\ }\href@noop {} {\emph {\bibinfo {title} {An Introduction to
  the Bootstrap}}},\ \bibinfo {series} {Monographs on Statistics and Applied
  Probability}\ No.~\bibinfo {number} {57}\ (\bibinfo  {publisher} {Chapman \&
  Hall/CRC},\ \bibinfo {address} {Boca Raton, Florida, USA},\ \bibinfo {year}
  {1993})\BibitemShut {NoStop}%
\bibitem [{\citenamefont {Dehaene}\ and\ \citenamefont
  {De~Moor}(2003)}]{DDM03-quadratic}%
  \BibitemOpen
  \bibfield  {author} {\bibinfo {author} {\bibfnamefont {J.}~\bibnamefont
  {Dehaene}}\ and\ \bibinfo {author} {\bibfnamefont {B.}~\bibnamefont
  {De~Moor}},\ }\bibfield  {title} {\bibinfo {title} {{C}lifford group,
  stabilizer states, and linear and quadratic operations over {GF(2)}},\ }\href
  {https://doi.org/10.1103/PhysRevA.68.042318} {\bibfield  {journal} {\bibinfo
  {journal} {Phys. Rev. A}\ }\textbf {\bibinfo {volume} {68}},\ \bibinfo
  {pages} {042318} (\bibinfo {year} {2003})}\BibitemShut {NoStop}%
\bibitem [{\citenamefont {Gross}\ and\ \citenamefont {Van~den
  Nest}(2008)}]{GV08-LU-LC}%
  \BibitemOpen
  \bibfield  {author} {\bibinfo {author} {\bibfnamefont {D.}~\bibnamefont
  {Gross}}\ and\ \bibinfo {author} {\bibfnamefont {M.}~\bibnamefont {Van~den
  Nest}},\ }\bibfield  {title} {\bibinfo {title} {The {LU-LC} conjecture,
  diagonal local operations and quadratic forms over {GF(2)}},\ }\href
  {https://doi.org/10.26421/QIC8.3-4-3} {\bibfield  {journal} {\bibinfo
  {journal} {Quantum Info. Comput.}\ }\textbf {\bibinfo {volume} {8}},\
  \bibinfo {pages} {263} (\bibinfo {year} {2008})}\BibitemShut {NoStop}%
\bibitem [{\citenamefont {Van~den Nes}(2010)}]{VDN10-gk}%
  \BibitemOpen
  \bibfield  {author} {\bibinfo {author} {\bibfnamefont {M.}~\bibnamefont
  {Van~den Nes}},\ }\bibfield  {title} {\bibinfo {title} {Classical simulation
  of quantum computation, the {G}ottesman-{K}nill theorem, and slightly
  beyond},\ }\href {https://doi.org/10.26421/QIC10.3-4-6} {\bibfield  {journal}
  {\bibinfo  {journal} {Quantum Info. Comput.}\ }\textbf {\bibinfo {volume}
  {10}},\ \bibinfo {pages} {258} (\bibinfo {year} {2010})}\BibitemShut
  {NoStop}%
\bibitem [{\citenamefont {Bravyi}\ and\ \citenamefont
  {Gosset}(2016)}]{BG16-clifford}%
  \BibitemOpen
  \bibfield  {author} {\bibinfo {author} {\bibfnamefont {S.}~\bibnamefont
  {Bravyi}}\ and\ \bibinfo {author} {\bibfnamefont {D.}~\bibnamefont
  {Gosset}},\ }\bibfield  {title} {\bibinfo {title} {Improved classical
  simulation of quantum circuits dominated by {C}lifford gates},\ }\href
  {https://doi.org/10.1103/PhysRevLett.116.250501} {\bibfield  {journal}
  {\bibinfo  {journal} {Phys. Rev. Lett.}\ }\textbf {\bibinfo {volume} {116}},\
  \bibinfo {pages} {250501} (\bibinfo {year} {2016})}\BibitemShut {NoStop}%
\end{thebibliography}%

\end{document}